%% file: Full-version-MAIN.tex
\pdfoutput=1
\documentclass[a4paper,UKenglish,cleveref, autoref,thm-restate]{lipics-v2021}
\input{preamble.tex}
\hideLIPIcs
\nolinenumbers %

\title{Parameterized Algorithms for Balanced Cluster Edge Modification Problems} %

\titlerunning{Parameterized Algorithms for Balanced Cluster Edge Modification Problems} %

\author{Jayakrishnan Madathil}{University of Glasgow, UK}{jayakrishnan.madathil@glasgow.ac.uk}{}{Supported by the Engineering and Physical Sciences Research Council [EP/V032305/1].}

\author{Kitty Meeks}{University of Glasgow, UK}{kitty.meeks@glasgow.ac.uk}{}{Supported by the Engineering and Physical Sciences Research Council [EP/V032305/1].}

\authorrunning{J. Madathil and K. Meeks} %

\Copyright{Jayakrishnan Madathil and Kitty Meeks} %



\category{} %

\relatedversion{} %

\acknowledgements{We thank Jessica Enright for helpful discussions on the problems studied in this paper.}%

\EventEditors{John Q. Open and Joan R. Access}
\EventNoEds{2}
\EventLongTitle{42nd Conference on Very Important Topics (CVIT 2016)}
\EventShortTitle{CVIT 2016}
\EventAcronym{CVIT}
\EventYear{2016}
\EventDate{December 24--27, 2016}
\EventLocation{Little Whinging, United Kingdom}
\EventLogo{}
\SeriesVolume{42}
\ArticleNo{23}

\begin{document}

\maketitle

\begin{abstract}
We study {\sc Cluster Edge Modification} problems with constraints on the size of the clusters. A graph $G$ is a cluster graph if every connected component of $G$ is a clique. In a typical {\sc Cluster Edge Modification} problem such as the widely studied {\sc Cluster Editing}, we are given a graph $G$ and a non-negative integer $k$ as input, and we have to decide if we can turn $G$ into a cluster graph by way of at most $k$ edge modifications---that is, by adding or deleting edges. In this paper, we study the parameterized complexity of such problems, but with an additional constraint: The size difference between any two connected components of the resulting cluster graph should not exceed a given threshold. Depending on which modifications are permissible---only adding edges, only deleting edges, both adding and deleting edges---we have three different computational problems. We show that all three problems, when parameterized by $k$, admit single-exponential time \FPT\ algorithms and polynomial kernels. Our problems may be thought of as the size-constrained or balanced counterparts of the typical {\sc Cluster Edge Modification} problems, similar to the well-studied size-constrained or balanced counterparts of other clustering problems such as {\sc $k$-Means Clustering}. 
\end{abstract}

\input{1-INTRO.tex}
\input{Full-version-prelims}

\input{Full-version-kernel}
\input{Full-version-FPT-section-intro}

\input{Full-version-FPT-branching.tex}
\input{Full-version-FPT-decision-deletion}

\input{Full-version-FPT-decision-completion-and-editing}

\input{counting.tex}

\section{Conclusion}
\label{sec:conclusion}
We studied variants of the widely studied {\sc Cluster Editing} problem, and designed single-exponential time \FPT\ algorithms and polynomial kernels for them. We also designed \FPT\ algorithms for the counting versions of the problems. These results add to a growing body of literature on algorithms for clustering under size or balance constraints. Our work triggers several questions for future research. First, while our kernel for the completion version has $\cO(k)$ vertices, our kernels for the deletion and editing versions have $\cO(k^4)$ and $\cO(k^3)$ vertices, respectively. Can we design $\cO(k)$ kernels for these problems as well? Steinvik~\cite{steinvik2020kernelization} also raised the same question. Second, it is known that {\sc Cluster Editing} and {\sc Cluster Deletion} do not admit sub-exponential time \FPT\ algorithms unless the Exponential Time Hypothesis (ETH) fails~\cite{DBLP:journals/dam/KomusiewiczU12}. This lower bound applies to \bcd\ and \bcc\ as well. But the case of \bcc\ is open: Can we design a sub-exponential time \FPT\ algorithm for \bcc\ or prove that such an algorithm does not exist unless ETH fails? %
Third, can we extend our results to similar edge modification problems under balance constraints, for example, the bipartite counterpart of {\sc Cluster Editing}, called {\sc Bicluster Editing}~\cite{DBLP:conf/tamc/GuoHKZ08,DBLP:journals/ipl/XiaoK22} and   generalizations of {\sc Cluster Editing} such as {\sc $s$-Plex Cluster Editing}~\cite{DBLP:journals/siamdm/GuoKNU10}?  Fourth, apart from these theoretical questions, it would be interesting to evaluate the practical competitiveness of our algorithms, particularly the efficacy of the reduction rules that we use in our kernelization algorithms on real-world  or synthetic instances. 
\bibliography{refs}
\newpage
\appendix
\input{nph_bcc.tex}
\end{document}

%% file: preamble.tex
\usepackage[usenames,svgnames,table]{xcolor}
\definecolor{myblue}{rgb}{0.1,0.1,0.5}
\definecolor{mygreen}{rgb}{0.00,0.26,0.15}

\usepackage{xspace}
\usepackage{framed}
\usepackage{xifthen}
\usepackage{tikz}
\usetikzlibrary{calc}
\usepackage{amsmath}
\usepackage{amsthm}
\usepackage{amssymb}
\usepackage{nicefrac}
\usepackage{algorithmic}
\usepackage{algorithm}
\usepackage{multirow}
\usepackage{lineno}
\usepackage{todonotes}

\usepackage[multiple]{footmisc}

\definecolor{shadecolor}{gray}{0.75}

\usepackage{comment}

\usepackage{bm}
\usepackage{apxproof}

\newlength{\RoundedBoxWidth}
\newsavebox{\GrayRoundedBox}
\newenvironment{GrayBox}[1]%
   {\setlength{\RoundedBoxWidth}{.93\textwidth}
    \def\boxheading{#1}
    \begin{lrbox}{\GrayRoundedBox}
       \begin{minipage}{\RoundedBoxWidth}}%
   {   \end{minipage}
    \end{lrbox}
    \begin{center}
    \begin{tikzpicture}%
       \node(Text)[draw=black!20,fill=white,rounded corners,%
             inner sep=2ex,text width=\RoundedBoxWidth]%
             {\usebox{\GrayRoundedBox}};
        \coordinate(x) at (current bounding box.north west);
        \node [draw=white,rectangle,inner sep=3pt,anchor=north west,fill=white] 
        at ($(x)+(6pt,.75em)$) {\boxheading};
    \end{tikzpicture}
    \end{center}}     
\newenvironment{defproblemx}[2][]{\noindent\ignorespaces%
                                \FrameSep=6pt%
                                \parindent=0pt%
                \vspace*{-1.5em}
                \ifthenelse{\isempty{#1}}{%
                  \begin{GrayBox}{\textsc{#2}}%
                }{%
                  \begin{GrayBox}{\textsc{#2}  parameterized by~{#1}}%
                }
                \begin{tabular*}{\textwidth}{@{\hspace{.1em}} >{\itshape} p{1.8cm} p{0.8\textwidth} @{}}%
            }{
                \end{tabular*}%
                \end{GrayBox}%
                \ignorespacesafterend
            }  

\newcommand{\defproblem}[3]{%
  \begin{defproblemx}{#1}
    Input:  & #2 \\
    Task: & #3
  \end{defproblemx}
}%

\newcommand{\FPT}{\textsf{FPT}}

\newcommand{\NPH}{\textsf{NP}-hard}

\newcommand{\card}[1]{\ensuremath{{\vert {#1} \vert }}}
\newcommand{\ca}[1]{\ensuremath{\mathcal{#1} }}
\newcommand{\set}[1]{\ensuremath{\left\{ {#1} \right\}}}
\newcommand{\fn}[3]{\ensuremath{{{#1} : {#2} \rightarrow {#3}}}}

\newcommand{\cO}{\mathcal{O}}

\newcommand{\ceil}[1]{\ensuremath{\left\lceil {#1} \right\rceil}}

\newcommand{\Lb}[1]{\ensuremath{\left({#1}\right)}}

\newcommand{\floor}[1]{\ensuremath{\left\lfloor {#1} \right\rfloor}}

\newtheorem{rrule}{Reduction Rule}

\newtheoremrep{lemmaapx}[lemma]{Lemma}

\newcommand{\poly}{\ensuremath{\text{poly}}}

\mathchardef\mh="2D

\newcommand{\mul}{\ensuremath{\mathtt{mul}}}

\newcommand{\lcomp}{\ensuremath{\Lambda}} %
\newcommand{\scomp}{\ensuremath{\lambda}} %

\newcommand{\spp}{\ensuremath{\mathtt{spp}}} %

\newcommand{\US}{\ensuremath{\mathtt{US}}}%

\newcommand{\cost}{\ensuremath{\mathtt{cost}}} %

\newcommand{\vol}{\ensuremath{\mathtt{vol}}} %

\newcommand{\CS}{\ensuremath{\mathtt{CS}}} %

\newcommand{\CC}{\ensuremath{\mathtt{CC}}} %

\newcommand{\wit}{\ensuremath{\mathtt{wit}}} %

\newcommand{\ndel}{\ensuremath{\#\mathtt{del}}} %

\newcommand{\ncom}{\ensuremath{\#\mathtt{com}}}

\newcommand{\tsum}{\ensuremath{\mathtt{sum}}} %

\newcommand{\binbfull}{{\sc Balls in Bins}}
\newcommand{\binb}{{\sc B-in-B}}

\newcommand{\bccfull}{{\sc Balanced Cluster Completion}}
\newcommand{\bcc}{{\sc BCC}}

\newcommand{\countbccfull}{{\sc \#Balanced Cluster Completion}}
\newcommand{\countbcc}{{\sc \#BCC}}

\newcommand{\countbcd}{{\sc \#BCD}}

\newcommand{\nsol}{{\# \mathtt{sol}}}

\newcommand{\countbce}{{\sc \#BCE}}

\newcommand{\crbccfull}{{\sc Colour-Restricted Balanced Cluster Completion}}
\newcommand{\crbcc}{{\sc CRBCC}}

\newcommand{\countcrbcc}{{\sc \#CRBCC}}

\newcommand{\cc}{$c$-closed}

\newcommand{\bcdfull}{{\sc Balanced Cluster Deletion}}
\newcommand{\bcd}{{\sc BCD}}

\newcommand{\bcefull}{{\sc Balanced Cluster Editing}}
\newcommand{\bce}{{\sc BCE}}

\newcommand{\cccdfull}{{\sc Cardinality-Constrained Cluster Deletion}}
\newcommand{\cccd}{{\sc CCCD}}

\newcommand{\algobcc}{{\textsf{Algo-BCC}}} %
\newcommand{\algobccruntime}{\ensuremath{2^{\cO(k)} n^{\cO(1)}}} %

\newcommand{\countbceruntime}{\ensuremath{2^{\cO(k \log k)} n^{\cO(1)}}} %

\newcommand{\colourruntime}{\ensuremath{2^{\cO(k)} (q + 1)^{\cO(k)} n^{\cO(1)}}}

\newcommand{\bralgobcc}{{\textsf{Branch-Algo-BCC}}}

\newcommand{\algobcec}{{\textsf{Algo-BCE-C}}}

\newcommand{\algocccd}{{\textsf{Algo-CCCD}}}%
\newcommand{\algobcdruntime}{\ensuremath{2^{k} n^{\cO(1)}}} %

\newcommand{\fastalgobcec}{{\textsf{Fast-Algo-BCE-C}}}

\newcommand{\algobinbruntime}{\ensuremath{\cO(2^{s} \cdot s^3 \cdot t \cdot W^2)}} %

\newcommand{\fastalgobcc}{{\textsf{Fast-Algo-BCC}}} %
\newcommand{\fastbccruntime}{\ensuremath{4^{k + o(k)} n^{\cO(1)}}} %

\newcommand{\annocmfull}{{\sc Annotated  Cluster Modification}}
\newcommand{\annocm}{{\sc Anno-CM}}

%% file: 1-INTRO.tex
\section{Introduction}
\label{sec:intro}

The term clustering refers to a broad range of methods that aim to group data points into subgroups---or clusters, as they are often called---so that points within a cluster are ``close'' or ``similar'' to one another and points in different clusters are ``distant'' or ``dissimilar.'' 
See the surveys by Ezugwu et al.~\cite{DBLP:journals/eaai/EzugwuIOAAEA22} and by Xu and Tian~\cite{xu2015comprehensive} for an overview of various clustering models and algorithms. 
Applications of clustering in domains such as circuit design~\cite{DBLP:journals/tcad/HagenK92}, content-based image searching~\cite{althoff2011balanced}, load-balancing in wireless networks~\cite{liao2012load,shang2010energy} etc. require clusters of reasonable sizes. %
And while popular clustering methods such as {\sc $k$-means} explicitly prescribe the number of clusters (the $k$ in {\sc $k$-means} stands for grouping into $k$ clusters), there is no guarantee for the sizes of the clusters, which may result in clusters with too few data points. This prompted the study of ``size-constrained'' clustering that seeks to ensure size guarantees for the clusters, and in particular,  ``balanced'' clustering that requires all clusters to be of roughly the same size; see, for example,~\cite{bradley2000constrained,DBLP:conf/ijcai/LinHX19,DBLP:conf/sspr/MalinenF14} for results on size-constrained or balanced variants of {\sc $k$-Means} clustering. In this paper, we focus on graph-based clustering  with balance constraints. %

{\sc Cluster Editing}, also called {\sc Correlation Clustering}, is the canonical  graph-based clustering problem. %
A cluster graph is one in which every connected component is a clique, and {\sc Cluster Editing} asks if a given graph can be turned into a cluster graph by adding and deleting at most $k$ edges. The problem and its several variants have been studied extensively, resulting in a large volume of literature, including approximation algorithms~\cite{DBLP:journals/jacm/AilonCN08,DBLP:journals/mp/AprileDFH23,DBLP:journals/ml/BansalBC04,DBLP:journals/dam/ShamirST04}, parameterized algorithms~\cite{DBLP:journals/jda/Bocker12,DBLP:conf/cie/BockerB13,DBLP:conf/iwpec/CaoC10,DBLP:journals/jcss/ChenM12,DBLP:journals/jcss/FominKPPV14,DBLP:journals/mst/Tsur21} and heuristics~\cite{DBLP:journals/jcb/Ben-DorSY99,rahmann2007exact,wittkop2007large}. Notice that the problem imposes no restriction on the number of clusters or the size of each cluster. While variants of {\sc Cluster Editing} with constraints on the number of clusters have previously been studied~\cite{DBLP:journals/jcss/FominKPPV14,DBLP:journals/tcs/Guo09,DBLP:journals/dam/ShamirST04}, variants involving size or balance constraints have not received much attention, barring a couple of exceptions~\cite{DBLP:journals/jda/Abu-Khzam17,steinvik2020kernelization}. In particular, Steinvik~\cite{steinvik2020kernelization} studied variants of {\sc Cluster Modification} problems with balance constraints that require the size difference between the components of the cluster graph to be bounded. %
We too study {\sc Cluster Edge Modification} with a balance constraint. In particular, we define appropriate problems in this direction and design efficient parameterized algorithms and kernels (parameterized by the number of modified edges). %

\subsubsection*{Our Contribution}
To formally define our problems, we introduce the following definition. %
For a non-negative integer $\eta$, we say that a graph $G$ is \emph{$\eta$-balanced} if the (additive) size difference between any two connected components of $G$ is at most $\eta$; by the size of a connected component, we mean the number of vertices in that component. Notice that the definition is equivalent to saying that the size difference between a largest connected component and a smallest connected component is at most $\eta$.  %
Notice also that the definition of $\eta$-balanced graphs does not place any restriction on the number of components. For example, a connected graph is $\eta$-balanced for every $\eta \geq 0$; so is an edge-less graph consisting of only isolated vertices.

We study problems of the following type. Given a graph $G$ and non-negative integers $k$ and $\eta$, is it possible to turn $G$ into an $\eta$-balanced  cluster graph by way of at most $k$ edge modifications? We have three different computational problems depending on which modifications are permissible---only edge additions, only edge deletions, both edge additions and deletions. We formally define the problem corresponding to edge addition as follows.

\defproblem{\bccfull\ (\bcc)}{An $n$-vertex graph $G$ and non-negative integers $k$ and $\eta$.}{Decide if it is possible to add at most $k$ edges to $G$ so that the resulting graph is an $\eta$-balanced cluster graph.}
We define the problems \bcdfull\ (\bcd), in which only edge deletions are permitted, and \bcefull\ (\bce), in which both edge deletions and additions are permitted, analogously. Steinvik, in his master's thesis~\cite{steinvik2020kernelization}, considered a closely related variant of each of these three problems, where $\eta$ is fixed as part of the problem definition rather than given as part of the input, and showed that the problems admit kernels of size $\poly(k) + \eta$. We discuss this work in more detail later. 

We often refer to our three problems, i.e., \bcc, \bcd\ and \bce, as the completion, deletion or editing versions, respectively, and we refer to the requirement that the resulting cluster graph be $\eta$-balanced as the balance constraint. 
We also study the natural counting variant of each of these three problems, in which we are interested in computing the number of solutions rather than whether or not a solution exists. For example, in the counting variant of \bcc, we are given $G, k$ and $\eta$ as input, and we need to output the number of sets $F$ of at most $k$ edges such that adding these edges to $G$ would result in an $\eta$-balanced cluster graph. The counting variants of \bcd\ and \bce\ are defined analogously. 

Notice that each of the three problems---\bcc, \bcd\ and \bce---generalises its respective counterpart without the balance constraint; we can simply take $\eta = n$, where $n$ is the number of vertices in the input graph, as any $n$-vertex graph is $n$-balanced. As a consequence, we can immediately conclude that  \bcd\ and \bce\ (i.e., the deletion and editing versions) are \NPH; these follow  respectively from the \NPH ness of {\sc Cluster Deletion} and {\sc Cluster Editing} (i.e., the deletion and editing versions without the balance constraint)~\cite{DBLP:journals/dam/ShamirST04}. The case of \bcc, however, is different, for without the balance constraint, the completion version is polynomial-time solvable: we simply need to add edges until each connected component becomes a clique. But we observe that the \NPH ness of \bcc\ follows from a hardness result due to Froese et al.~\cite{DBLP:conf/aaai/FroeseKN22};\footnote{Froese et al.~\cite{DBLP:conf/aaai/FroeseKN22} showed that the following problem, called {\sc Cluster Transformation by Edge Addition}, is \NPH. We are given a cluster graph $G$ and a non-negative integer $k$, and we have to decide if we can add exactly $k$ edges to $G$ so that the resulting graph is also a cluster graph. The reduction that shows the \NPH ness of this problem~\cite[Theorem 6]{DBLP:conf/aaai/FroeseKN22} also shows the \NPH ness of \bcc. We outline a proof of this fact in Appendix~\ref{sec:completion-nph}.}\textsuperscript{,}\footnote{In fact,   Steinvik~\cite{steinvik2020kernelization} proved a stronger result that for each $\eta \geq 0$, the corresponding $\eta$-\bcc\ problem (in which we are given $G$ and $k$ as input, and we want to decide if we can turn $G$ into an $\eta$-balanced cluster graph by adding at most $k$ edges) is \NPH.}  %
Thus, all three problems---\bcc, \bcd \ and \bce---are \NPH.

\subparagraph*{Our Results.} We show that all three problems---\bcc, \bcd\ and \bce---when parameterized by $k$, admit \FPT\ algorithms and polynomial kernels. Specifically, we prove the following results. 

\begin{enumerate}
    \item \bcc\ admits a \algobcdruntime\ time algorithm and a kernel with  $10k$ vertices. 
    \item \bcd\ admits a \algobcdruntime\ time  algorithm and a kernel with  $\cO(k^4)$ vertices.\footnote{Using arguments that are nearly identical to the ones in our kernel for \bce, it is also possible to design a kernel for \bcd\ with $\cO(k^3)$ vertices. But we choose to present a kernel with $\cO(k^4)$ vertices as it involves new and simple  arguments. In particular, our arguments yield a simple kernel with $\cO(k^4)$ vertices for {\sc Cluster Deletion} (i.e., the deletion version without the balance constraint), which might be of independent interest. But these arguments work only for the deletion version and do not extend to the editing version.} 
    \item \bce\ admits a \algobccruntime\ time algorithm and a kernel with $\cO(k^3)$ vertices.
\end{enumerate}

The \FPT\ algorithm for \bcc\ is a simple branching procedure that runs in time \algobcdruntime. We also design another algorithm for \bcc\ that runs in time \algobccruntime; although the big-Oh in the $2^{\cO(k)}$ factor hides a fairly large constant, this algorithm makes use of a  correspondence between solutions for \bcc\ and integer partitions. This correspondence turns out to be quite useful, as the same idea   extends easily to \bce. Not just that, we show that we can use the same idea (i.e., the correspondence between solutions and integer partitions) to yield the following two sets of results. %

\begin{enumerate}
\setcounter{enumi}{3}
\item \bcc\ and \bce\ admit algorithms that run in time \fastbccruntime; these algorithms rely on fast polynomial multiplication. 
\item The counting variant of \bcc\ admits an algorithm that runs in time \algobccruntime, and the counting variants of \bcd\ and \bce\ admit algorithms that run in time \countbceruntime. 
\end{enumerate}

Our results add to a long line of work on the parameterized complexity of graph modification problems (see~\cite{DBLP:journals/csr/CrespelleDFG23}), and specifically on {\sc Cluster Editing} and its variants~\cite{DBLP:journals/jda/Bocker12,DBLP:conf/cie/BockerB13,DBLP:conf/apbc/BockerBBT08,DBLP:conf/iwpec/CaoC10,DBLP:conf/iwpec/Fellows06,DBLP:conf/fct/FellowsLRS07,DBLP:journals/jcss/FominKPPV14,DBLP:journals/mst/GrammGHN05,DBLP:journals/tcs/Guo09,DBLP:journals/mst/ProttiSS09}. In particular, after a series of improvements~\cite{DBLP:journals/algorithmica/GrammGHN04,DBLP:journals/mst/GrammGHN05,DBLP:conf/apbc/BockerBBT08,DBLP:conf/cie/BockerB13}, the current fastest parameterized algorithm for {\sc Cluster Editing} runs in time $\cO(1.62^k + n + m)$, where $n$ and $m$ respectively are the number of vertices and edges in the input graph~\cite{DBLP:journals/jda/Bocker12}, and the current smallest kernel, again, after a series of improvements~\cite{DBLP:journals/mst/GrammGHN05,DBLP:conf/iwpec/Fellows06,DBLP:conf/fct/FellowsLRS07,DBLP:journals/tcs/Guo09,DBLP:journals/mst/ProttiSS09}, has $2k$ vertices~\cite{DBLP:conf/iwpec/CaoC10,DBLP:journals/jcss/ChenM12}. We must also emphasise  that parameterized algorithms  for {\sc Cluster Editing} and its variants are not only theoretically significant, but have been found to be effective in practice, particularly in the clustering of biological data~\cite{morris2011clustermaker,wittkop2010partitioning}. Also, experimental evaluations indicate that pre-processing methods inspired by kernelization algorithms for {\sc Cluster Editing} are useful in practice~\cite{DBLP:conf/cie/BockerB13,DBLP:journals/algorithmica/BockerBK11,DBLP:conf/iwpec/KellerhalsKNZ21}. 

\paragraph*{Overview of our arguments}
 Before we discuss some of our arguments, we first highlight some of the challenges that our problems pose, when compared to a typical graph modification problem such as {\sc Cluster Editing}. %
 To begin with, for any fixed $\eta \geq 0$, the class of $\eta$-balanced graphs is not hereditary; a class of graphs is hereditary if it is closed under vertex deletions.  
 For example, a graph with exactly two equal-sized components is $0$-balanced, but deleting one vertex will destroy the $0$-balanced property. This poses an immediate difficulty in  designing kernels for our problems, as we cannot delete any ``irrelevant'' vertex; for example, when designing a kernel for {\sc Cluster Editing}, we can delete any component of the input graph that is a clique (because no edge of such a component needs to be modified), but that option is not available to us. Also, unlike {\sc Cluster Editing} and many other graph modification problems, for any of the three problems that we study, we cannot enumerate all minimal solutions of size at most $k$ in \FPT\ time as the number of such solutions may not be bounded by a function of $k$; see Example~\ref{eg:minimal}.  
\begin{example}
\label{eg:minimal}
Consider \bcc. For fixed $k$, we  construct a yes-instance $(G, k, \eta)$ of \bcc\ as follows. Let $k \in \mathbb{N}$ be a perfect square. Take $\eta = 1$ and $n$ to be a sufficiently large integer so that $(n - \sqrt{k})/{\sqrt{k}}$ is an integer. We define an $n$-vertex cluster graph $G$ as follows: $G$ has exactly $\sqrt{k}$ isolated vertices and exactly $(n - \sqrt{k})/{\sqrt{k}}$ components of size $\sqrt{k}$. Notice that to turn $G$ into an $\eta$-balanced cluster graph by adding at most $k$ edges, we must either merge all the isolated vertices into a single component by adding $\binom{\sqrt{k}}{2}$ edges, or merge each isolated vertex with a component of size $\sqrt{k}$, which requires the addition of $\sqrt{k} \cdot \sqrt{k} = k$ edges. For the latter option, notice that we have $\binom{(n - \sqrt{k})/{\sqrt{k}}}{\sqrt{k}} \cdot \sqrt{k}! = n^{\Omega(\sqrt{k})}$ ways in which we can form a solution.  \end{example}

Despite the challenges noted above, we are still able to make use of arguments from the existing literature. For example, we use the familiar strategy of branching on induced paths on $3$ vertices in our \FPT\ algorithms for \bcd\ and \bce. And we use a small-sized modulator---a set of vertices whose removal will turn the input graph into a cluster graph---to design our kernel for \bce. We also crucially rely on the fact that by  modifying at most $k$ (non-)edges, we modify the adjacencies of at most $2k$ vertices. But these arguments alone are insufficient for our problems. We now briefly outline some of the other  arguments that we use. 

 \subparagraph*{BCC is a problem about numbers.} The one idea that we invoke throughout this paper is this: \bcc, at its core, is a problem about numbers. To make our meaning clear, consider an instance $(G, k, \eta)$ of \bcc. First of all, we can assume that the input graph $G$ is a cluster graph; if any component of $G$ is not a clique, then we must necessarily add all the ``missing edges'' so that the component becomes a clique. But then, as $G$ is already a cluster graph, we only need to add edges to $G$ to satisfy the balance constraint (and ensure that the resulting graph is a cluster graph). Now, consider two connected components $H$ and $H'$ of $G$; again, $H$ and $H'$ are cliques. Notice then that in any solution for \bcc, either we do not add any edges between $H$ and $H'$, or we add all edges between $H$ and $H'$. That is, in the latter case, we merge $H$ and $H'$ together by adding all possible edges between them, and this merging requires the addition of exactly $\card{H} \card{H'}$ edges, where $\card{H}$ and $\card{H'}$ respectively the number of vertices in $H$ and $H'$. So solving an instance of \bcc\ amounts to deciding which components of $G$ are merged with which other components. In fact, we need not even decide which components are involved in the  merging, and instead we only need to decide the sizes of the components involved in the merging. We apply this idea to \bcd\ and \bce\ as well, after appropriately reducing these problems to problems on cluster graphs. In the next paragraph, we outline how we use this idea to design \FPT\ algorithms. 

 \subparagraph*{Enumerating all representative solutions using integer partitions.} Notice that our \FPT\ algorithms (for the decision variant of the problems) run in single-exponential time, i.e., in time \algobccruntime.  %
  As noted above (Example~\ref{eg:minimal}), we cannot enumerate all minimal solutions of size at most $k$ in \FPT\ time, let alone in time \algobccruntime. But we argue that we can still enumerate all ``representative solutions'' of size at most $k$ in time \algobccruntime. To do this, we use the idea in the preceding paragraph that only the sizes of the components are relevant; we also use the fact that we only modify the adjacencies of at most $2k$ vertices if $(G, k, \eta)$ is a yes-instance. Use these two observations, we associate every solution with a set of integers $\ell_1, \ell_2,\ldots, \ell_t$ and a partition of $\ell_i$ for every $i$; we will have the guarantee that $\ell_1 + \ell_2 + \cdots + \ell_t \leq 2k$. Each  $\ell_i$ corresponds to a component of the $\eta$-balanced cluster graph, and  a partition of $\ell_i$ corresponds to the sizes of the components of $G$ that were merged together to form a component of size $\ell_i$.  By enumerating all possible choices for $\ell_1, \ell_2,\ldots, \ell_t$ and their   partitions, we can enumerate all representative solutions. To bound the running time, we then use a result due to Hardy and Ramanujan~\cite{hardy1918asymptotic} that says that the number of partitions of an integer $\ell$ is $2^{\cO(\sqrt{\ell})}$. 

  \subparagraph*{Faster algorithms using fast polynomial multiplication.} To design our algorithms for \bcc\ and \bce\ that run in time \fastbccruntime, we show that solving these problems amounts to solving $2^{o(k)}$ many instances of an assignment problem, where we have to assign balls of different sizes to bins of different capacities subject to a cost constraint. And we show that we can encode the solutions for this assignment problem as a polynomial so that solving the problem amounts to computing a  polynomial, for which we use the fact that we can multiply two polynomials of degree $d$ in time $\cO(d \log d)$~\cite{DBLP:conf/issac/Moenck76}.

 \subparagraph*{Kernel for BCD  based on the bound for Ramsey number of $c$-closed graphs.} For \bcd, while we can design a kernel with $\cO(k^3)$ vertices, we present a kernel with $\cO(k^4)$ vertices instead, as this kernel uses an elegant and completely new set of arguments that rely on the recently introduced class of $c$-closed graphs~\cite{DBLP:journals/siamcomp/FoxRSWW20}. For a positive integer $c$, a graph $G$ is $c$-closed if any two distinct non-adjacent vertices in $G$ have at most $c - 1$ common neighbours. We argue that if $(G, k, \eta)$ is a yes-instance of \bcd, then any two non-adjacent vertices of $G$ have at most $k$ common neighbours, and thus $G$ is $(k + 1)$-closed. We then fashion a reduction based on a polynomial bound for the Ramsey number of $c$-closed graphs~\cite{DBLP:journals/siamdm/KoanaKS22} to derive our kernel. 

\subparagraph*{Polynomial-time algorithm for \bcd\ on cluster graphs.} As noted above, \bcc\ is \NPH, even when the input graph is a cluster graph. Interestingly, we show that \bcd\ is polynomial-time solvable when the input graph is a cluster graph. In fact, we solve a more general variant of the problem, where we want to delete edges from a graph $G$ so that every component of the resulting cluster graph has size at least $\gamma_1$ and at most $\gamma_2$, for given $\gamma_1, \gamma_2 \geq 1$. 
Our algorithm, again, relies on the idea that when the graph is already a cluster graph, the problem is essentially one about numbers.

\subsubsection*{Related Work} 

We first discuss the work of Steinvik~\cite{steinvik2020kernelization} that has some overlap with our work. 
\subparagraph*{The master's thesis of Andreas Steinvik.} Stenvik~\cite{steinvik2020kernelization} studied the kernelization complexity of nearly identical problems as we do, but with one difference: The balance constraint $\eta$ is fixed as part of the problem definition, rather than given as part of the input. Thus for each value of $\eta$, there is a different problem. The thesis also considered the related set of problems where the requirement is that the multiplicative difference rather than the additive difference between the components be at most $\eta$. And also the {\sc Cluster Vertex Deletion} problem with  (additive and multiplicative) balance constraints. Steinvik designed  polynomial kernels for all these problems. But as $\eta$ is a fixed-constant, the kernel sizes depend on both $k$ and $\eta$. In particular, Steinvik's kernel for \bcc\ has size $\cO(k) + \eta$, and the kernels for \bcd\ and \bce\ have sizes $\cO(k^2) + \eta$. While these kernels imply \FPT\ algorithms with a $k^{\cO(k + \eta)}$ running time, the thesis explicitly asked whether the problems admit single-exponential \FPT\ algorithms; our \FPT\ algorithms thus settle this question. We must also add that our kernel for \bcc\ uses arguments that are quite similar to the ones in \cite{steinvik2020kernelization}, but we need additional  arguments, as $\eta$ is not a fixed constant in our problem.  
To the best of our knowledge, the thesis of Steinvik~\cite{steinvik2020kernelization} is the only work that considered {\sc Cluster Editing} with balance constraints, and we came to know about this thesis only after the conclusion of our work.

\subparagraph*{Graph-based clustering in the parameterized complexity framework.} {\sc Cluster Editing} was introduced by Ben-Dor et al.~\cite{DBLP:journals/jcb/Ben-DorSY99} and by Bansal et al.~\cite{DBLP:journals/ml/BansalBC04}. The problem is known to be \NPH~\cite{DBLP:journals/dam/ShamirST04}. As for the parameterized complexity of {\sc Cluster Editing}, it is textbook knowledge that the problem admits an algorithm that runs in time $3^k \cdot n^{\cO(1)}$ and a kernel with $\cO(k^3)$ vertices. As mentioned earlier, there has been a steady series of improvements~\cite{DBLP:journals/algorithmica/GrammGHN04,DBLP:journals/mst/GrammGHN05,DBLP:conf/apbc/BockerBBT08,DBLP:conf/cie/BockerB13,DBLP:conf/iwpec/Fellows06,DBLP:conf/fct/FellowsLRS07,DBLP:journals/tcs/Guo09,DBLP:journals/mst/ProttiSS09}, with the current best being a $\cO(1.62^k + n + m)$ time algorithm~\cite{DBLP:journals/jda/Bocker12} and a $2k$-vertex kernel~\cite{DBLP:conf/iwpec/CaoC10,DBLP:journals/jcss/ChenM12}. Even though {\sc Cluster Editing} was conjectured to be solvable in subexponential time~\cite{DBLP:conf/iwpec/CaoC10}, this was later disproved (under the Exponential Time Hypothesis) by Komusiewicz and Uhlmann~\cite{DBLP:journals/dam/KomusiewiczU12}. On  the approximation algorithms front, the problem is known to be {\sf APX}-hard~\cite{DBLP:journals/jcss/CharikarGW05}, but  constant factor approximation algorithms were obtained by Ailon et al.~\cite{DBLP:journals/jacm/AilonCN08}, Bansal et al.~\cite{DBLP:journals/ml/BansalBC04} and Charikar et al.~\cite{DBLP:journals/jcss/CharikarGW05}. 

The variant of {\sc Cluster Editing} with the additional requirement that there be $p$ components in the resulting cluster graph has also been studied in the literature. This problem was shown to be \NPH\ for every fixed $p \geq 2$ by Shamir et al.~\cite{DBLP:journals/dam/ShamirST04}, and a kernel with $(p + 2)k + p$ vertices was given by Guo~\cite{DBLP:journals/tcs/Guo09}. Fomin et al.~\cite{DBLP:journals/jcss/FominKPPV14} showed that the problem admits a sub-exponential algorithm running in time $\cO(2^{\cO(\sqrt{pk)}} + n + m)$ and that for $p \geq 6$, no algorithm running in time $2^{o(\sqrt{k})}$ exists unless the Exponential Time Hypothesis fails. Froese et al.~\cite{DBLP:conf/aaai/FroeseKN22} studied {\sc Cluster Editing} under a fairness constraint: Roughly speaking, there are two kinds of vertices, red and blue, and for each color we require that the number of vertices of that color involved in the edge edits must be almost proportional to the total number of vertices of that color. Froese et al.~~\cite{DBLP:conf/aaai/FroeseKN22} showed that this problem is \NPH\ even under various restricted settings, but admits an algorithm that runs in time $2^{\cO(k \log k)}$.

Abu-Khzam~\cite{DBLP:journals/jda/Abu-Khzam17} introduced a variant of {\sc Cluster Editing} with multiple optimisation objectives, including the total number of edge edits, the number of edge additions and deletions incident with each vertex and a lower bound for the size of each connected component of the resulting cluster graph.     %
This problem is \NPH\ even under restricted settings in which the only constraints are on the number of edge additions and deletions incident with each vertex. This special case has been studied by Komusiewicz and Uhlmann~\cite{DBLP:journals/dam/KomusiewiczU12} and Gutin and Yeo~\cite{DBLP:journals/dam/GutinY23} as well.

Lokshtanov and Marx~\cite{DBLP:journals/iandc/LokshtanovM13} introduced graph-based clustering problems where each ``cluster'' has to satisfy local restrictions. They considered three main variants that ask if the vertices of a graph can be partitioned in such a way that at most $q$ edges leave each part and (a) the size of each part is at most $p$, or (b) the number of non-edges in each part is at most $p$, or (c) the maximum number of non-neighbors a vertex has in its part is at most $p$. Lokshtanov and Marx~\cite{DBLP:journals/iandc/LokshtanovM13} showed that all these variants admit algorithms that run in time $2^{\cO(p)} n^{\cO(1)}$ and in time $2^{\cO(q)} n^{\cO(1)}$. %

\subparagraph*{Clustering in the parameterized complexity framework.}  Apart from the graph-based clustering results mentioned above, there have also been several works that explore the parameterized complexity of various clustering problems. Recent examples include \FPT-approximation algorithms for {\sc $k$-Means Clustering with Outliers}~\cite{DBLP:conf/aaai/Agrawal00023}, parameterized algorithms for {\sc Categorical Clustering with Size Constraints}~\cite{DBLP:journals/jcss/FominGP23} and lossy kernels for {\sc Same Size Clustering}~\cite{DBLP:journals/mst/BandyapadhyayFGPS23}; in particular, \cite{DBLP:journals/mst/BandyapadhyayFGPS23} deals with {\sc $k$-Median} clustering with the additional requirement that the clusters be of equal size.

\paragraph*{Organisation of the paper} %
The rest of the paper is organised as follows. In Section~\ref{sec:prelims}, we introduce terminology and notation   that we will use throughout the paper. 
In Section~\ref{sec:kernel}, we discuss our kernelization results, with a subsection dedicated to each of the three problems. Section 4 consists of our \FPT\ algorithms for the decision versions of the three problems. And in Section~\ref{sec:counting}, we present our \FPT\ algorithms for the counting variants of the problems. In Section~\ref{sec:conclusion}, we conclude with some pointers for future work. 

%% file: Full-version-prelims.tex
\section{Preliminaries}
\label{sec:prelims}
For $n \in \mathbb{N}$, we use $[n]$ to denote the set $\set{1, 2,\ldots, n}$, and $[n]_0 = [n] \cup \set{0}$. For integers $a$ and $b$ with $a \leq b$, we use $[a, b]$ to denote the set $\set{a, a + 1, a + 2,\ldots, b}$. For a multiset $X$ and $x \in X$, we denote the multiplicity of $x$ in $X$ by $\mul(x, X)$.

For a graph $G$, we use $V(G)$ and $E(G)$ to denote the vertex set and the edge set of $G$, respectively. Consider a graph $G$.  We use $\card{G}$ as a shorthand for $\card{V(G)}$, i.e., the number of vertices of $G$. For $F \subseteq \binom{V(G)}{2}$, we use $V(F)$ to denote the set of vertices that are the endpoints the elements of $F$, i.e., $V(F) = \set{v \in V(G) ~|~ vw \in F  \text{ for some } w \in V(G)}$; for each vertex $v \in V(F)$, we say that \emph{$F$ modifies $v$}. For $F \subseteq \binom{V(G)}{2}$ and a connected component $H$ of $G$, we say that $F$ modifies $H$ if $F$ modifies at least one vertex of $H$. For $F \subseteq \binom{V(G)}{2} \setminus E(G)$, by $G + F$, we mean the graph with vertex set $V(G)$ and edge set $E(G) \cup F$. For $F \subseteq E(G)$, by $G - F$, we mean the graph with vertex set $V(G)$ and edge set $E(G)\setminus F$. Also, for $F \subseteq \binom{V(G)}{2}$, by $G \triangle F$, we denote the graph obtained from $G$ by ``editing'' the edges and non-edges of $F$; that is, $G \triangle F$ is the graph whose vertex set is $V(G)$ and edge set is the symmetric difference of $E(G)$ and $F$, i.e., $E(G) \triangle F = (E(G) \setminus F) \cup (F \setminus E(G))$. {\bf We use $\bm{\lcomp(G)}$ to denote the number of vertices in a largest connected component of $\bm{G}$; and $\bm{\scomp(G)}$ to denote the number of vertices in a smallest connected component of $\bm{G}$.} %
For $\eta \geq 0$ and a connected component $H$ of $G$, we say that $H$ is an \emph{$\eta$-blocker} in $G$ if there exists a component $H'$ of $G$ such that $\card{\card{H} - \card{H'}} > \eta$, i.e., the size difference between $H$ and $H'$ strictly exceeds $\eta$. For $\eta \geq 0$, we say that $G$ is \emph{$\eta$-balanced} if no component of $G$ is an $\eta$-blocker; that is, $G$ is $\eta$-balanced if the size difference between no two components exceeds $\eta$. For future reference, we record the following facts that follow immediately from the definitions of an $\eta$-blocker and an $\eta$-balanced graph. 

\begin{observation} 
\label{obs:blocker-balanced}
Consider a graph $G$ and $\eta \geq 0$. 
    \begin{enumerate}
        \item Consider a connected component $H$ of $G$. Then $H$ is an $\eta$-blocker if and only if $H'$ is an $\eta$-blocker for every component $H'$ of $G$ with $\card{H'} = \card{H}$. 

        \item The graph $G$ is $\eta$-balanced if and only if $\lcomp(G) - \scomp(G) \leq \eta$. 
    \end{enumerate}
\end{observation}

\begin{observation}
\label{obs:component-delete}
Let $G$ be a graph. 
\begin{enumerate}
\item Then, $\lcomp(G') \leq \lcomp(G)$ for every subgraph $G'$ of $G$. 

\item\label{item:component-delete} Suppose that $G$ has at least two connected components. Let $H$ be a connected component of $G$, and let $G''$ be the graph obtained from $G$ by deleting $H$. Then, $\scomp(G'') \geq \scomp(G)$.  

\item Observe that the inequality in item~\ref{item:component-delete} need not hold for \emph{all} subgraphs $G''$ of $G$. 
\end{enumerate} 
\end{observation}

Throughout this paper, whenever we deal with an instance $(G, k, \eta)$ of one of the three problems---\bcc, \bcd\ or \bce---by a \emph{\bf solution} for $(G, k, \eta)$, we mean a set $F \subseteq \binom{V(G)}{2}$ such that $\card{F} \leq k$ and (a) in the case of \bcc, $F \subseteq \binom{V(G)}{2} \setminus E(G)$ such that $G + F$ is an $\eta$-balanced cluster graph, (b) in the case of \bcd, $F \subseteq E(G)$ such that $G - F$ is an $\eta$-balanced cluster graph, and (c) in the case of \bce, $G \triangle F$ is an $\eta$-balanced cluster graph. 

We use standard terminology from parameterized complexity. For terms not defined here, we refer the reader to Cygan et. al~\cite{DBLP:books/sp/CyganFKLMPPS15}.

%% file: Full-version-kernel.tex
\section{Polynomial Kernels for Balanced Cluster Modification Problems}\label{sec:kernel}
In this section, we show that \bcc, \bcd\ and  \bce\ admit polynomial kernels. 

\subsection{Polynomial Kernel for Balanced Cluster Completion} 
We start with \bcc. First of all, notice that given an instance $(G, k, \eta)$ of \bcc, we can assume without loss of generality that $G$ is a cluster graph, as every solution for $(G, k, \eta)$ must necessarily contain all the edges that we must add to make $G$ a cluster graph. So, throughout this paper, whenever we deal with \bcc, we assume that the input graph is a cluster graph. Accordingly, we restate the problem \bccfull\ (\bcc) as follows.

\defproblem{\bccfull\ (\bcc)}{An $n$-vertex  \emph{cluster} graph $G$ and non-negative integers $k$ and $\eta$.}{Decide if there exists $F \subseteq \binom{V(G)}{2} \setminus E(G)$ such that $\card{F} \leq k$ and $G + F$ is an $\eta$-balanced cluster graph.}

The following theorem is the main result of this subsection.  

\begin{theorem}
\label{thm:completion-kernel}
\bcc\ admits a kernel with at most $10k$ vertices. 
\end{theorem}

Throughout this subsection, $(G, k, \eta)$ is a given instance of \bcc. Consider a solution $F$ for the instance $(G, k, \eta)$. Recall that for a vertex $v \in V(G)$, we say that \emph{$F$ modifies $v$} if $v \in V(F)$, and for a connected component $H$ of $G$, we say that \emph{$F$ modifies $H$} if $F$ modifies at least one vertex of $H$. 
We first present an outline of our kernel. 

\subparagraph*{Outline of the kernel.} Given an instance $(G, k, \eta)$ of \bcc, we proceed as follows. We first perform a  couple of sanity checks (Reduction Rules~\ref{rule:sanity-check-1} and \ref{rule:sanity-check-2}) that eliminate obvious yes and no-instances. Let $H_1, H_2,\ldots, H_r$ be an ordering of the connected components of $G$ such that $\card{H_1} \leq \card{H_2} \leq \cdots \leq \card{H_r}$. For a carefully chosen index $s \in [r]$ such that $\sum_{j = 1}^s \card{H_j} =  \cO(k)$, we keep the components $H_1, H_2,\ldots, H_s$ and $H_r$, and delete $H_{s + 1}, H_{s + 2},\ldots, H_{r - 1}$ (Reduction Rule~\ref{rule:mid-components}). This is possible because we argue that the components $H_{s + 1}, H_{s + 2},\ldots, H_{r - 1}$ need not get modified if $(G, k, \eta)$ were a yes-instance. Finally, we delete sufficiently many vertices from $H_r$ so that we will have $\card{H_r} = \cO(k)$, and we adjust $\eta$  accordingly. These steps yield the required kernel  with $\cO(k)$ vertices. \lipicsEnd

We now present a  formal proof of Theorem~\ref{thm:completion-kernel}. And we begin by applying the following reduction rule, which rules out obvious yes-instances. Recall that the input graph $G$ is a cluster graph. %
\begin{rrule} 
\label{rule:sanity-check-1}
If $G$ is $\eta$-balanced, then we return a trivial yes-instance.
\end{rrule}

Assume from now on that Reduction Rule~\ref{rule:sanity-check-1} is not applicable. We will use the following two  facts throughout this subsection. Suppose $F$ is a solution for $(G, k, \eta)$. (1) Since $G$ and $G + F $ are both cluster graphs, $F$ modifies a vertex $v$ if and only if $F$ modifies every vertex in the connected component of $G$ that contains $v$. In fact, $F$ modifies $v$ if and only if $F$ modifies every vertex in the connected component of $G + F$ that contains $v$. (2) As $\card{F} \leq k$, $F$ modifies $\card{V(F)} \leq 2k$ vertices. These facts immediately lead to the following lemma.  

\begin{lemma}\label{lem:enlarged-cliques}
Assume that $(G, k, \eta)$ is a yes-instance and let $F \subseteq \binom{V(G)}{2} \setminus E(G)$ be a solution for $(G, k, \eta)$. Let $\ca{H}$ be the set of connected components of $G$ that $F$ modifies. Then, (1)  $\card{H} \leq k$ for every $H \in \ca{H}$, and (2) $\sum_{H \in \ca{H}} \card{H} \leq 2k$. 
\end{lemma}
\begin{proof}
Consider $H \in \ca{H}$. 
Since $F$ modifies $H$, $F$ contains at least one edge incident with each vertex of $H$. And notice that each edge in $F$ has at most one endpoint in $V(H)$. Thus $F$ must contain at least $\card{H}$ edges incident with $V(H)$, and hence $\card{H} \leq \card{F} \leq k$. 
Now, notice that since  $\card{F} \leq k$ and since each edge in $F$ modifies exactly two vertices, the number of vertices that $F$ modifies is at most $2\card{F} \leq 2k$. Notice also that $F$ modifies a vertex $v$ if and only if $F$ modifies every vertex in the connected component of $G$ that contains $v$. Thus $F$ modifies every vertex of $H$ for every $H \in \ca{H}$. Hence statement (2) follows. %
\end{proof}

Recall that $\lcomp(G)$ and $\scomp(G)$ respectively denote the number of vertices in a largest connected component of $G$ and the number of vertices in a smallest connected component of $G$. We now prove the following lemma that establishes bounds on $\scomp(G)$ and $\lcomp(G)$, if $(G, k, \eta)$ were indeed a yes-instance. 
\begin{lemma}
\label{lem:sanity-check-2}
If $(G, k, \eta)$ is a yes-instance of \bcc, then $\scomp(G) \leq k$ and $\lcomp(G) \leq 2k + \eta$.
\end{lemma}
\begin{proof}
Assume that $(G, k, \eta)$ is a yes-instance of \bcc, and let $F$ be a solution for $(G, k, \eta)$.  Recall that since Reduction Rule~\ref{rule:sanity-check-1} is not applicable, $G$ is not $\eta$-balanced. Let $G_1$ be a smallest connected component of $G$ and $G_2$ a largest connected component of $G$. Then $\scomp(G) = \card{G_1}$ and $\lcomp(G) = \card{G_2}$, and since $G$ is not $\eta$-balanced, $\card{G_2} - \card{G_1} > \eta$.  But then $F$ must modify $G_1$. Hence, by Lemma~\ref{lem:enlarged-cliques}, we have $\scomp(G) = \card{G_1} \leq k$.  

To see that $\lcomp(G) \leq 2k + \eta$, consider the graph $G' = G + F$. Since $F$ is a solution, we have $\lcomp(G') - \scomp(G') \leq \eta$. Let $G'_1$ and $G'_2$ be the connected components of $G'$ that contain $G_1$ and $G_2$, respectively. (Notice that we may have $G'_1 = G'_2$.) Since $F$ is a solution, the graph $G' = G + F$ is $\eta$-balanced, and therefore, we have $\card{G'_2} - \card{G'_1} \leq \eta$. Now, observe the following facts. (1) Since $\card{G_1} \leq k$ and $\card{F} \leq k$, we have $\card{G'_1} \leq \card{G_1} + \card{F} \leq 2k$. (2) We have $\card{G_2} \leq \card{G'_2}$ and hence $\card{G_2} - \card{G'_1} \leq \card{G'_2} - \card{G'_1} \leq \eta$, which implies that $\card{G_2} \leq \card{G'_1} + \eta \leq 2k + \eta$. That is, $\lcomp(G) = \card{G_2} \leq 2k + \eta$. 
\end{proof}

We now introduce the following reduction rule, the correctness of which follows immediately from Lemma~\ref{lem:sanity-check-2}. Recall that we are still under the assumption that $G$ is not $\eta$-balanced. 

\begin{rrule} 
\label{rule:sanity-check-2}
If $\scomp(G) \geq k + 1$ or $\lcomp(G) \geq 2k + \eta + 1$, then we return a trivial no-instance.
\end{rrule}

\begin{sloppypar}
Let $r$ be the number of connected components of $G$.  Fix an ordering $(H_1, H_2,\ldots, H_r)$ of the components of $G$ such that $\card{H_1} \leq \card{H_2} \leq \cdots \leq \card{H_r}$. Recall that since Reduction Rules~\ref{rule:sanity-check-1} and \ref{rule:sanity-check-2} are not applicable, we have $\card{H_1} = \scomp(G) \leq k$. We define $s \in [r]$ as follows. {\bf Let $\bm{s}$ be the largest index in $\bm{[r]}$ such that $\bm{\sum_{j = 1}^s \card{H_j} \leq 4k}$.} We now apply the following reduction rule. 
\end{sloppypar}

\begin{rrule}
\label{rule:mid-components}
If $s + 1 < r$, then we delete the components $H_{s + 1}, H_{s + 2},\ldots, H_{r - 1}$ from $G$. 
\end{rrule}

\begin{lemma}
\label{lem:mid-components-safe}
Reduction Rule~\ref{rule:mid-components} is safe. 
\end{lemma}
\begin{proof}
Let $G'$ be the graph obtained from $G$ by applying Reduction Rule~\ref{rule:mid-components}. We will show that $(G, k, \eta)$ is a yes-instance of \bcc\ if and only if $(G', k, \eta)$ is a yes-instance. 

Assume that $(G, k, \eta)$ is a yes-instance. 
We first prove the following claim. 
\begin{sloppypar}
\begin{claim}
\label{claim:mid-components}
There exists a solution $F$ for $(G, k, \eta)$ such that $F$ does not modify any of the components $H_{s + 1}, H_{s + 2},\ldots, H_{r - 1}$.
\end{claim}
\end{sloppypar}
\begin{claimproof}
Let $F$ be a solution for $(G, k, \eta)$ such that of all the solutions for $(G, k, \eta)$, $F$ modifies the fewest number of components from $\set{H_{s + 1}, H_{s + 2},\ldots, H_{r - 1}}$.  We will show that $F$ does not modify any of the components $H_{s + 1}, H_{s + 2},\ldots, H_{r - 1}$. Assume for a contradiction that $F$ does modify $H_q$ for some $q \in [s + 1, r - 1]$. 

We first claim that there exists $q' \in [s]$ such that $F$ does not modify $H_{q'}$. Suppose this is not true. Then $F$ modifies $H_j$ for every $j \in [s]$. Thus, $F$ modifies $H_1, H_2,\ldots, H_s$ and $H_q$, where $q \geq s + 1$. That is, $F$ modifies at least $\card{H_q} + \sum_{j = 1}^s \card{H_j} \geq \sum_{j = 1}^{s + 1} \card{H_j} > 4k$ vertices, which, by Lemma~\ref{lem:enlarged-cliques}, is not possible; the last inequality holds because of the definition of $s$. Fix $q' \in [s]$ such that $F$ does not modify $H_{q'}$. 

We now construct another solution $F'$ from $F$ as follows. Informally, we obtain $F'$ from $F$ by swapping the roles of $H_{q}$ and $H_{q'}$.   
More formally, let Let $F_{q} \subseteq F$ be the set of edges in $F$ that have an endpoint in $H_{q}$, and let $X_{q} = \bigcup_{v \in V(H_q)}\{u \in V(G)\setminus V(H_{q}) ~|~ uv \in F_{q} \}$; that is, $X_{q}$ consists of the ``non-$H_{q}$-endpoints'' of the edges in $F_{q}$. Notice that since $G + F$ is a cluster graph  and since $H_q$ is a connected component of $G$, for a vertex $u \in V(G)$, we have $u \in X_q$ if and only if $u$ is not adjacent to any vertex of $H_q$ in the graph $G$, but adjacent to every vertex of $H_q$ in the graph $G + F$, (in which case all the edges between $u$ and $H_q$ belong to $F_q \subseteq F$). We thus have $\card{F_q} = \card{X_q}\card{H_q}$. Now, let $F_{q'} = \{uw \in \binom{V(G)}{2} ~|~ u \in X_{q} \text{ and } w \in V(H_{q'})\}$; that is, $F_{q'}$ consists of ``new edges'' between $V(H_{q'})$ and all those vertices in $V(G) \setminus V(H_q)$ that are modified by $F_q$. Notice that $\card{F_{q'}} = \card{X_q} \card{H_{q'}}$. Finally, we define $F' = (F \setminus F_{q}) \cup F_{q'}$. Notice that $F'$ does not modify $H_{q}$ and hence $H_q$ is a component of $G + F'$. 

We now show that $F'$ is also a solution for the instance $(G, k, \eta)$. Let us first see that $\card{F'} \leq k$. As $q' \leq s < s + 1 \leq q$, we have $\card{H_{q'}} \leq \card{H_{q}}$, and hence $\card{F_{q'}} = \card{X_q}\card{H_{q'}} \leq \card{X_q} \card{H_q} = \card{F_q}$. Notice also that $F_q \cap F_{q'} = \emptyset$. We thus have $\card{F'} = \card{(F \setminus F_q) \cup F_{q'}} = \card{F} - \card{F_q} + \card{F_{q'}} \leq \card{F} \leq k$. Also, since $G + F$ is a cluster graph, so is $G + F'$. To show that $F'$ is a solution for $(G, k, \eta)$, we now only need to show that $G + F'$ is $\eta$-balanced. And for this, it is enough to show  that {\bf (i)} $\lcomp(G + F') \leq \lcomp(G + F)$ and {\bf (ii)} $\scomp(G + F') \geq \scomp(G + F)$, which will imply that $\lcomp(G + F') - \scomp(G + F') \leq \lcomp(G + F) - \scomp(G + F) \leq \eta$. 
To see that (i) and (ii) hold, observe that the only difference between the graphs $G + F$ and $G + F'$ are the components of the these graphs that contain $V(H_q)$ and $V(H_{q'})$. Let $C_{F, q}$ be the component of $G + F$ that contains $V(H_{q})$ and $C_{F', q'}$ be the component of $G + F'$ that contains $H_{q'}$. Because $\card{H_{q'}} \leq \card{H_q}$ and $C_{F, q} - V(H_q) = C_{F', q'} - V(H_{q'})$, we have $\card{C_{F', q'}} \leq \card{C_{F, q}}$. Let us first see that (i) holds. 
\begin{enumerate}[(a)]
    \item If $C_{F', q'}$ is a largest component of $G + F'$, then $\lcomp(G + F') = \card{C_{F', q'}} \leq \card{C_{F, q}} \leq \lcomp(G + F)$. 
    \item If $H_q$ is largest component of $G + F'$, then again, $\lcomp(G + F') = \card{H_q} < \card{C_{F, q}} \leq \lcomp(G + F)$. 
    \item If the previous two cases do not hold, then $\lcomp(G + F') = \lcomp(G + F)$, as every component $H'$ of $G + F'$ such that $H' \neq C_{F', q'}$ and $H' \neq H_{q}$ is also a component of $G + F$.  
\end{enumerate}
In any case, (i) holds. Now, to see that (ii) holds,  observe the following facts. 
\begin{enumerate}[(a)]
    \item If $C_{F', q'}$ is a smallest component of $G + F'$, then $\scomp(G + F') = \card{C_{F', q'}} > \card{H_{q'}} \geq \scomp(G + F)$. 
    \item If $H_q$ is a smallest component of $G + F'$, then again, $\scomp(G + F') = \card{H_q} \geq \card{H_{q'}} \geq \scomp(G + F)$. 
    \item If the previous two cases do not hold, then $\scomp(G + F') = \scomp(G + F)$, as every component $H'$ of $G + F'$ such that $H' \neq C_{F', q'}$ and $H' \neq H_{q}$ is also a component of $G + F$. 
\end{enumerate}
Thus (ii) holds as well. We have thus shown that $F'$ is also a solution for $(G, k, \eta)$. Notice that (i) $F$  modifies $H_q$ but $F'$ does not, and (ii) for every component $H_j$ of $G$ with $j \in [s + 1, r - 1] \setminus \set{q}$, $F$ modifies $H_j$ if and only if $F'$ modifies $H_{j}$. Thus, $F'$ modifies strictly fewer components from $\{H_j ~|~ j \in [s + 1, r - 1]\}$ than $F$ does, which contradicts the definition of $F$. 
\end{claimproof}

Now, consider a solution $F$ for $(G, k, \eta)$. In light of Claim~\ref{claim:mid-components}, we assume that $F$ does not modify any of the components $H_{s + 1}, H_{s + 2},\ldots, H_{r - 1}$. We now argue that $F$ is a solution for $(G', k, \eta)$ as well. (Recall that $G'$ is the graph obtained from $G$ by applying Reduction Rule~\ref{rule:mid-components}). Notice first that $G' + F$ is an induced subgraph of $G + F$, and hence $G' + F$ is a cluster graph. To see that $G' + F$ is $\eta$-balanced, notice that each component of $G'+ F$ is also a component of $G + F$, and hence $\lcomp(G' + F) \leq \lcomp(G + F)$ and $\scomp(G' + F) \geq \scomp(G + F)$. We thus have $\lcomp(G' + F) - \scomp(G' + F) \leq \lcomp(G + F) - \scomp(G + F) \leq \eta$. We have thus shown that $(G', k, \eta)$ is a yes-instance.

Conversely, assume that $(G', k, \eta)$ is a yes-instance, and let $R$ be a solution for $(G', k, \eta)$. We claim that $R$ is a solution for $(G, k, \eta)$ as well. Notice that $G + R$ is indeed a cluster graph. So we only need to prove that $G + R$ is $\eta$-balanced. Again, to prove this, it is enough to prove that (i) $\lcomp(G + R) \leq \lcomp(G' + R)$ and (ii) $\scomp(G + R) \geq \scomp(G' + R)$, which will imply that $\lcomp(G + R) - \scomp(G + R) \leq \lcomp(G' + R) - \scomp(G' + R) \leq \eta$. 

Let $C_{R, r}$ be the component of $G' + R$ that contains $V(H_{r})$. Now, notice that every component of $G' + R$ is also a component of $G + R$. In addition, $G + R$ contains the components $H_{s + 1}, H_{s + 2}, \ldots, H_{r - 1}$. Hence, we have 
\begin{align*}
\lcomp(G + R) 
&= \max\set{\lcomp(G' + R), \card{H_{s + 1}}, \card{H_{s + 2}},\ldots, \card{H_{r - 1}}} \\
&= \max\set{\lcomp(G' + R), \card{H_{r - 1}}} ~~~~~~~~~~(\text{Because } \card{H_{s + 1}} \leq \card{H_{s + 2}} \leq \cdots \leq \card{H_{r - 1}}.)\\
&\leq \max\set{\lcomp(G' + R), \card{H_{r}}} ~~~~~~~~~~~~~(\text{Because } \card{H_{r - 1}} \leq \card{H_r}.)\\
&\leq \max\set{\lcomp(G' + R), \card{C_{R, r}}} ~~~~~~~~~~~(\text{Because } \card{H_{r}} \leq \card{C_{R, r}}.)\\
&=\lcomp(G' + R). ~~~~~~~~~~~~~~~~~~~~~~~~~~~~~(\text{Because } C_{R, r} \text{ is a component of } G' + R.)
\end{align*} 
\begin{sloppypar}
Let us now prove that $\scomp(G + R) \geq \scomp(G' + R)$. As before, $\scomp(G + R) = \min\set{\scomp(G' + R), \card{H_{s + 1}}, \ldots, \card{H_{ r - 1}}} = \min\set{\lcomp(G' + R), \card{H_{s + 1}}}$. There are two possible cases. 
\end{sloppypar}
\begin{description}
    \item[Case 1:] $R$ does not modify $H_j$ for some $j \in [s]$. Then $H_j$ is a component of $G' + R$. We thus have $\scomp(G' + R) \leq \card{H_j} \leq \card{H_{s + 1}}$, which implies that $\scomp(G + R) = \min\set{\scomp(G' + R), \card{H_{s + 1}}} = \scomp(G' + R)$. 
    \item[Case 2:] $R$ modifies $H_j$ for every $j \in [s]$. Then, by Lemma~\ref{lem:enlarged-cliques}, we have $\sum_{j \in [s]}\card{H_j} \leq 2k$. But then, by the definition of $s$, we must have $\card{H_{s + 1}} \geq 2k + 1$. We now argue that $G' + R$ contains a component of size at most $2k$, which will imply that $\scomp(G' + R) \leq 2k$.

    Since $\card{H_{s + 1}} \geq 2k + 1$, we have $\card{H_r} \geq \card{H_{s + 1}} \geq 2k + 1$. Thus, by Lemma~\ref{lem:enlarged-cliques}, $R$ does not modify the component $H_r$. But recall that the only components of $G'$ are $H_1, H_2,\ldots, H_s$ and $H_r$. Since $R$ does not modify $H_r$, we can  conclude that $G' + R$ contains a component with at most $\card{G'} - \card{H_r} = \sum_{j = 1}^s \card{H_j} \leq 2k$ vertices. 
    Therefore, $\scomp(G' + R) \leq 2k < \card{H_{s + 1}}$, which implies that $\scomp(G + R) = \min\set{\scomp(G' + R), \card{H_{s + 1}}} = \scomp(G' + R)$.
\end{description}
We have thus shown that $\lcomp(G + R) \leq \lcomp(G' + R)$ and $\scomp(G + R) \geq \scomp(G' + R)$, which implies that $\lcomp(G + R) - \scomp(G + R) \leq \lcomp(G' + R) - \scomp(G' + R) \leq \eta$. Thus $G + R$ is $\eta$-balanced, and hence $(G, k, \eta)$ is a yes-instance. 
\end{proof}

Assume from now on that Reduction Rule~\ref{rule:mid-components} is no longer applicable. Then $G$ contains $s + 1$ components: $H_1, H_2,\ldots, H_s$ and $H_r$, and we have $\sum_{j = 1}^s \card{H_{j}} \leq 4k$ and $\lcomp(G) = \card{H_r}$. We observe the following facts. 

\begin{observation}
\label{obs:bounds}
\begin{enumerate}

\item\label{item:eta-is-small} If $\eta \geq \card{H_r} = \lcomp(G)$, then $(G, k, \eta)$ is a yes-instance as $G$ is $\eta$-balanced; and since Reduction Rule~\ref{rule:sanity-check-1} is not applicable, we must have $\eta < \lcomp(G)$. 
\item If $\lcomp(G) = \card{H_r} \leq 4k$, then we already have $\card{G} = \card{H_r} + \sum_{j = 1}^s \card{H_{j}} \leq 4k + 4k = 8k$. Thus, $(G, k, \eta)$ is a kernel with $8k$ vertices. 
\item If $\eta \leq 4k$, then since Reduction Rule~\ref{rule:sanity-check-2} is not applicable, we have $\card{H_r} = \lcomp(G) \leq 2k + \eta \leq 2k + 4k = 6k$. We thus have $\card{G} = \card{H_r} + \sum_{j = 1}^s \card{H_{j}} \leq 6k + 4k = 10k$. Thus, $(G, k, \eta)$ is a kernel with $10k$ vertices in this case. 
\end{enumerate}
\end{observation}

In light of Observation~\ref{obs:bounds}, we  apply the following reduction rule.

\begin{rrule}
\label{rule:small-Hr-eta}
If $\lcomp(G) = \card{H_r} \leq 4k$ or $\eta \leq 4k$, then we simply return the instance $(G, k, \eta)$.  
\end{rrule}

We assume from now on that Reduction Rule~\ref{rule:small-Hr-eta} is no longer applicable, and hence $\lcomp(G) = \card{H_r} > 4k$ and $\eta > 4k$. {\bf We define $\bm{\eta' = 4k}$ and $\bm{N = \lcomp(G) - (\eta - 4k)}$.} 

\begin{observation}
\label{obs:new-bounds}
\begin{enumerate}
\item\label{item:lb-for-N} Recall that we have $\eta < \lcomp(G)$ (Observation~\ref{obs:bounds}-\ref{item:eta-is-small}). Therefore, $N = \lcomp(G) - (\eta - 4k) > \eta - (\eta - 4k) = 4k$. 
\item\label{item:ub-for-N} Recall also that we have $\lcomp(G) \leq 2k + \eta$. Therefore, $N = \lcomp(G) - (\eta - 4k) \leq 2k + \eta - (\eta - 4k) = 6k$. 
\end{enumerate}
\end{observation}

We now apply the following reduction rule once. 
\begin{rrule}
\label{rule:final-rule-for-bcc}
 We delete $\lcomp(G) - N$ (arbitrarily chosen) vertices from $H_r$; let us denote the resulting component by $H'$ and the resulting graph by $G'$. We return the instance $(G', k, \eta')$
\end{rrule}

\begin{lemma}
Reduction Rule~\ref{rule:final-rule-for-bcc} is safe. 
\end{lemma}
\begin{proof}
Informally, Reduction Rule~\ref{rule:final-rule-for-bcc} is safe because we deleted $\lcomp(G) - N = \eta - 4k$ vertices from $H_r$ and we decremented $\eta$ by the same amount; that is, we set $\eta' = 4k = \eta - (\eta - 4k)$. The safeness of the rule then follows from the facts that no solution for $(G, k, \eta)$ modifies $H_r$ as $\card{H_r} > 4k$, and no solution for $(G', k, \eta')$ modifies $H'$ as $\card{H'} = \lcomp(G) - (\lcomp(G) - N) = N > 4k$. We now prove this more formally. 

Recall that $G$ consists of the connected components $H_1, H_2,\ldots, H_s$ and $H_r$; and $G'$ consists of the connected components $H_1, H_2,\ldots, H_s$ and $H'$. Also, $\lcomp(G) = \card{H_r}$; and $\card{H'} = \card{H_r} - (\lcomp(G) - N) = N$. By Observation~\ref{obs:new-bounds}-\ref{item:lb-for-N}, we have $N > 4k$. That is, $H'$ is a component of $G'$ with at least $4k + 1$ vertices. And by the definition of $s$, we have $\sum_{j = 1}^{s} \card{H_j} \leq 4k$. Thus $H'$ is the unique largest component of $G'$, and hence $\lcomp(G') = \card{H'} = N$.  

Assume that $(G, k, \eta)$ is a yes-instance, and let $F$ be a solution for $(G, k, \eta)$. We argue that $F$ is a solution for $(G', k, \eta')$ as well. First, since $\card{H_r} > 4k$, by Lemma~\ref{lem:enlarged-cliques}, $F$ does not modify $H_r$. Now, consider the graph $G' + F$; this graph is well-defined and it is indeed a cluster graph as $F$ does not modify $H_r$. In particular, $F$ does not modify $H'$ (as $V(H') \subseteq V(H_r)$). Also, since $\sum_{j = 1}^{s} \card{H_j} \leq 4k$, $\card{H_r} > 4k$ and $\card{H'} > 4k$, $H_{r}$ and $H'$ respectively are the unique largest components of $G + F$ and $G' + F$. Thus, $\lcomp(G + F) =  \card{H_r} = \lcomp(G)$ and $\lcomp(G' + F) = \card{H'} = N  = \lcomp(G')$.  
Again, since $F$ does not modify $H_r$, each component of $G + F$, except $H_r$, is also a component of $G' +F$; and each component of $G' + F$, except $H'$, is also a component of $G +F$. In particular, a smallest component of $G + F$ is a smallest component of $G' + F$, and vice versa. That is, $\scomp(G + F) = \scomp(G' + F)$. 
Since $G + F$ is $\eta$-balanced, we have 
$\lcomp(G + F) - \scomp(G + F) \leq \eta$, which implies that $\lcomp(G + F) - (\eta - 4k) - \scomp(G + F) \leq \eta - (\eta - 4k)$, which implies that $N - \scomp(G' + F) \leq 4k$. That is, $\lcomp(G' + F) - \scomp(G' + F) \leq \eta'$, and thus $G' + F$ is $\eta$-balanced. 

Conversely, assume that $(G', k', \eta')$ is a yes-instance, and let $F'$ be a solution for $(G', k, \eta')$. Then, $F'$ does not modify $H'$, and  we can argue that $F'$ is a solution for $(G, k, \eta)$. In fact, we have (i) $\lcomp(G + F') = \card{H_r} = \lcomp(G)$, (ii) $\lcomp(G' + F') = \card{H'} = N = \lcomp(G')$ and (iii) $\scomp(G + F') = \scomp(G' + F')$. Since $G' + F'$ is $\eta'$-balanced, we have $\lcomp(G' + F') - \scomp(G' + F') \leq \eta'$, which implies that $N - 
\scomp(G' + F') \leq \eta'$, which implies that $N + (\eta - 4k) - \lcomp(G' + F') \leq \eta' + (\eta - 4k)$, which implies that $\lcomp(G + F') - \scomp(G + F') \leq \eta$. Thus $G + F'$ is $\eta$-balanced. 
\end{proof}

Consider the instance $(G', k, \eta')$ returned by Reduction Rule~\ref{rule:final-rule-for-bcc}. Recall that $\eta' = 4k$ and $\card{H'} = N \leq 6k$ (Observation~\ref{obs:new-bounds}). The components of $G'$ are $H_1, H_2, \ldots, H_s$ and $H'$, and we thus have $\card{G'} = \sum_{j = 1}^s \card{H_j} + \card{H'} \leq 4k + 6k = 10k$. That is, $(G', k, \eta')$ is a kernel with at most $10k$ vertices. We have thus proved Theorem~\ref{thm:completion-kernel}.

\subsection{Polynomial Kernel for Balanced Cluster Deletion}
We formally define the \bcdfull\ (\bcd) problem as follows. 

\defproblem{\bcdfull\ (\bcd)}{A graph $G$ and non-negative integers $k$ and $\eta$.}{Decide if there exists $F \subseteq E(G)$ such that $\card{F} \leq k$ and $G - F$ is an $\eta$-balanced cluster graph.}

We now show that \bcd\ admits a kernel with $\cO(k^4)$ vertices. To design our kernel, we rely on a structural property exhibited by yes-instances: We argue that if $(G, k, \eta)$ is a yes-instance, then any  two non-adjacent vertices in $G$ have at most $k$ common neighbours. To exploit this observation algorithmically,  we turn to the class of $c$-closed graphs, which was introduced by Fox et al.~\cite{DBLP:journals/siamcomp/FoxRSWW20}. We define $c$-closed graphs below, and briefly summarise their properties that we will be using.   
\subparagraph*{$\bm{c}$-Closed Graphs.} For a positive integer $c$, we say that a graph $G$ is $c$-closed if any two distinct non-adjacent vertices in $G$ have at most $c -1$ neighbours in common. That is, for distinct $u, v \in V(G)$, we have $\card{N(u) \cap N(v)} \leq c - 1$ if $uv \notin E(G)$. We will use the following lemma, which relies on the fact that $c$-closed graphs have polynomially bounded Ramsey numbers.  %

\begin{lemma}[\cite{DBLP:journals/siamdm/KoanaKS22}]
\label{lem:clique-or-indset}
For $a, b, c \in \mathbb{N}$, let $\bm{R_c(a, b) = (a - 1)(b - 1) + (c - 1){{b - 1} \choose {2}} + 1}$. 
There is an algorithm that, given $a, b, c \in \mathbb{N}$ and a \cc\ graph $G$ on at least $R_c(a, b)$ vertices as input, runs in polynomial time, and returns either a maximal clique in $G$ of size at least $a$ or an independent set in $G$ of size $b$. %
\end{lemma}

\subparagraph*{Outline of the kernel.} Consider an instance $(G, k, \eta)$ of \bcd. Our kernel has three main steps. {\bf In the first step, we bound the number of vertices that belong to the non-clique components of $\bm{G}$ by $\bm{\cO(k^4)}$.} We first bound the number of non-clique components by $k$; this is straightforward as we must delete at least one edge from each non-clique component to turn $G$ into a cluster graph. We then bound the size of each non-clique component by $\cO(k^3)$. To do this, we argue that if $(G, k , \eta)$ is a yes-instance, then every pair of non-adjacent vertices have at most $k$ common neighbours, and thus $G$ is $(k + 1)$-closed. In particular, each non-clique component of $G$ is $(k + 1)$-closed. To bound the size of such components, we fashion a reduction rule based on Lemma~\ref{lem:clique-or-indset} that works as follows. For each non-clique component $H$ of $G$ of size at least $R_{k + 1}(k + 2, k + 2)$, we run the algorithm of Lemma~\ref{lem:clique-or-indset} on $H$ (with $a = b = k + 2$ and $c = k + 1$). If the algorithm returns an independent set of size $k + 2$, then we argue that $(G, k, \eta)$ is a no-instance. If the algorithm returns a maximal clique $Q$ of size at least $k + 2$, then we delete all the edges with exactly one endpoint in $V(Q)$; this is possible because we argue that $Q$ must necessarily be a connected component of $G - F$ for every solution $F$ for $(G, k, \eta)$. When this reduction rule is no longer applicable, every non-clique component will have size at most $R_{k + 1}(k + 2, k + 2) - 1 = \cO(k^3)$. 

Having dealt with non-clique components, we then turn to the clique-components of $G$. We classify the clique-components of $G$ into two types---small and large. By small cliques, we mean cliques of size at most $k + 1$, and by large cliques, we mean cliques of size at least $k + 2$. {\bf In the second step, we bound the number of vertices of $\bm{G}$ that belong to small clique-components  by $\bm{\cO(k^3)}$.} To do this, we show that for each $j \in [k + 1]$, we only need to keep at most $k + 1$ clique-components of size exactly $j$. 

{\bf In the third step, we bound the number of vertices of $\bm{G}$ that belong to large clique-components by $\bm{\cO(k^3)}$.} Specifically, we apply a reduction rule that works as follows. Let $H_1, H_2,\ldots, H_r$ be the large clique-components of $G$ such that $\card{H_1} \leq \card{H_2} \leq \cdots \leq \card{H_r}$. If $\card{H_1} \leq R_{k + 1}(k + 2, k + 2) = \cO(k^3)$, then we delete the components $H_2, H_3,\ldots, H_{r - 1}$; otherwise, we delete the components $H_1, H_2,\ldots, H_{r - 1}$. This is possible because we argue that no solution for $(G, k, \eta)$ modifies the large clique-components. Finally, we bound $\card{H_r}$ by $\cO(k^3)$ by deleting sufficiently many vertices from $H_r$, and we adjust $\eta$ accordingly. These steps lead to the required kernel with $\cO(k^4)$ vertices.  \lipicsEnd

We now proceed to designing our kernel for \bcd. Recall that for a graph $G$, $\lcomp(G)$ and $\scomp(G)$ respectively denote the number of vertices in a largest connected component and a smallest connected component of $G$. Consider an instance $(G, k, \eta)$ of \bcd. We first apply the following reduction rule that eliminates obvious yes-instances. 

\begin{rrule}
\label{rule:deletion-sanity-check}
If $G$ is an $\eta$-balanced cluster graph, then we return a trivial yes-instance. 
\end{rrule}

We now apply the following reduction rule, which upper bounds $\eta$ by $\lcomp(G)$. The correctness of the rule follows from the observation that $\eta$ need never be larger than $\lcomp(G)$. 

\begin{rrule}
\label{rule:deletion-bound-eta}
If $\eta > \lcomp(G)$, then we return the instance $(G, k, \hat \eta)$, where $\hat \eta = \lcomp(G)$. 
\end{rrule}

\begin{lemma}
\label{lem:deletion-bound-eta}
Reduction Rule~\ref{rule:deletion-bound-eta} is safe. 
\end{lemma}
\begin{proof}
Suppose that $\eta > \lcomp(G)$. 
Observe that for any  set $F \subseteq E(G)$, the graph $G - F$ is trivially $\eta$-balanced and trivially $\hat \eta$-balanced. To see this, notice that as $G - F$ is a subgraph of $G$, we have $\lcomp(G - F) \leq \lcomp(G) = \hat \eta < \eta$. Therefore, we always have $\lcomp(G - F) - \scomp(G - F) \leq \lcomp(G - F) \leq \hat \eta < \eta$. Thus, $F$ is a solution for the instance $(G, k, \eta)$ if and only if $F$ is a solution for the instance $(G, k, \hat \eta)$.  
\end{proof}

Assume from now on that Reduction Rules~\ref{rule:deletion-sanity-check} and \ref{rule:deletion-bound-eta} are no longer applicable. Since Rule~\ref{rule:deletion-sanity-check} is not applicable, either $G$ contains a non-clique component or $G$ is a cluster graph but not $\eta$-balanced. Since Reduction Rule~\ref{rule:deletion-bound-eta} is not applicable, we have $\eta \leq \lcomp(G)$. 
We now apply a series of reduction rules that deal separately with the non-clique components (Reduction Rules~\ref{rule:deletion-sanity-check-2} and \ref{rule:deletion-clique-or-indset}), the components that are ``small cliques'' (Reduction Rule~\ref{rule:deletion-small-cliques}) and the components that are ``large'' cliques (Reduction Rules~\ref{rule:deletion-large-cliques-1}-\ref{rule:deletion-final-rule-for-bcd}). We begin with non-clique components, and first  apply the following reduction rule. The safeness of the rule follows immediately from the fact that we must delete at least one edge from each non-clique component of $G$ to turn $G$ into a cluster graph. 
\begin{rrule}
\label{rule:deletion-sanity-check-2}
If $G$ contains at least $k + 1$ non-clique components, then we return a trivial no-instance. 
\end{rrule}

We assume from now on that Reduction Rule~\ref{rule:deletion-sanity-check-2} is no longer applicable. Thus $G$ has at most $k$ non-clique connected components. We now introduce the following reduction rule. 
\begin{rrule}
\label{rule:non-adjacent}
If $G$ contains two distinct non-adjacent vertices with at least $k + 1$ common neighbours, then we return a trivial no-instance. 
\end{rrule}

\begin{lemma}
Reduction Rule~\ref{rule:non-adjacent} is safe. 
\end{lemma}
\begin{proof}
To prove the lemma, it is enough to show that $(G, k, \eta)$ is a no-instance if $G$ contains two distinct non-adjacent vertices with at least $k + 1$ common neighbours. Let $u, v \in V(G)$ be distinct vertices such that $uv \notin E(G)$ and $\card{N(u) \cap N(v)} \geq k + 1$, and let $x_1, x_2,\ldots, x_{k + 1} \in V(G)$ be $k + 1$ distinct common neighbours of $u$ and $v$.  
Suppose that $(G, k, \eta)$ is a yes-instance, and let $F$ be a solution for $(G, k, \eta)$. Then $G - F$ is a cluster graph. Since $u$ and $v$ are non-adjacent, $u$ and $v$ must be in different connected components of $G - F$. Therefore, for every $i \in [k + 1]$, $F$ must contain either the edge $ux_i$ or the edge $vx_i$. Thus $\card{F} \geq k + 1$, which contradicts the fact that $\card{F} \leq k$. 
\end{proof}

Assume from now on that Reduction Rule~\ref{rule:non-adjacent} is no longer applicable. Thus any two distinct non-adjacent vertices in $G$ have at most $k$ common neighbours, and hence $G$ is $(k + 1)$-closed. We now prove the following two lemmas, which we will use to fashion a reduction rule (Reduction Rule~\ref{rule:deletion-clique-or-indset}) that bounds the size of non-clique components. 

\begin{lemma}
\label{lem:indset-k+2}
If a connected component of $G$ contains an independent set of size $k + 2$, then $(G, k , \eta)$ is a no-instance. 
\end{lemma}
\begin{proof}
Let $H$ be a connected component of $G$, and let $I \subseteq V(H)$ be an independent set of size $k + 2$ in $H$. Suppose now that $(G, k, \eta)$ is a yes-instance, and let $F \subseteq E(G)$ be a solution for $(G, k, \eta)$. Then $G - F$ is a cluster graph, and in particular, $H - F$ is a cluster graph. Therefore, as the vertices of $I$ are pairwise non-adjacent, every connected component of $H - F$ contains at most one vertex of $I$. Since $\card{F} \leq k$ and $H$ is connected, the graph $H - F$ has at most $k + 1$ connected components. Then, as $\card{I} = k + 2$, by the pigeonhole principle, there exists a connected component of $H - F$ that contains at least two vertices of $I$, which is a contradiction.   
\end{proof}

\begin{lemma}
\label{lem:deletion-clique-k+2}
Assume that $(G, k , \eta)$ is a yes-instance. If $G$ contains a maximal clique, say $Q$, of size at least $k + 2$, then $Q$ is a connected component of $G - F$ for every solution $F \subseteq E(G)$ for $(G, k, \eta)$. 
\end{lemma}
\begin{proof}
Fix a solution $F$ for $(G, k, \eta)$. Let $Q$ be a maximal clique  in $G$ of size at least $k + 2$. Consider the cluster graph $G - F$. Notice that as $\card{Q} \geq k + 2$, we have to delete at least $k + 1$ edges from $Q$ to separate the vertices of $Q$ into two or more connected components. Since $\card{F} \leq k$, we can conclude that $Q$ is fully contained in a connected component of $G - F$. Now, since $Q$ is a maximal clique in $G$ and since each component of $G - F$ is a clique, the connected component of $G - F$ that contains $Q$ does not contain any vertex from $V(G) \setminus V(Q)$. We can thus conclude that $Q$ is a connected component of $G - F$. 
\end{proof}

Based on Lemmas~\ref{lem:clique-or-indset}, \ref{lem:indset-k+2} and \ref{lem:deletion-clique-k+2}, we now introduce the following reduction rule. The correctness of the rule follows from Lemmas~\ref{lem:indset-k+2} and \ref{lem:deletion-clique-k+2}. Recall that we are under the assumption that $G$ is $(k + 1)$-closed. 
\begin{rrule}
\label{rule:deletion-clique-or-indset}
    For each non-clique connected component $H$ of $G$ of size at least $R_{k + 1}(k + 2, k + 2)$, we run the algorithm of Lemma~\ref{lem:clique-or-indset} on $H$ with $c = k + 1$ and $a = b = k + 2$. If the algorithm returns an independent set of size $k + 2$, then we return a trivial no-instance. And if the algorithm returns a maximal clique, say $Q$, of size at least $k + 2$, then we do as follows. Let $\ell$ be the number of edges in $G$ that have exactly one endpoint in $V(Q)$. If $\ell > k$, then we return a trivial no-instance; otherwise, we delete all the edges from $G$ that have exactly one endpoint in $V(Q)$ and decrement $k$ by $\ell$. 
\end{rrule}

Assume from now on that Reduction Rule~\ref{rule:deletion-clique-or-indset} is no longer applicable. We can immediately derive the following bound for the number of vertices in $G$ that belong to non-clique components. 
\begin{observation}
\label{obs:deletion-non-clique}
    Every non-clique component of $G$ has size at most $R_{k + 1}(k + 2, k + 2) - 1 = \cO(k^3)$. Since Reduction Rule~\ref{rule:deletion-sanity-check-2} is no longer applicable, the number of non-clique components is at most $k$. Thus the number of vertices in $G$ that belong to non-clique components is $\cO(k^4)$.\footnote{Observe that for the standard {\sc Cluster Deletion} problem, i.e., deletion version without the balance constraint, these arguments are sufficient to yield a kernel with $\cO(k^4)$ vertices, as we can safely delete all connected components of the input graph that are cliques.} 
\end{observation}

We have thus bounded the number of vertices that belong to components that are not cliques. We now bound the number of vertices that belong to components that are cliques. To that end, we classify such components into two types as follows. Consider a connected component $H$ of $G$. We say that $H$ is \emph{manageable} if $H$ is a clique and $\card{H} \leq k + 1$. And we say that $H$ is \emph{unmanageable} if $H$ is a clique and $\card{H} > k + 1$. 
Before we bound the number of vertices that belong to these components, we prove the following lemma, which says that no solution for $(G, k, \eta)$ modifies an unmanageable component; this fact is essentially a corollary to Lemma~\ref{lem:deletion-clique-k+2}. 

\begin{lemma}
\label{lem:deletion-unmanageable}
Assume that $(G, k, \eta)$ is a yes-instance of \bcd. Consider any solution $F \subseteq E(G)$ for $(G, k , \eta)$.  Let $H$ be an unmanageable component of $G$. Then $F$ does not modify $H$; that is, $H$ is a connected component of $G - F$. 
\end{lemma}

\begin{proof}
First, by the definition of an unmanageable component, $H$ is a clique of size at least $k + 2$; and since $H$ is a connected component of $G$, $H$ is a maximal clique in $G$. By Lemma~\ref{lem:deletion-clique-k+2}, $H$ is a component of $G - F$. As $H$ is a component of both $G$ and $G - F$, we can conclude that $F$ does not modify $H$. 
\end{proof}

We now bound the size of the components. First, to deal with manageable components, we introduce the following reduction rule. 

\begin{rrule}   
\label{rule:deletion-small-cliques}
For $j \in [k + 1]$, if $G$ has at least $k + 2$ manageable components of size exactly $j$, then we delete one such component. 
\end{rrule}

\begin{lemma}
Reduction Rule~\ref{rule:deletion-small-cliques} is safe. 
\end{lemma}
\begin{proof}
Let $(G', k, \eta)$ be the instance obtained from $(G, k, \eta)$ by a single application of Reduction Rule~\ref{rule:deletion-small-cliques}. Let $H$ be the connected component of $G$ that we deleted from $G$ to obtain $G'$. Then $H$ is a clique and $\card{H} = j$ for some $j \in [k +2]$; and $G$ contains at least $k + 1$ other manageable components of size exactly $j$. 

Assume that $(G, k, \eta)$ is a yes-instance, and let $F \subseteq E(G)$ be a minimal solution for $(G, k, \eta)$. We first prove the following claim, which says that $F$ does not modify $H$. 
\begin{claim}
We have $F \cap E(H) = \emptyset$. 
\end{claim}
\begin{claimproof}
Suppose for a contradiction that $F \cap E(H) \neq \emptyset$. We will show that the set $F_H = F \setminus E(H)$ is also a solution for $(G, k, \eta)$, which will contradict the minimality of $F$. Observe first that $G - F_H$ is a cluster graph. To see this, notice that $H$ is a connected component of $G - F_H$, and every other connected component of $G - F_H$ is also a connected component of $G - F$. As $H$ is a clique and $G - F$ is a cluster graph, we can conclude that $G - F_H$ is a cluster graph. Now, to complete the proof of the claim, we only need to prove that $G - F_H$ is $\eta$-balanced. And for this, we will prove that $\lcomp(G - F_H) \leq \lcomp(G - F)$ and $\scomp(G - F_H) \geq \scomp(G - F)$, which will imply that $\lcomp(G - F_H) - \scomp(G - F_H) \leq \lcomp(G - F) - \scomp(G - F) \leq \eta$. Again, as $H$ is a component of $G - F_H$ and every other component of $G - F_H$ is also a component of $G - F$, we have $\lcomp(G - F_H) \leq \max\set{\card{H}, \lcomp(G - F)}$. Notice that since $\card{F} \leq k$ and since $G$ contains at least $k + 2$ manageable components of size exactly $j = \card{H}$, there exists a component $H'$ of $G$ such that $H' \neq H$,  $H'$ is a manageable component of size exactly $j$ and $F$ does not modify $H'$. Thus $H'$ is a component of $G - F$. Therefore, $\card{H'} \leq   \lcomp(G - F)$, which, along with the fact that $\card{H} = \card{H'} = j$, implies that $\lcomp(G - F_H) = \max\set{\card{H}, \lcomp(G - F)} = \lcomp(G - F)$. Now, to see that $\scomp(G - F_H) \geq \scomp(G - F)$, we consider two cases depending on whether or not $H$ is a smallest component of $G - F_H$. If $H$ is a smallest component of $G - F_H$, then $\scomp(G - F_H) = \card{H} = \card{H'} \geq \scomp(G - F)$; the last inequality follows from the fact that $H'$ is a component of $G - F$. On the other hand, if $H$ is not a smallest component of $G - F_H$, then, as every component of $G - F_H$, and in particular a smallest component of $G - F_H$, is also a component of $G - F$, we trivially have $\scomp(G - F_H) \geq \scomp(G - F)$. We thus have $\lcomp(G - F_H) - \scomp(G - F_H) \leq \lcomp(G - F) - \scomp(G - F) \leq \eta$, and hence $G - F_H$ is $\eta$-balanced, which contradicts the assumption that $F$ is a minimal solution for $(G, k , \eta)$.   
\end{claimproof}
We now show that $F$ is a solution for $(G', k, \eta)$. As $F$ does not modify $H$, we indeed have $F \subseteq E(H')$. Notice that each component of $G' - F$ is also a component of  $G - F$, which is an $\eta$-balanced cluster graph (in addition, $G - F$ contains the component $H$). 
Therefore, $G' - F$ is a cluster graph, and we have $\lcomp(G' - F) \leq \lcomp(G - F)$ and $\scomp(G' - F) \geq \scomp(G - F)$. We thus have $\lcomp(G' - F) - \scomp(G' - F) \leq \lcomp(G - F) - \scomp(G - F) \leq \eta$, which shows that $G' - F$ is $\eta$-balanced. We have thus shown that $(G', k, \eta)$ is a yes-instance. 

Assume now that $(G', k, \eta)$ is a yes-instance of \bcd, and let $F' \subseteq E(G')$ be a solution for $(G', k ,\eta)$. We claim that $F'$ is a solution for $(G, k, \eta)$ as well. Consider the graph $G - F'$. Notice that the only difference between the graphs $G - F'$ and $G' - F'$ is that $G - F'$ contains the component $H$ whereas $G' - F'$ does not. All the other components of $G - F'$ are also components of $G' - F'$. As $H$ is a clique and $G' - F'$ is a cluster graph, we can conclude that $G - F'$ is also a cluster graph. We now prove that $G' - F'$ is $\eta$-balanced. To prove this, notice that we only need to prove that $H$ is not an $\eta$-blocker in $G - F'$. Recall that $G$ contains at least $k + 2$ components of size exactly $j = \card{H}$. Therefore, $G'$ contains at least $k + 1$ components of size exactly $j$. Then, as $\card{F'} \leq k$, there exists a component $H''$ of $G'$ such that $\card{H''} = j$ and $F'$ does not modify $H''$. That is, $H''$ is a component of $G' - F'$, and therefore a component of $G - F'$. Since $H''$ is not an $\eta$-blocker in $G' - F'$ and since $\card{H''} = \card{H} = j$, we can conclude that $H''$ is not an $\eta$-blocker in $G' - F'$. But then, by Observation~\ref{obs:blocker-balanced}, $H$ is not an $\eta$-blocker in $G - F'$. This completes the proof. 
\end{proof}

\begin{observation}
\label{obs:deletion-small-cliques}
After an exhaustive application of Reduction Rule~\ref{rule:deletion-small-cliques}, for each $j \in [k + 1]$, $G$ has at most $k + 1$ manageable components of size exactly $j$. Hence the number of manageable components of $G$ is at most $(k + 1)^2 = \cO(k^2)$, and the number of vertices of $G$ that belong to manageable components is at most $\sum_{j = 1}^{k + 1} (k + 1) j = (k + 1) \sum_{j = 1}^{k + 1} j = (1/2)(k + 1)^2 (k + 2) = \cO(k^3)$.  
\end{observation}

We assume from now on that Reduction Rule~\ref{rule:deletion-small-cliques} is no longer applicable. We have thus bounded the number of vertices of $G$ that belong to non-clique components (by $\cO(k^4)$; Observation~\ref{obs:deletion-non-clique}) or manageable components (by $\cO(k^3)$; Observation~\ref{obs:deletion-small-cliques}). If $G$ has no unmanageable component, then we already have $\card{G} = \cO(k^4)$, and thus the instance $(G, k, \eta)$ is the required kernel.  (Recall that as Reduction Rule~\ref{rule:deletion-sanity-check-2} is not applicable, we have $\eta \leq \lcomp(G) \leq \card{G}$.) This observation immediately leads to the following reduction rule. 
\begin{rrule}
\label{rule:deletion-no-unmanageable}
If $G$ has no unmanageable component, then we simply return the instance $(G, k, \eta)$. 
\end{rrule}

Assume from now on that Reduction Rule~\ref{rule:deletion-no-unmanageable} is not applicable. So $G$ has at least one unmanageable component. Let $r \geq 1$ be the number of unmanageable components of $G$, and {\bf let $\bm{H_1, H_2, \ldots, H_r}$ be an ordering of the unmanageable components of $\bm{G}$ such that $\bm{\card{H_1} \leq \card{H_2} \leq \cdots \leq \card{H_r}}$.} 
Recall that to bound $\card{G}$, now we only need to bound the number of vertices that belong to unmanageable components. To that end, we first  introduce the following reduction rule, which rules out an obvious no-instance. 

\begin{rrule}
\label{rule:deletion-large-cliques-1}
If $\card{H_r}- \card{H_{1}} > \eta$, then we return a trivial no-instance. 
\end{rrule}

\begin{lemma}
Reduction Rule~\ref{rule:deletion-large-cliques-1} is safe. 
\end{lemma}

\begin{proof}
To prove the lemma, it is enough to prove that $(G, k, \eta)$ is a no-instance if $\card{H_r} - \card{H_1} > \eta$. So suppose that $\card{H_r} - \card{H_1} > \eta$ and assume for a contradiction that $(G, k , \eta)$ is a yes-instance. Let $F$ be a solution for $(G, k, \eta)$. Then $G - F$ is $\eta$-balanced. 
Since $H_1$ and $H_r$ are unmanageable components, they are both cliques of size at least $k + 2$; since they are components of $G$, $H_1$ and $H_r$ are indeed maximal cliques in $G$. Then, by Lemma~\ref{lem:deletion-unmanageable}, $F$ does not modify $H_1$ or $H_r$. That is, $H_1$ and $H_r$ are components of $G - F$. But this is not possible as $G - F$ is $\eta$-balanced. 
\end{proof}

From now on, we assume that Reduction Rule~\ref{rule:deletion-large-cliques-1} is not applicable. We now bound the number of unmanageable components. To that end, 
{\bf we define $\bm{s \in \set{0, 1}}$ as follows: $\bm{s = 0}$ if $\bm{\card{H_1} > R_{k + 1}(k + 2, k + 2)}$, and $\bm{s = 1}$ otherwise.} 

\begin{rrule}
\label{rule:deletion-large-cliques-2}
If $s + 1 < r$, then we delete the components $H_{s + 1}, H_{s + 2},\ldots, H_{r - 1}$.  
\end{rrule}

\begin{lemma}
\label{lem:deletion-large-cliques-2}
Reduction Rule~\ref{rule:deletion-large-cliques-2} is safe. 
\end{lemma}

\begin{proof}
Assume that $s + 1 < r$. Let $(G', k, \eta)$ be the instance obtained from $(G, k, \eta)$ by applying Reduction Rule~\ref{rule:deletion-large-cliques-2}. 

Assume that $(G, k, \eta)$ is a yes-instance, and let $F \subseteq E(G)$ be a solution for $(G, k, \eta)$. Recall that the components $H_{s + 1}$, $H_{s + 2},\ldots, H_{r - 1}$ are all unmanageable components. 
By Lemma~\ref{lem:deletion-unmanageable}, $F$ does not modify the components $H_{s + 1}, H_{s + 2},\ldots, H_{r - 1}$. Thus $F \subseteq E(G')$. Notice now that each component of $G' - F$ is also a component of $G - F$. Thus, $\lcomp(G' - F) \leq \lcomp(G - F)$ and $\scomp(G' - F) \geq \scomp(G - F)$. We thus have $\lcomp(G' - F) - \scomp(G' - F) \leq \lcomp(G - F) - \scomp(G' - F) \leq \eta$. That is, $G' - F$ is $\eta$-balanced, and hence $(G', k, \eta)$ is a yes-instance. 

Assume now that $(G', k, \eta)$ is a yes-instance, and let $F' \subseteq E(G')$ be a solution for $(G', k, \eta)$. We claim that $F'$ is a solution for $(G, k, \eta)$ as well. 
Notice that each component of $G' - F'$ is also a component of $G - F'$; in addition, $G - F'$ contains the components $H_{s + 1}, H_{s + 2},\ldots, H_{r - 1}$, which are all cliques.  Thus $G - F'$ is a cluster graph. So we only need to prove that $G - F'$ is $\eta$-balanced. And for  that, as $G' - F'$ is $\eta$-balanced, we only need to prove that for every $j \in [s + 1, r - 1]$, the component $H_j$ is not an $\eta$-blocker in $G - F'$; that is, $\card{\card{H_j} - \card{H'}} \leq \eta$ for every connected component $H'$ of $G - F'$. To prove this, we will consider several cases  below. Fix $j \in [s + 1, r - 1]$. Consider $H_j$ and any other connected component $H'$ of $G - F'$. 

Before we proceed further, we first highlight two arguments that we will repeatedly use in the following case analysis. {\bf (A1)} In several cases that we  consider below, we will show that $\scomp(G' - F') \leq \card{H_j}, \card{H'} \leq \lcomp(G' - F')$, which will imply that $\card{\card{H_j} - \card{H'}} \leq \lcomp(G' - F') - \scomp(G' - F') \leq \eta$; the last inequality follows from the fact that $G' - F'$ is $\eta$-balanced. {\bf (A2)} Notice that $H_r$ is a component of $G'$. In particular, $H_r$ is an unmanageable component of $G'$. Hence, by Lemma~\ref{lem:deletion-unmanageable}, $F'$ does not modify $H_r$. Thus $H_r$ is a component of $G' - F'$, and hence $\card{H_r} \leq \lcomp(G' - F')$. And since $j < r$, we have $\card{H_j} \leq \card{H_r} \leq \lcomp(G' - F')$. 

Suppose first that $F'$ does not modify $H'$; that is, $H'$ is a connected component of $G$, and in particular, $H'$ is a clique in $G$. There are two possible cases here depending on whether $H'$ is a manageable component or an unmanageable component. 
\begin{description}
\item[Case 1:] $H'$ is an unmanageable component of $G$. That is, $H' = H_i$ for some $i \in [r]$. Recall that we have $\card{H_1} \leq \card{H_2} \leq \cdots \card{H_r}$. Then, since Reduction Rule~\ref{rule:deletion-large-cliques-1} is not applicable, we have $\card{\card{H_j} - \card{H_i}} \leq \card{H_r} - \card{H_1} \leq \eta$. 

\item[Case 2:] $H'$ is a manageable component of $G$. Hence $H' \neq H_i$ for any $i \in [r]$, and in particular, $H' \neq H_i$ for any $i \in [s + 1, r - 1]$, and therefore, $H'$ is a component of $G'$. And since $F'$ does not modify $H'$, we can conclude that $H'$ is a component of $G' - F'$. Thus, $\scomp(G' - F') \leq \card{H'}$. %
Recall now that $H_j$ is an unmanageable component of $G$. By the definitions of manageable and unmanageable components, we have $ \card{H'} \leq \card{H_j}$. Recall also that $\card{H_j} \leq \card{H_r}$ as $j < r$. We thus have $\scomp(G' - F') \leq \card{H'} \leq \card{H_j} \leq \card{H_r} \leq \lcomp(G' - F')$; the last inequality follows from argument (A2) that we discussed above. We thus have $\scomp(G' - F') \leq \card{H'} \leq \card{H_j} \leq \lcomp(G' - F')$, and thus by argument (A1) that we discussed above, we have $\card{H_j} - \card{H'} \leq \lcomp(G' - F') - \scomp(G' - F') \leq \eta$. %
\end{description}

Suppose now that $F$ does modify $H'$. Let $H$ be the connected component of $G$ that contains $H'$. Notice that $H \neq H_i$ for any $i \in [s + 1, r - 1]$; and in particular, $H$ is a component of $G'$. Notice also that $F'$ modifies $H$ as $H$ contains $H'$. We again split the proof into two cases depending on whether or not $H$ is a clique. 
\begin{description}
\item[Case 1:] $H$ is a clique. Then, as $H$ is a component of $G'$, and  since $F'$ modifies $H$, Lemma~\ref{lem:deletion-unmanageable} implies that $H$ is a manageable component of $G'$ (and hence of $G$). We thus have $\card{H} \leq \card{H_j}$. As $\card{H'} \leq \card{H}$, we thus have $\card{H'} \leq \card{H_j}$. By argument (A2), we have $\card{H_j} \leq \card{H_r} \leq  \lcomp(G' - F')$. %
Finally, as $H'$ is a component of $G' - F'$, we have $\scomp(G' - F') \leq \card{H'}$.  Putting all these together, we have $\scomp(G' - F') \leq \card{H'} \leq  \card{H_j} \leq \lcomp(G' - F')$, which by argument (A1), implies that $\card{H_j} - \card{H'} \leq \eta$. 

\item[Case 2:] $H$ is not a clique. Then, by Observation~\ref{obs:deletion-non-clique}, $\card{H} \leq R_{k + 1}(k + 2, k+ 2) - 1$. Since $H'$, which is also a component of $G' - F'$, is contained in $H$, we have $\scomp(G' - F') \leq \card{H'} \leq \card{H} \leq R_{k + 1}(k + 2, k + 2) - 1$. We further split our analysis into two cases depending on whether $s = 0$ or $s = 1$. 
\begin{enumerate}
\item Suppose that $s = 0$. Then, by the definition of $s$, we have $\card{H_1} > R_{k + 1}(k + 2, k + 2)$, which implies that $\lcomp(G' - F') \geq \card{H_r} \geq \card{H_j} \geq \card{H_1} > R_{k + 1}(k + 2, k + 2)$; the first inequality follows from argument (A1). We thus have $\scomp(G' - F') \leq \card{H'} < R_{k + 1}(k + 2, k + 2) < \card{H_j} \leq \lcomp(G' - F')$, which implies that $\card{H_j} - \card{H'} \leq \lcomp(G' - F') - \scomp(G' - F') \leq \eta$. 

\item Suppose that $s =1$. We will first argue  that $\scomp(G' - F') \leq \card{H_j} \leq \lcomp(G' - F')$. Since $s = 1$, $H_1$ is a component of $G'$, and in fact, $H_1$ is an unmanageable component of $G'$. Hence, by Lemma~\ref{lem:deletion-unmanageable}, $F'$ does not modify $H_1$, and thus $H_1$ is a component of $G' - F'$. Therefore, $\scomp(G' - F') \leq \card{H_1}$. As $\card{H_1} \leq \card{H_j}$, we thus have $\scomp(G' - F') \leq \card{H_1} \leq \card{H_j}$. By argument (A2), we also have $ \card{H_r} \leq \lcomp(G' - F')$, and since $\card{H_j} \leq \card{H_r}$, we have $\card{H_j} \leq \card{H_r} \leq \lcomp(G' - F')$. We have thus shown that $\scomp(G' - F') \leq \card{H_j} \leq \lcomp(G' - F')$. Let us now prove that $\scomp(G' - F') \leq \card{H'} \leq \lcomp(G' - F')$. And to prove this, we only need to argue that $H'$ is a component of $G' - F'$. Recall that $H'$ is a component of $G - F'$, and that $H$ is the component of $G$ that contains $H'$. Now, since $H$ is not a clique, $H \neq H_i$ for any $i \in [s + 1, r - 1]$, and therefore, $H$ is a component of $G'$;  this implies that $H'$ is a component of $G' - F'$, and therefore, $\scomp(G' - F') \leq \card{H'} \leq \lcomp(G' - F')$. Since $\scomp(G' - F') \leq \card{H_j}, \card{H'} \leq \lcomp(G' - F')$, we have $\card{\card{H_j} - \card{H'}} \leq \lcomp(G' - F') - \scomp(G' - F') \leq \eta$. 
\end{enumerate}
\end{description}

We have thus shown that $G - F'$ is $\eta$-balanced. Therefore, $F'$ is a solution for $(G, k, \eta)$, and thus, $(G, k, \eta)$ is a yes-instance. 
\end{proof}
Assume from now on that Reduction Rules~\ref{rule:deletion-sanity-check}-\ref{rule:deletion-large-cliques-2} are no longer applicable. Thus, $G$ consists of non-clique components (at most $k$ such components, with at most $R_{k + 1}(k + 2, k + 2) - 1 = \cO(k^3)$ vertices in each of them;  Observation~\ref{obs:deletion-non-clique}), manageable components (at most $\cO(k^2)$ such components, with at most $k + 1$ vertices in each of them; Observation~\ref{obs:deletion-small-cliques}), the component $H_r$, and possibly the component $H_1$. Notice that $G$ contains $H_1$ only if $s = 1$, in which case we have $\card{H_1} \leq R_{k + 1}(k + 2, k + 2)$. To summarise, all components of $G$, except $H_r$, have size at most $R_{k + 1}(k +2, k +2)$, and in particular, {\bf the number of vertices of $\bm{G}$ that belong to components other than $\bm{H_r}$ is bounded by $\bm{\cO(k^4)}$. Thus $\bm{\card{G} = \card{H_r} + \cO(k^4)}$.} Hence, to bound $\card{G}$, we now need to bound only $\card{H_r}$. %
And to bound $\card{H_r}$, we first prove the following lemma, which says that $\card{H_r}$ cannot exceed $R_{k + 1}(k + 2, k + 2) + \eta$ if $(G, k, \eta)$ were indeed a yes-instance. 

\begin{lemma}
\label{lem:deletion-bound-Hr-1}
If $(G, k, \eta)$ is a yes-instance, then $\card{H_r} \leq R_{k + 1}(k + 2, k +2) + \eta$. 
\end{lemma}

\begin{proof}
Assume that $(G, k, \eta)$ is a yes-instance, and let $F$ be a solution for $(G, k, \eta)$. As $H_r$ is an unmanageable component, by Lemma~\ref{lem:deletion-unmanageable}, $F$ does not modify $H_r$, and hence $H_r$ is a component of $G - F$. Thus $\card{H_r} \leq \lcomp(G - F)$. Now, since Reduction Rule~\ref{rule:deletion-sanity-check} is not applicable, $G$ is not an $\eta$-balanced cluster graph. Thus, either $G$ contains a non-clique component or $G$ is a cluster graph, but not $\eta$-balanced. In either case $G$ contains a component $\hat H$ such that $\hat H \neq H_r$, and hence $\card{\hat H} \leq R_{k + 1}(k + 2, k + 2)$. Thus $\scomp(G) \leq \card{\hat H} \leq R_{k + 1}(k + 2, k + 2)$. Now, since $G - F$ is a spanning subgraph of $G$, we have $\scomp(G - F) \leq \scomp(G) \leq R_{k + 1}(k + 2, k + 2)$. 
Finally, since $G - F$ is $\eta$-balanced, we have $\lcomp(G - F) - \scomp(G - F) \leq \eta$, which implies that $\lcomp(G - F) \leq \scomp(G - F) + \eta$. We thus get $\card{H_r} \leq \lcomp(G - F) \leq \scomp(G - F) + \eta \leq R_{k + 1}(k  + 2, k + 2) + \eta$. 
\end{proof}

Lemma~\ref{lem:deletion-bound-Hr-1} immediately yields the following reduction rule. 

\begin{rrule}
\label{rule:deletion-bound-Hr-1}
If $\card{H_r} > R_{k + 1}(k + 2, k + 2) + \eta$, then we return a trivial no-instance. 
\end{rrule}

Assume from now on that Reduction Rule~\ref{rule:deletion-bound-Hr-1} is not applicable. Hence $\card{H_r} \leq R_{k + 1}(k + 2, k + 2) + \eta$. We now observe the following facts. 

\begin{observation}
\label{obs:deletion-bound-eta-Hr}
\begin{enumerate}
    \item\label{item:deletion-eta} Recall that since Reduction Rule~\ref{rule:deletion-bound-eta} is not applicable, we have $\eta \leq \lcomp(G) \leq \max\set{\card{H_r}, R_{k + 1}(k +2, k + 2)}$. 

    \item\label{item:deletion-Hr} Recall also that since Reduction Rule~\ref{rule:deletion-bound-Hr-1} is not applicable, we have $\card{H_r} \leq R_{k + 1}(k + 2, k + 2) + \eta$. 
    \item If $\card{H_r} \leq 2 R_{k + 1}(k +2, k + 2) = \cO(k^3)$, then $\card{G}$ is bounded by $\cO(k^4)$, and thus $(G, k, \eta)$ is a kernel with $\cO(k^4)$ vertices. 
    
    \item If $\eta \leq 2 R_{k + 1}(k + 2, k + 2)$, then since Reduction Rule~\ref{rule:deletion-bound-Hr-1} is not applicable, we have $\card{H_r} \leq R_{k + 1}(k + 2, k + 2) + \eta \leq 3 R_{k + 1}(k + 2, k + 2) = \cO(k^3)$. Thus, $(G, k, \eta)$ is a kernel with $\cO(k^4)$ vertices in this case as well.  
\end{enumerate}
\end{observation}

In light of Observation~\ref{obs:deletion-bound-eta-Hr}, we apply the following reduction rule. 

\begin{rrule}
\label{rule:deletion-bound-Hr-2}
If $\card{H_r} \leq 2 R_{k + 1}(k + 2, k + 2)$ or $\eta \leq 2 R_{k + 1}(k + 2, k + 2)$, then we simply return the instance $(G, k, \eta)$. 
\end{rrule}

We assume from now on that Reduction Rule~\ref{rule:deletion-bound-Hr-2} is no longer applicable. Hence $\card{H_r} > 2 R_{k + 1}(k + 2, k + 2)$ and $\eta > 2 R_{k + 1}(k + 2, k + 2)$. Then, as every component of $G$ except $H_r$ has size at most $R_{k + 1}(k + 2, k + 2)$, we can conclude that $H_r$ is the unique largest component of $G$, and thus $\card{H_r} = \lcomp(G)$. {\bf Let $\bm{\eta' = 2 R_{k + 1}(k + 2, k + 2)}$ and $\bm{N = \lcomp(G) - (\eta - \eta')}$.} 

\begin{observation}
\label{obs:deletion-bound-N}
\begin{enumerate}
\item\label{item:deletion-N-lb} Recall that we have $\eta \leq \lcomp(G)$ (Observation~\ref{obs:deletion-bound-eta-Hr}-\ref{item:deletion-eta}). Therefore, $N = \lcomp(G) - (\eta - \eta') \geq \eta - (\eta - \eta') = \eta' = 2 R_{k + 1}(k + 2, k + 2)$. 

\item\label{item:deletion-N-ub} Recall also that we have $\lcomp(G) = \card{H_r} \leq R_{k + 1}(k + 2, k + 2) + \eta$ (Observation~\ref{obs:deletion-bound-eta-Hr}-\ref{item:deletion-Hr}). Therefore, $N = \lcomp(G) - (\eta - \eta') \leq R_{k + 1}(k + 2, k + 2) + \eta - (\eta - \eta') = R_{k + 1}(k + 2, k + 2) + \eta' = 3 R_{k + 1}(k + 2, k + 2)$. 
\end{enumerate}
\end{observation}

We now apply the following reduction rule \emph{once}. 

\begin{rrule}
\label{rule:deletion-final-rule-for-bcd}
 We delete $\lcomp(G) - N$ (arbitrarily chosen) vertices from $H_r$; let us denote the resulting component by $H'$ and the resulting graph by $G'$. We return the instance $(G', k, \eta')$.
\end{rrule}

\begin{lemma}
Reduction Rule~\ref{rule:deletion-final-rule-for-bcd} is safe. 
\end{lemma}

\begin{proof}
Informally, Reduction Rule~\ref{rule:deletion-final-rule-for-bcd} is safe because we deleted $\lcomp(G) - N = \eta - \eta'$ vertices from $H_r$ and we decremented $\eta$ by the same amount; that is, we set $\eta' = \eta - (\eta - \eta')$. The safeness of the rule then follows from the facts that no solution for $(G, k, \eta)$ modifies $H_r$ as $H_r$ is an unmanageable component of $G$, and no solution for $(G', k, \eta')$ modifies $H'$ as $\card{H'} = \lcomp(G) - (\lcomp(G) - N) = N \geq 2 R_{k + 1}(k + 2, k + 2) \geq k + 2$, and thus $H'$ is an unmanageable component of $G'$. We now prove this more formally. 

Recall that each component of $G$ except $H_r$ is a component of $G'$, and each component of $G'$  except $H'$ is a component of $G$. Also, each such  component $H \neq H_r, H'$ (of $G$ or $G'$) has size at most $R_{k + 1}(k + 2, k + 2)$. But $\card{H_r} > 2 R_{k + 1}(k + 2)$ and $\card{H'} = N \geq 2 R_{k + 1}(k + 2)$, and thus $H_r$ is the unique largest component of $G$ and $H'$ is the unique largest component of $G'$. That is, $\lcomp(G) = \card{H_r}$ and $\lcomp(G') = \card{H'}$. 

We now prove that the instances $(G, k, \eta)$ and $(G', k, \eta')$ are equivalent. Assume first that $(G, k, \eta)$ is a yes-instance, and let $F$ be a solution for $(G, k, \eta)$. We argue that $F$ is a solution for $(G', k, \eta')$ as well. First, since $H_r$ is an unmanageable component of $G$, by Lemma~\ref{lem:deletion-unmanageable}, $F$ does not modify $H_r$. Now, consider the graph $G' - F$; this graph is well-defined as $F$ does not modify $H_r$, and and it is indeed a cluster graph.  In particular, $F$ does not modify $H'$ (as $V(H') \subseteq V(H_r)$). Thus, $F$ does not modify $H_r$ or $H'$, and hence $H_r$ is a component of $G - F$ and $H'$ is a component of $G' - F$. Thus $\card{H_r} \leq \lcomp(G - F)$ and $\card{H'} \leq \lcomp(G' - F)$. Now, as $G - F$ is a subgraph of $G$, we have $\lcomp(G - F) \leq \lcomp(G)$, which implies that $\card{H_r} \leq \lcomp(G - F) \leq \lcomp(G) = \card{H_r}$, and thus $\card{H_r} = \lcomp(G - F) = \lcomp(G)$. Using similar arguments, we get  $N = \card{H'} = \lcomp(G' - F) = \lcomp(G')$. Observe now that since Reduction Rule~\ref{rule:deletion-sanity-check} is not applicable, $G$ contains a component other than $H_r$, and hence $G'$ contains a component other than $H'$. Again, since $F$ does not modify $H_r$, each component of $G - F$, except $H_r$, is also a component of $G' - F$; and each component of $G' - F$, except $H'$, is also a component of $G - F$. In particular, a smallest component of $G - F$ is a smallest component of $G' - F$, and vice versa. That is, $\scomp(G - F) = \scomp(G' - F)$. 
Since $G - F$ is $\eta$-balanced, we have 
$\lcomp(G - F) - \scomp(G - F) \leq \eta$, which implies that $\lcomp(G - F) - (\eta - \eta') - \scomp(G - F) \leq \eta - (\eta - \eta')$, which implies that $N - \scomp(G' - F) \leq \eta'$. That is, $\lcomp(G' - F) - \scomp(G' - F) \leq \eta'$, and thus $G' - F$ is $\eta$-balanced. 

Conversely, assume that $(G', k', \eta')$ is a yes-instance, and let $F'$ be a solution for $(G', k, \eta')$. Then, as $H'$ is an unmanageable component of $G'$, $F'$ does not modify $H'$, and  we can argue that $F'$ is a solution for $(G, k, \eta)$. In fact, we have (i) $\lcomp(G - F') = \card{H_r} = \lcomp(G)$, (ii) $\lcomp(G' - F') = \card{H'} = N = \lcomp(G')$ and (iii) $\scomp(G - F') = \scomp(G' - F')$. Since $G' - F'$ is $\eta'$-balanced, we have $\lcomp(G' - F') - \scomp(G' - F') \leq \eta'$, which implies that $N - 
\scomp(G' - F') \leq \eta'$, which implies that $N + (\eta - \eta') - \lcomp(G' - F') \leq \eta' + (\eta - \eta')$, which implies that $\lcomp(G - F') - \scomp(G - F') \leq \eta$. Thus $G - F'$ is $\eta$-balanced. 
\end{proof}

Consider the instance $(G', k, \eta')$ returned by Reduction Rule~\ref{rule:deletion-final-rule-for-bcd}. Recall that $\eta' = 2 R_{k + 1}(k + 2, k + 2)$ and $\card{H'} = N \leq 3 R_{k + 1}(k + 2, k + 2) = \cO(k^3)$  (Observation~\ref{obs:deletion-bound-N}). The components of $G'$ are precisely those components of $G$, except $H_r$, and $H'$. As $\card{G'} = \card{H_r} + \cO(k^4)$, we have $\card{G'} = \card{H'} + \cO(k^4) = \cO(k^4)$. That is, $(G', k, \eta')$ is a kernel with $\cO(k^4)$ vertices. We have thus proved the following result. 

\begin{theorem}
\label{thm:bcd-par-solution-size}
\bcd\ admits a kernel with $\cO(k^4)$ vertices. 
\end{theorem}

\subsection{Polynomial Kernel for Balanced Cluster Editing}\label{sec:editing-kernel}

We formally define the \bcefull\ (\bce) problem as follows. 

\defproblem{\bcefull\ (\bce)}{A graph $G$ and non-negative integers $k$ and $\eta$.}{Decide if there exists $F \subseteq E(G)$ such that $\card{F} \leq k$ and $G \triangle F$ is an $\eta$-balanced cluster graph.}

In this section, we show that \bce\ admits a kernel with $\cO(k^3)$ vertices. We first briefly outline our strategy. 

\subparagraph*{Outline of the kernel.} Consider an instance $(G, k, \eta)$ of \bce. We first construct a modulator---a set $S$ of vertices such that $G - S$ is a cluster graph. If $(G, k, \eta)$ is a yes-instance then a modulator of size $\cO(k)$ exists, and we can find  such a modulator in polynomial time. We then bound the number of components and the size of each component of $G - S$. To do this, we first consider the components of $G - S$ that have at least one neighbour in $S$; we show that the number of such components is $\cO(k^2)$, and that the size of each such component is $\cO(k)$. Next, we consider the components of $G - S$ that have no neighbour in $S$; notice that these are indeed components of $G$. To bound the number of vertices that belong to such components, we use arguments that are identical to the ones we used for \bcd. That is, we classify the components into two types---small and large---based on their sizes, and apply a host of reduction rules to bound the number of vertices that belong to such components. \lipicsEnd

We now proceed to designing our kernel. To do this, we rely on the fact that a graph $G$ is a cluster graph if and only if $G$ does not contain $P_3$ as an induced subgraph, where $P_3$ is the path on three vertices. We say that an edge or non-edge $uv \in \binom{V(G)}{2}$ is part of an induced $P_3$ if there exists a vertex $w \in V(G)$ such that the subgraph of $G$ induced by $\set{u, v, w}$ is a $P_3$. Notice now that if an edge $uv \in E(G)$ is part of at least $k + 1$ distinct induced $P_3$s---that is, there exist $k + 1$ distinct vertices $w_1, w_2,\ldots, w_{k + 1}$ such that $G[\set{u, v, w_i}]$ is a $P_3$ for every $i \in [k + 1]$---then \emph{every} solution $F$ for $(G, k, \eta)$ must necessarily contain the edge $uv$; for otherwise, for every $i \in [k + 1]$, $F$ must contain either $uw_i$ or $vw_i$, which is not possible as $\card{F} \leq k$. A similar reasoning applies to non-edges $uv$ that are part of at least $k + 1$ induced $P_3$s. These observations immediately lead to the following two reduction rules, which we apply exhaustively. 
\begin{rrule}
\label{rule:editing-edge}
If an edge $uv \in E(G)$ is part of at least $k + 1$ induced $P_3$s, then delete the edge $uv$ from $G$ and decrement $k$ by $1$. 
\end{rrule}

\begin{rrule}
\label{rule:editing-non-edge}
If a non-edge edge $uv \in \binom{V(G)}{2} \setminus E(G)$ is part of at least $k + 1$ induced $P_3$s, then add the edge $uv$ to $G$ and decrement $k$ by $1$. 
\end{rrule}

We now apply the following two reduction rules. The first one rules out obvious yes-instances, and the second one bounds $\eta$. 

\begin{rrule}
\label{rule:editing-sanity-check}
If $G$ is an $\eta$-balanced cluster graph, then we return a trivial yes-instance. 
\end{rrule}
 
\begin{rrule}
\label{rule:editing-bound-eta}
If $\eta > \max \set{\lcomp(G), k}$, then we return the instance $(G, k, \hat \eta)$, where $\hat \eta = \max \set{\lcomp(G), k}$. 
\end{rrule}

\begin{lemma}
\label{lem:editing-bound-eta}
Reduction Rule~\ref{rule:editing-bound-eta} is safe. 
\end{lemma}
\begin{proof}
Suppose that $\eta > \max \set{\lcomp(G), k}$. 
Assume that $(G, k, \eta)$ is a yes-instance, and let $F \subseteq \binom{V(G)}{2}$ be a solution for $(G, k, \eta)$. We will show that $G \triangle F$ is $\hat \eta$ balanced. And for that, it is enough to show that $\lcomp(G \triangle F) \leq \max \set{\lcomp(G), k} = \hat \eta$. To prove this, consider a connected component $H$ of $G \triangle F$. We will show that $\card{H} \leq \max\set{\lcomp(G), k} = \hat \eta$. If $H$ is fully contained in a connected component of $G$, then $\card{H} \leq \lcomp(G)$. So suppose this is not the case. Then $H$ intersects $r$ distinct connected components of $G$ for some $r \geq 2$. Let $\{X_1, X_2, \ldots, X_r\}$ be the unique partition of $V(H)$ such that for each $i \in [r]$, $X_i$ is fully contained in a connected component of $G$, and for distinct $i, j \in [r]$, $X_i$ and $X_j$ are not contained in the same connected component of $G$. Then $\sum_{i \in [r]} \card{X_i} = \card{H}$. Also, for distinct $i, j \in [r]$, as $X_i$ and $X_j$ are contained in distinct connected components of $G$, $E(G)$ does not contain any edge with one endpoint in $X_i$ and the other in $X_j$. But as $H$ is a connected component of the cluster graph $G \triangle F$, $H$ is a clique, and therefore, $E(H)$ contains all the edges between $X_i$ and $X_j$. Thus, for distinct $i, j \in [r]$, we can conclude that $F$ contains all the edges between $X_i$ and $X_j$; that is, $\bigcup_{\set{i, j} \subseteq \binom{[r]}{2}} \set{u, v ~|~ u \in X_i, v \in X_j} \subseteq F$.  We thus have $\card{F} \geq \sum_{\set{i, j} \in \binom{[r]}{2}} \card{X_i} \card{X_j} \geq \sum_{i \in [r]} \card{X_i} = \card{H}$. That is, $\card{H} \leq \card{F} \leq k$. As $H$ is an arbitrary component of $G \triangle F$, we can conclude that $\lcomp(G \triangle F) \leq \max \set{\lcomp(G), k} = \hat \eta$, and hence $G \triangle F$ is $\hat \eta$-balanced. 

Recall that we are under the assumption that $\hat \eta = \max \set{\lcomp(G), k} < \eta$. Hence, for any $F' \subseteq \binom{V(G)}{2}$, if $G \triangle F'$ is $\hat \eta$-balanced, then  $G \triangle F'$ is $\eta$-balanced as well. Therefore, if  $(G, k, \hat \eta)$ is a yes-instance, then so is $(G, k, \eta)$. 
\end{proof}

Assume from now on that Reduction Rules~\ref{rule:editing-edge}-\ref{rule:editing-bound-eta} are no longer applicable. We will now use the well-known fact that given a graph $G$ and a non-negative integer $k$, in polynomial time, we can either find a set of vertices $S \subseteq V(G)$ such that $\card{S} \leq 3k$ and $G - S$ is a cluster graph, or conclude that $(G, k)$ is a no-instance of {\sc Cluster Deletion} (i.e., the deletion version without the balance constraint). If $(G, k)$ is a no-instance of {\sc Cluster Deletion}, then notice that $(G, k, \eta)$ is a no-instance of \bcd\ for any $\eta \geq 0$. We state this result below for future reference.  
\begin{lemma}[folklore]
\label{lem:editing-modulator}
There exists an algorithm that, given a graph $G$ and a non-negative integers $k$ and $\eta$ as input, runs in polynomial time, and either returns an inclusion-wise minimal set $S \subseteq V(G)$ such that $\card{S} \leq 3k$ and $G - S$ is a cluster graph, or correctly returns that $(G, k, \eta)$ is a no-instance of \bce. 
\end{lemma}

We invoke the algorithm of Lemma~\ref{lem:editing-modulator} on the instance $(G, k, \eta)$. If the algorithm returns that $(G, k, \eta)$ is a no-instance, then we return a trivial no-instance of \bce. Otherwise, the algorithm returns an inclusion-wise minimal set $S$ such that $\card{S} \leq 3k$ and $G - S$ is a cluster graph; assume from now on that the algorithm of Lemma~\ref{lem:editing-modulator} returned such a set $S$. We will use $S$ and the fact that $G - S$ is a cluster graph throughout the remainder of this section. As $G - S$ is a cluster graph, every connected component of $G - S$ is a clique. Now, to bound $\card{G}$, we will bound the number of connected components of $G - S$ and the size of each such component.   

We classify the components of $G - S$ into two types depending on whether or not a component has a neighbour in $S$. Consider a connected component $H$ of $G - S$. We say that $H$ is \emph{visible} if a vertex in $V(H)$ is adjacent to a vertex in $S$; otherwise, we say that $H$ is \emph{invisible}. If $H$ is visible, then for any vertex $x \in S$ such that $x$ is adjacent to a vertex in $V(H)$, we say that $x$ sees the component $H$. We first bound the number of visible components. 

\begin{rrule}
\label{rule:editing-number-of-visible}
If there exists $x \in S$ such that $x$ sees at least $2k+ 2$ distinct components of $G - S$, then we return a trivial no-instance. 
\end{rrule}
\label{lem:editing-number-of-visible}

\begin{lemma}
Reduction Rule~\ref{rule:editing-number-of-visible} is safe. 
\end{lemma}
\begin{proof}
To prove the lemma, it is enough to prove that $(G, k, \eta)$ is a no-instance if there exists $x \in S$ such that $x$ sees at least $2k + 2$ distinct components of $G - S$. Suppose such an $x \in S$ exists, and assume for a contradiction that $(G, k, \eta)$ is a yes-instance. Let $F$ be a solution for $(G, k, \eta)$. Then, as $\card{F} \leq k$, $F$ modifies at most $2k$ components of $G - S$.  In particular, $F$ modifies at most $2k$ components of $G - S$ that are seen by $x$. Then, since $x$ sees at least $2k + 2$ components of $G - S$, there exists two distinct components $H$ and $H'$ of $G - S$ such that $x$ sees both $H$ and $H'$ and $F$ does not modify either of those components. As $x$ sees both $H$ and $H'$,  there exist vertices $y \in V(H)$ and $y' \in V(H')$ such that $xy, xy' \in E(G)$. Also, since $H$ and $H'$ are distinct components of $G - S$, we have $yy' \notin E(G)$. Since $F$ does not modify $H$ or $H'$, we can conclude that $xy, xy' \in E(G \triangle F)$ and $yy' \notin E(G \triangle F)$. But then $yxy'$ is a $P_3$ in $G \triangle F$, which contradicts the assumption that $G \triangle F$ is a cluster graph. 
\end{proof}

Assume from now on that Reduction Rule~\ref{rule:editing-number-of-visible} is no longer applicable. Hence, for every $x \in S$, $x$ sees at most $2k + 1$ components of $G - S$. But notice that every visible component is seen by at least one vertex of $S$. We can thus conclude that the number of visible components of $G - S$ is at most $\card{S} \cdot (2k + 1) \leq 3k \cdot (2k + 1) = \cO(k^2)$. We record this observation below for future reference. 

\begin{observation}
\label{obs:editing-number-of-visible}
The number of visible components of $G - S$ is at most $\card{S} \cdot (2k + 1) = \cO(k^2)$. 
\end{observation}
We now bound the size of each visible component by $\cO(k)$. To that end, we classify the visible components into two types. Consider a visible component $H$ of $G - S$. We say that $H$ is a \emph{type-1 component} if there exists $x \in S$ such that $x$ has both a neighbour and a non-neighbour in $V(H)$; otherwise, we say that $H$ is a \emph{type-2 component}. Observe that if $H$ is a type 1 component, then there exist $y, z \in V(H)$ such that $x y z$ is an induced $P_3$. And recall that as Reduction Rules~\ref{rule:editing-edge} and \ref{rule:editing-non-edge} are not applicable, every edge and every non-edge is part of at most $k$ $P_3$s. These observations immediately lead to the following lemma. 

\begin{lemma}
\label{lem:editing-type-1}
If $H$ is a type-1 component of $G - S$, then $\card{H} \leq 2k$.
\end{lemma}
\begin{proof}
Consider a type-1 component $H$ of $G - S$. Then there exist $x \in S$ and $y, z \in V(H)$ such that $xy \in E(G)$ and $xz \notin E(G)$. Notice now that every vertex in $V(H)$ is either a neighbour or a non-neighbour of $x$. If $V(H)$ contains at least $k + 1$ neighbours of $x$, say $y_1, y_2,\ldots, y_{k + 1}$, then $x y_i z$ is an induced $P_3$ for every $i \in [k + 1]$. Thus the non-edge $xz$ is part of at least $k + 1$ induced $P_3$s, which is not possible, as Reduction Rule~\ref{rule:editing-non-edge} is not applicable. Similarly, if $V(H)$ contains at least $k + 1$ non-neighbours of $x$, then the edge $xy$ would be part of at least $k + 1$ induced $P_3$s, which is not possible either, as Reduction Rule~\ref{rule:editing-edge} is not applicable. Thus $V(H)$ contains at most $k$ neighbours and at most $k$ non-neighbours of $x$, and hence $\card{H} \leq k + k = 2k$. 
\end{proof}

Observe now that if $H$ is a type-2 component, then ($H$ is visible and) for every $x \in S$, either $x$ is adjacent to each vertex of $H$ or $x$ is not adjacent to any vertex of $H$. We use this observation and the minimality of $S$ to prove the following lemma, which bounds the size of every type-2 component. 

\begin{lemma}
\label{lem:editing-type-2}
If $H$ is a type-2 component of $G - S$, then $\card{H} \leq k$. 
\end{lemma}
\begin{proof}
Consider a type-2 component $H$ of $G - S$. Then $H$ is visible, and hence there exists a vertex $x \in S$ such that $x$ sees $H$. And as $H$ is a type-2 component, $x$ is adjacent to every vertex of $H$. In particular, $G[V(H) \cup \set{x}]$ is a clique. Observe now that $x$ must see another component $H' \neq H$ of $G - S$; otherwise, $G - (S \setminus \set{x})$ is a cluster graph, which contradicts the fact that $S$ is inclusion-wise minimal. Fix such a component $H'$ such that $H' \neq H$ and $x$ sees $H$, and fix a vertex $y' \in V(H')$ such that $xy' \in E(H)$. Notice now that for every vertex $y \in V(H)$, $y x y'$ is an induced $P_3$. Hence, if $\card{H} \geq k + 1$, then the edge $xy'$ would be part of at least $k + 1$ induced $P_3$s, which is not possible as Reduction Rule~\ref{rule:editing-edge} is not applicable. We can thus conclude that $\card{H} \leq k$. 
\end{proof}

\begin{observation}
\label{obs:editing-size-of-visible}
By Observation~\ref{obs:editing-number-of-visible}, $G - S$ has at most $\card{S} (2k + 1) \leq 3k (2k + 1)$ visible components. 
By Lemmas~\ref{lem:editing-type-1} and \ref{lem:editing-type-2}, the size of each such component is at most $2k$. Hence the number of vertices of $G - S$ that belong to the visible components is at most $3k (2k + 1) \cdot 2k = \cO(k^3)$.  
\end{observation}

We have thus bounded the number of vertices that belong to the visible components of $G - S$. Now we only need to bound the number of vertices that belong to the invisible components of $G - S$.  To do this, we follow the same steps as we did for \bcdfull. As the arguments are identical to that of \bcd, we only briefly outline them here. Recall first that each invisible component $H$ has no neighbours in $S$, and hence $H$ is a connected component of $G$. Thus $H$ is a clique and it is a connected component of $G$. We classify the invisible components into two types---manageable and unmanageable---depending on the size of the components. Consider an invisible component $H$. We say that $H$ is \emph{manageable} if $\card{H} \leq k + 1$; otherwise we say that $H$ is \emph{unmanageable}. Notice now that as $\card{F} \leq k$ for every solution $F$, $F$ modifies at most $2k$ components of $G$, and in particular $F$ modifies at most $2k$ manageable components. Thus, for any $j \in [k + 1]$, if $G$ contains more than $2k + 1$ manageable components of size exactly $j$, then we can safely delete one of them. We thus have the following reduction rule, (which is similar to Reduction Rule~\ref{rule:deletion-small-cliques} for \bcd). 
\begin{rrule}
\label{rule:editing-small-cliques}
For $j \in [k + 1]$, if $G$ has at least $2k + 2$ manageable components, each of size exactly $j$, then we delete one such component. 
\end{rrule}

\begin{observation}
\label{obs:editing-small-cliques}
After an exhaustive application of Reduction Rule~\ref{rule:editing-small-cliques}, for each $j \in [k + 1]$, $G$ has at most $2k + 1$ manageable components of size exactly $j$. Hence the number of manageable components of $G$ is at most $(k + 1)(2k + 1) = \cO(k^2)$, and the number of vertices of $G$ that belong to manageable components is at most $\sum_{j = 1}^{k + 1} (2k + 1) j = (2k + 1) \sum_{j = 1}^{k + 1} j = (1/2) (2k + 1) (k + 1) (k + 2) = \cO(k^3)$.  
\end{observation}

We have so far bounded the number of vertices of $G$ that do not belong to the unmanageable components, and now we only have to bound the number of vertices that belong to the unmanageable components. Before proceeding further, we make the following observation, which summarises the bounds we have proved thus far. This will be useful in bounding the number of vertices that belong to the unmanageable components.

\begin{observation}
\label{obs:editing-H1-bound}
Recall that we have already bounded the number of vertices that belong to the visible components of $G - S$ by $6k^2(2k + 1)$ (Observation~\ref{obs:editing-size-of-visible}); and the number of vertices that belong to the manageable components of $G$ by $(1/2) (2k + 1) (k + 1) (k + 2)$ (Observation~\ref{obs:editing-small-cliques}); and we have $\card{S} \leq 3k$ (Lemma~\ref{lem:editing-modulator}). In other words, we have bounded the number of vertices of $G$ that do not belong to the unmanageable components---and this number is at most $6k^2 (2k + 1) + (1/2) (2k + 1) (k + 1) (k + 2) + 3k \leq 10(2k + 1)^3$. 
\end{observation}

We now bound the number of vertices that belong to the unmanageable components. As in the case of \bcd, we first show that no solution for $(G, k, \eta)$ modifies an unmanageable component. 

\begin{lemma}
\label{lem:editing-unmanageable}
Assume that $(G, k, \eta)$ is a yes-instance of \bce. Consider any solution $F \subseteq \binom{V(G)}{2}$ for $(G, k , \eta)$.  Let $H$ be an unmanageable component of $G$. Then $F$ does not modify $H$; that is, $H$ is a connected component of $G \triangle F$. 
\end{lemma}
\begin{proof}
Fix a solution $F$ for $(G, k, \eta)$, and consider an unmanageable component $H$.  Consider the cluster graph $G \triangle F$. Recall that $H$ is a clique of size at least $k + 2$. And as $\card{H} \geq k + 2$, we have to delete at least $k + 1$ edges from $H$ to separate the vertices of $H$ into two or more connected components. Since $\card{F} \leq k$, and in particular $\card{F \cap E(G)} \leq k$, we can conclude that $H$ is fully contained in a connected component of $G \triangle F$. Notice now that $H$ is indeed a connected component of $G \triangle F$, for otherwise, there exists a vertex $v \in V(G) \setminus V(H)$ such that $H$ and $v$ are in the same connected component of $G \triangle F$. Then $F$, and in particular $F \cap (\binom{V(G)}{2} \setminus E(G))$, must contain all the edges between $v$ and each vertex of $H$. That is,  $F$ must contain at least $\card{H} \geq k + 2$ edges. But this is not possible as $\card{F} \leq k$. 
\end{proof}

We now only need to bound the number of vertices that belong to the unmanageable components. For this, we use arguments that are nearly identical to the ones we used for \bcd. So, we sketch the sequence of arguments without proofs in the following itemised list. 
\begin{enumerate}
    \item Let $H_1, H_2,\ldots, H_r$ be the unmanageable components of $G$ such that $\card{H_1} \leq \card{H_2} \leq \cdots \leq \card{H_r}$. 

    \item Lemma~\ref{lem:editing-unmanageable} shows that no solution modifies an unmanageable component. Hence, if $(G, k, \eta)$ is a yes-instance and $F \subseteq \binom{V(G)}{2}$ is a solution for $(G, k, \eta)$, then for every $i \in [r]$, $H_i$ is also a component of $G \triangle F$. So if $\card{H_r}- \card{H_{1}} > \eta$, then we return a trivial no-instance (identical to Reduction Rule~\ref{rule:deletion-large-cliques-1} for \bcd). 
    
    \item 
    To bound the number of unmanageable components, 
    {\bf we define $\bm{s \in \set{0, 1}}$ as follows: $\bm{s = 0}$ if $\bm{\card{H_1} > 10(2k + 1)^3}$, and $\bm{s = 1}$ otherwise.} Our choice of the bound $10(2k + 1)^3$ is based on Observation~\ref{obs:editing-H1-bound}, which says that the number of vertices that do not belong to the unmanageable components is at most $10(2k + 1)^3$. 

    \item 
    Now, if $s + 1 < r$, then we delete the components $H_{s + 1}, H_{s + 2},\ldots, H_{r - 1}$ (identical to Reduction Rule~\ref{rule:deletion-large-cliques-2} for \bcd). 
    So, to bound $\card{G}$, we now only need to bound $\card{H_r}$; in particular, we now have $\card{G} = \cO(k^3) + \card{H_r}$. 
    
    \item 
    We can show that if $(G, k, \eta)$ is a yes-instance, then $\card{H_r} \leq 10(2k + 1)^3 + \eta$. To prove this, we can argue that since Reduction Rule~\ref{rule:editing-sanity-check} is not applicable, the graph $G$ and consequently the graph $G \triangle F$ for any solution $F$ for $(G, k, \eta)$, contains a component of size at most $10(2k + 1)^3$. Hence, if $\card{H_r} > 10(2k + 1)^3 + \eta$, then we return a trivial no instance. 

    \item\label{item:editing-Hr} 
    Recall that every component of $G$, except possibly $H_r$, has size at most $10(2k + 1)^3$. 
    If $\card{H_r} \leq 20(2k + 1)^3$, then since Reduction Rule~\ref{rule:editing-bound-eta} is not applicable, we have $\eta \leq \max \set{\lcomp(G), k} \leq 20(2k + 1)^3$, and thus $(G, k, \eta)$ is a kernel with $\cO(k^3)$ vertices. 

    \item\label{item:editing-eta} 
    If $\eta \leq 20(2k + 1)^3$, then as observed above, we have $\card{H_r} \leq 10(2k + 1)^3 + \eta \leq 30(2k + 1)^3$, and thus we again have a kernel with $\cO(k^3)$ vertices. 

    \item 
    In light of items~\ref{item:editing-Hr} and \ref{item:editing-eta} above, we can assume that $\card{H_r} > 20(2k + 1)^3$ and $\eta > 20(2k + 1)^3$. Also, as every component of $G$ except $H_r$ has size at most $10(2k + 1)^3$, $H_r$ is the unique largest component of $G$, and thus $\lcomp(G) = \card{H_r}$. 
    {\bf We set $\bm{\eta' = 20(2k + 1)^3}$ and $\bm{N = \lcomp(G) - (\eta - \eta')}$.} 

    \item As Reduction Rule~\ref{rule:editing-bound-eta} is not applicable, we have $\eta \leq \max\set{\lcomp(G), k}$. As $\lcomp(G) = \card{H_r} > k$, we can conclude that $\eta \leq \lcomp(G)$. Therefore, $N = \lcomp(G) - (\eta - \eta') \geq \eta - (\eta - \eta') = \eta' = 20(2k + 1)^3$. We also have $\lcomp(G) = \card{H_r} \leq 10(2k + 1)^3 + \eta$, which implies that $N = \lcomp(G) - (\eta - \eta') \leq 10(2k + 1)^3 + \eta - (\eta - \eta') = 10(2k + 1)^3 + \eta' = 30(2k + 1)^3$. Thus, $\bm{20(2k + 1)^3 \leq N \leq 30(2k + 1)^3}$.  

    \item 
    Finally, we delete $\lcomp(G) - N$ (arbitrarily chosen) vertices from $H_r$; let us denote the resulting component by $H'$ and the resulting graph by $G'$. We return the instance $(G', k, \eta')$, which is the required kernel with $\cO(k^3)$ vertices. We thus have the following theorem. 
    
\end{enumerate}

\begin{theorem}
\label{thm:bce-par-solution-size}
\bce\ admits a kernel with $\cO(k^3)$ vertices. 
\end{theorem}

%% file: Full-version-FPT-section-intro.tex
\section{FPT Algorithms for Balanced Cluster Modification Problems}
\label{sec:fpt}

In this section, we show that \bcc, \bcd\ and \bce\ admit single-exponential \FPT\ algorithms. We first show that \bcc\ and \bcd\ admit algorithms that run in time \algobcdruntime; these results are presented in Sections~\ref{sec:branching} and \ref{sec:deletion-fpt}. We then show in Section~\ref{sec:fpt-completion-editing} that \bcc\ admits another algorithm that runs in time \algobccruntime; this algorithm relies on a correspondence between solutions for \bcc\ and integer-partitions. We also show that the same idea can be exploited to design an algorithm for \bce\ that runs in time \algobccruntime. Finally, in Section~\ref{sec:fast}, we show that by refining the arguments used in Section~\ref{sec:fpt-completion-editing}, we can design algorithms for \bcc\ and \bce\ that run in time \fastbccruntime.

%% file: Full-version-FPT-branching.tex
\subsection{FPT  Algorithm for Balanced Cluster Completion}\label{sec:branching}

We first design an algorithm for \bcc\ that runs in time \algobcdruntime. The algorithm is a branching procedure based on the following simple observation: If $(G, k, \eta)$ is a yes-instance of \bcc\ and $G$ is not $\eta$-balanced, then every smallest component $H$ of $G$ must be merged with at least one other component of size $j$ for some $j$, which will require the addition of exactly $j \card{H}$ edges. So we branch into all possible choices for $j$.\footnote{We are grateful to an anonymous reviewer of an earlier version of this manuscript, who pointed out that this simple observation would lead to an algorithm with running time \algobcdruntime.} 

\begin{lemma}\label{lem:bralgobcc-correctness}
Let $(G, k, \eta)$ be a yes-instance  of \bcc, and let $H$ be any smallest component of $G$. Then, either $G$ is $\eta$-balanced, or there exists a positive integer $\ell$ such that $\ell \card{H} \leq k$, $G$ contains at least one component $H' \neq H$ of size exactly $\ell$, and for every component $H' \neq H$ of size exactly $\ell$, $(G + F_{HH'}, k - \card{F_{HH'}}, \eta)$ is a yes-instance of \bcc, where $F_{HH'}$ is the set of all possible edges with one endpoint in $V(H)$ and the other in $V(H')$.  
\end{lemma}

\begin{proof}
    Let $F$ be a solution for $(G, k, \eta)$.  Assume that $G$ is not $\eta$-balanced. 
    Then, $F \neq \emptyset$, and in particular, $F$ must modify $H$. 
    Let $H' \neq H$ be a component of $G$ such that $F$ merges $H$ and $H'$ together; then $F_{HH'} = \set{uv \in \binom{V(G)}{2} ~|~ u \in V(H), v \in V(H')} \subseteq F$. Let let $\ell = \card{H'}$. Notice that $\card{F_{HH'}} = \ell \card{H}$, and as $F_{HH'} \subseteq F$, we have  $\ell \card{H} = \card{F_{HH'}} \leq \card{F} \leq k$. Notice also that $G + F_{HH'}$ is a cluster graph, and that $F \setminus F_{HH'}$ is a solution for the instance $(G + F_{HH'}, k - \card{F_{HH'}}, \eta)$, and hence $(G + F_{HH'}, k - \card{F_{HH'}}, \eta)$ is a yes-instance of \bcc. The lemma now follows from the fact  that for any component $H'' \neq H$ such that $\card{H''} = \card{H'} = \ell$, we may ``replace'' $H'$ with $H''$ in $F$; that is, we have $\card{F_{HH'}} = \card{F_{HH''}} = \ell \card{H}$, where $F_{HH'} = \set{uv \in \binom{V(G)}{2} ~|~ u \in V(H), v \in V(H'')}$, the cluster graphs $G + F_{HH'}$ and $G + F_{HH''}$ are isomorphic, and therefore, $(G + F_{HH'}, k - \card{F_{HH'}}, \eta)$ is a yes-instance of \bcc\ if and only if $(G + F_{HH''}, k - \card{F_{HH''}}, \eta)$ is a yes-instance of \bcc.  
\end{proof}

We are now ready to describe our algorithm, which we call \bralgobcc. 

\subparagraph*{\bralgobcc.}Given an instance $(G, k, \eta)$ of \bcc\ as input, we proceed as follows. 
Recall that $G$ is a cluster graph. 
\begin{description}
\item[Step 1.] If $G$ is  $\eta$-balanced, then we return that $(G, k, \eta)$ is a yes-instance. 

\item[Step 2.] Fix a smallest connected component, say $H$, of $G$. 

\item[Step 2a.] If there does not exist a positive integer $\ell$ such that $\ell \card{H} \leq k$ and $G$ has a component $H' \neq H$ of size exactly $\ell$, then we return that $(G, k, \eta)$ is a no-instance. 

\item[Step 2b.] Otherwise, for every positive integer $\ell$ such that $\ell \card{H} \leq k$ and $G$ has a component $H' \neq H$ of size exactly $\ell$, we arbitrarily pick such a component $H'$, and recursively call \bralgobcc\ on the instance $(G + F_{HH'}, k - \card{F'}, \eta)$, where $F_{HH'} = \set{uv \in \binom{V(G)}{2} ~|~ u \in V(H), v \in V(H')}$. If at least one recursive call returns yes, then we return that $(G, k, \eta)$ is a yes-instance; otherwise, we return that $(G, k, \eta)$ is a no-instance. 
\end{description}

The correctness of \bralgobcc\ follows from Lemma~\ref{lem:bralgobcc-correctness}. We now show that the algorithm runs in time \algobcdruntime. To see this, notice that the most time-consuming step of the algorithm is Step 2b, where we branch into all possible choices for $\ell$ and make recursive calls. And notice that $\ell$ could range from $1$ to $k$ (when $\card{H} = 1$), and on each branch, the parameter drops by $\ell \card{H}$, which also could range from $1$ to $k$. Thus the number of recursive calls is governed by the recurrence relation $T(k) =  T(k - 1) + T(k -2) + \cdots + T(0)$ with $T(0) = 1$. Using induction on $k$, it is straightforward to verify that $T(k) \leq 2^k$. That is, the algorithm makes at most $2^k$ recursive calls, and notice that each recursive call can be executed in polynomial time. We thus have the following result. 
\begin{theorem}\label{thm:bcc-branching-algorithm}
    \bcc\ admits an algorithm that runs in time \algobcdruntime. 
\end{theorem}

%% file: Full-version-FPT-decision-deletion.tex
\subsection{FPT Algorithm for Balanced Cluster Deletion} 
\label{sec:deletion-fpt}
In this section, we show that \bcd\ admits an algorithm that runs in time \algobcdruntime. And in particular, we show that \bcd\ is polynomial-time solvable when the input graph is a cluster graph. 
To design our  algorithm, we consider an auxiliary problem, where rather than requiring that the size difference between the components of be bounded, we simply demand that the size of each component be within a specified range. Formally, for integers $\gamma_1, \gamma_2 \geq 1$, we say that a graph $G$ is $(\gamma_1, \gamma_2)$-cardinality-constrained if $\gamma_1 \leq \card{H} \leq \gamma_2$ for every connected component $H$ of $G$. We now define the \cccdfull\ problem as follows. 

\defproblem{\cccdfull\ (\cccd)}{A graph $G$ and non-negative integers $\gamma_1, \gamma_2$ and $k$.}{Decide if there exists $F \subseteq E(G)$ such that $\card{F} \leq k$ and $G - F$ is a $(\gamma_1, \gamma_2)$-cardinality-constrained cluster graph.} 

Consider an instance $(G, k, \eta)$ of \bcd. Notice that $(G, k, \eta)$ is a yes-instance of \bcd\ if and only if there exist positive integers $\gamma_1$ and $\gamma_2$ with $\gamma_2 - \gamma_1 \leq \eta$ such that $(G, k, \gamma_1, \gamma_2)$ is a yes-instance of \cccd. Thus, given an instance of \bcd, we can guess all possible choices for $\gamma_1$ and $\gamma_2$ and solve the corresponding instance of \cccd. As $1 \leq \gamma_1, \gamma_2 \leq n$, we only have at most $n^2$ valid guesses. In short, to solve \bcd, it is enough to solve at most $n^2$ instances of \cccd. Hence, assuming that we can solve \cccd\ in time \algobcdruntime, we have the following theorem.  
\begin{theorem}
\label{thm:deletion-fpt}
\bcd\ admits an algorithm that runs in time \algobcdruntime. 
\end{theorem}
As noted above, to prove Theorem~\ref{thm:deletion-fpt}, it is enough to prove the following theorem. 
\begin{theorem}
\label{thm:deletion-cap-fpt}
\cccd\ admits an algorithm that runs in time \algobcdruntime. 
\end{theorem}

The rest of this section is dedicated to proving Theorem~\ref{thm:deletion-cap-fpt}. To that end, we first prove the following lemma, which says that \cccd\ is polynomial-time solvable if the input graph is a clique. 

\begin{lemma}
\label{lem:deletion-fpt-clique}
There is an algorithm that, given a clique $G$ on $n$ vertices and positive integers $\gamma_1, \gamma_2 \leq n$, runs in polynomial time, and either returns $\card{F}$, where $ F \subseteq E(G)$ is a minimum-sized set of edges such that $G - F$ is a $(\gamma_1, \gamma_2)$-cardinality-constrained cluster graph, or correctly reports that no such set $F$ (of any size) exists. 
\end{lemma}

\begin{proof}
Let $G, \gamma_1, \gamma_2$ be given, where $G$ is a clique on $n$ vertices. Throughout this proof, by an optimal solution, we mean a minimum-sized set $F \subseteq E(G)$ such that $G - F$ is a $(\gamma_1, \gamma_2)$-cardinality-constrained cluster graph (if such an $F$ exists). 
 First, if $G$ is $(\gamma_1, \gamma_2)$-cardinality-constrained, then $F = \emptyset$ is the unique optimal solution, and we return $\card{F} = 0$ accordingly. So, assume that this is not the case.  
Then, every optimal solution $F$ must split $G$ into connected components, say $t$ connected components $H_1, H_2,\ldots, H_t$, such that $H_i$ is a clique and $\gamma_1 \leq \card{H_i} \leq \gamma_2$ for every $i \in [t]$. Then $\set{V(H_1), V(H_2), \ldots, V(H_t)}$ is a partition of $V(G)$, and as $G$ is a clique and $F$ is precisely the set of all edges in $G$ with exactly one endpoint in $V(H_i)$ and the other in $V(H_j)$ for some distinct $i, j \in [t]$, we have  $\card{F} = \sum_{1 \leq i < j \leq t}\card{H_i} \card{H_j}$. Also, as $\gamma_1 \leq \card{H_i} \leq \gamma_2$ for every $i \in [t]$, we must have $n \leq t \gamma_2$ and $t \gamma_1 \leq n$, or equivalently, $\ceil{n /\gamma_2} \leq t \leq \floor{n / \gamma_1}$. 
Informally, to find an optimal solution $F$, we guess $t$, i.e., the number of components of $G - F$, and find a partition $\set{X_1, X_2, \ldots, X_t}$ of $V(G)$ into $t$ parts that minimises $\sum_{1 \leq i < j \leq t} \card{X_i} \card{X_j}$.\footnote{We will show later in the proof that we need not guess $t$; the quantity $\sum_{1 \leq i < j \leq t} \card{X_i} \card{X_j}$ is minimised when $t = \ceil{n/\gamma_2}$.} 

In light of the above  observations, our algorithm works as follows. {\bf Step 1:} If there does not exist an integer $t$ such that $\ceil{n/\gamma_2} \leq t \leq \floor{n/\gamma_1}$, then we report that there does not exist $F \subseteq E(G)$ such that $G - F$ is a $(\gamma_1, \gamma_2)$-cardinality-constrained cluster graph. {\bf Step 2:} Otherwise, for each $t$ such that $\ceil{n/\gamma_2} \leq t \leq \floor{n/\gamma_1}$, we construct the following partition $\set{X^t_1, X^t_2,\ldots, X^t_t}$ of $V(G)$ into $t$ parts. {\bf Step 2.1:} We first initialise $X^t_1 = X^t_2 = \cdots = X^t_t = \emptyset$. {\bf Step 2.2:} If  there exists a vertex $v \in V(G) \setminus \bigcup_{i \in [t]} X^t_i$, we assign $v$ to one of the $t$ sets as follows. {\bf Step 2.2.1:} If there exists $i \in [t]$ such that $\card{X^t_i} < \gamma_1$, then we assign $v$ to $X^t_i$ where $i$ is the least index in $[t]$ with $\card{X^t_i} < \gamma_1$; {\bf Step 2.2.2:} Otherwise, we assign $v$ to $X^t_j$ where $j$ is the least index in $[t]$ such that $\card{X^t_j} < \gamma_2$. {\bf Step 3:} Finally, we return $\min_{t} \sum_{1 \leq i < j \leq t} \card{X^t_i} \card{X^t_j}$, where the minimum is over all integers $t$ with $\ceil{n/\gamma_2} \leq t \leq \floor{n/ \gamma_1}$. 

Before proving that this algorithm is indeed correct, notice that for each choice of $t$, we can construct $\set{X^t_1, X^t_2,\ldots, X^t_t}$ and compute $\sum_{1 \leq i < j \leq t} \card{X^t_i} \card{X^t_j}$ in polynomial time. And as $t$ has only at most $n$ choices, our algorithm runs in polynomial time.  

Now, to see that our algorithm is correct, fix $t$ and consider the corresponding partition $\set{X^t_1, X^t_2,\ldots, X^t_t}$ of $V(G)$ that our algorithm constructs. Notice that in Step 2.2,1, we assign exactly $\gamma_1$ vertices to each $X^t_i$; and in Step 2.2,2, once $\card{X^t_i} = \gamma_2$, we no longer assign vertices to $X^t_i$.  Also, in Step 2.2.2, we assign a vertex to $X^t_j$ only if $\card{X^t_i} = \gamma_2$ for every $i < j$. These observations imply that the partition $\set{X^t_1, X^t_2,\ldots, X^t_t}$ has the following properties. 
\begin{enumerate}[(A).]
\item For every $i \in [t]$, we have $\gamma_1 \leq \card{X^t_i} \leq \gamma_2$. 
\item We have $\card{X^t_1} \geq \card{X^t_2} \geq \cdots \geq \card{X^t_t}$. 
\item\label{item:except-possibly} Finally, exactly one of the following statements holds: (i) $n = t \gamma_2$ in which case $\card{X^t_i} = \gamma_2$ for every $i \in [t]$; (ii) $n = t \gamma_1$ in which case $\card{X^t_i} = \gamma_1$ for every $i \in [t]$; there exists a unique index $r \in [t]$ such that $\card{X^t_i} = \gamma_2$ for every $i < r$ and $\card{X^t_i} = \gamma_1$ for every $i > r$.  
\end{enumerate}
We now prove that the partition $\set{X^t_1, X^t_2,\ldots, X^t_t}$ is indeed optimal. 
\begin{claim}
For each fixed integer $t$ with $\ceil{n /\gamma_2} \leq t \leq \floor{n/ \gamma_1}$, we have $\sum_{1 \leq i < j \leq t} \card{X^t_i} \card{X^t_j} = \min \sum_{1 \leq i < j \leq t} \card{Z_i} \card{Z_j}$, where the minimum is over all partitions $\set{Z_1, Z_2,\ldots, Z_t}$ of $V(G)$ into $t$ parts with $\gamma_1 \leq \card{Z_i} \leq \gamma_2$ for every $i \in [t]$. 
\end{claim}
\begin{claimproof}
    Fix an integer $t$ with $\ceil{n /\gamma_2} \leq t \leq \floor{n/ \gamma_1}$ and a partition $\set{Y_1, Y_2,\ldots, Y_t}$ of $V(G)$ such that $\gamma_1 \leq \card{Y_i} \leq \gamma_2$ for every $i \in [t]$ and $\sum_{1 \leq i < j \leq t} \card{Y_i} \card{Y_j} = \min \sum_{1 \leq i < j \leq t} \card{Z_i} \card{Z_j}$, where the minimum is over all partitions $\set{Z_1, Z_2,\ldots, Z_t}$ of $V(G)$ into $t$ parts with $\gamma_1 \leq \card{Z_i} \leq \gamma_2$ for every $i \in [t]$. Assume without loss of generality that $\card{Y_1} \geq \card{Y_2} \geq \cdots \geq \card{Y_t}$. Notice that to prove the claim, it is enough to prove that $\card{X^t_i} \leq \card{Y_i}$ for every $i \in [t]$, which will imply that $\sum_{1 \leq i < j \leq t} \card{X^t_i} \card{X^t_j} \leq \sum_{1 \leq i < j \leq t} \card{Y_i} \card{Y_j}$; in fact, $\card{X^t_i} \leq \card{Y_i}$ for every $i \in [t]$ will imply that $\card{X^t_i} = \card{Y_i}$ for every $i \in [t]$, as $\sum_{i \in [t]} \card{X^t_i} = \sum_{i \in [t]} \card{Y_i}$, and this will in turn imply that $\sum_{1 \leq i < j \leq t} \card{X^t_i} \card{X^t_j} = \sum_{1 \leq i < j \leq t} \card{Y_i} \card{Y_j}$

    We now prove that $\card{X^t_i} \leq \card{Y_i}$ for every $i \in [t]$. Assume for a contradiction that there exists an index $r \in [t]$ such that $\card{X^t_r} > \card{Y_r}$, and assume without loss of generality that $r$ is the least such index. Then, as $\sum_{i \in [t]} \card{X^t_i} = n = \sum_{i \in [t]} \card{Y_i}$,    there exists another index $s \in [t]$ such that $\card{Y_s} > \card{X^t_s}$. We claim that $s > r$. To see this, recall that we have $ \gamma_ 1 \leq \card{X^t_j}, \card{Y_j} \leq \gamma_2$ for every $j \in [t]$; then, as $\card{X^t_r} > \card{Y_r}$, we can conclude that $\card{X^t_r} > \gamma_1$. But then  Property~(\ref{item:except-possibly}) above implies that $\card{X^t_j} = \gamma_2$ for every $j < r$, and therefore there cannot exist $s < r$ such that $\card{Y_s} > \card{X^t_s}$. We can thus conclude that there exists $s > r$ such that $\card{Y_s} > \card{X^t_s}$. 
    Notice that we have $\gamma_2 \geq \card{X^t_r} > \card{Y_r}$ and $\gamma_1 \leq \card{X^t_s} < \card{Y_s}$.  
    Let $\set{Y'_1, Y'_2,\ldots, Y'_{t}}$ be the partition of $V(G)$ obtained from $\set{Y_1, Y_2,\ldots, Y_t}$ by moving exactly one vertex from $Y_{s}$ to $Y_{r}$. Notice then that $\card{Y'_s} = \card{Y_s} - 1 \geq \gamma_1$, $\card{Y'_{r}} = \card{Y_r} + 1 \leq \gamma_2$, and $\card{Y'_i} = \card{Y_i}$ for every $i \in [t] \setminus \set{r, s}$. In particular, $\gamma_1 \leq \card{Y'_i} \leq \gamma_2$ for every $i \in [t]$.  
    
    We will show that $\sum_{1 \leq i < j \leq t} \card{Y'_i} \card{Y'_j} < \sum_{1 \leq i < j \leq t} \card{Y_i} \card{Y_j}$, which will contradict the definition of  $\set{Y_1, Y_2,\ldots, Y_t}$. 
To prove this, observe the following facts. 
\begin{enumerate}
    \item For  $i, j \in [t] \setminus \set{r, s}$, we have $\card{Y'_i} = \card{Y_i}$ and $\card{Y'_{j}} = \card{Y_{j}}$, and hence $\card{Y'_i} \card{Y'_{j}} = \card{Y_i} \card{Y_j}$. 
    \item \begin{sloppypar}We have $\sum_{j \in [t] \setminus \set{r, s}} (\card{Y'_j} \card{Y'_r} + \card{Y'_j} \card{Y'_{s}}) = \sum_{j \in [t] \setminus \set{r, s}} \card{Y_j} (\card{Y_r} + 1 + \card{Y_{s}} - 1) = \sum_{j \in [t] \setminus \set{r, s}} (\card{Y_j} \card{Y_r} + \card{Y_j} \card{Y_{s}})$.\end{sloppypar}
    \item  We have $\card{Y'_r} \card{Y'_{s}} < \card{Y_r} \card{Y_{s}}$. To see this, notice that $\card{Y'_r} \card{Y'_{s}} = (\card{Y_r} + 1)(\card{Y_{s}} - 1) = \card{Y_r} \card{Y_{s}} - \card{Y_r} + \card{Y_{s}} - 1 < \card{Y_r} \card{Y_s}$, where the last inequality holds because $\card{Y_r} \geq \card{Y_{s}}$ as $r < s$ and $\card{Y_1} \geq \card{Y_2} \geq \cdots \geq \card{Y_t}$. 
\end{enumerate}  
We thus have  
\begin{align*}
\sum_{1 \leq i < j \leq t} \card{Y'_i} \card{Y'_j} &= \card{Y'_r} \card{Y'_{s}} +  \sum_{j \in [t] \setminus \set{r, s}} (\card{Y'_j} \card{Y'_r} + \card{Y'_j} \card{Y'_{s}}) + \sum_{i, j \in [t] \setminus \set{r , s}} \card{Y'_i} \card{Y'_{j}}  \\
&< \card{Y_r} \card{Y_{s}}  + \sum_{j \in [t] \setminus \set{r, s}} (\card{Y_j} \card{Y_r} + \card{Y_j} \card{Y_{s}}) + \sum_{i, j \in [t] \setminus \set{r , s}} \card{Y_i} \card{Y_{j}} \\
&=\sum_{1 \leq i < j \leq t} \card{Y_i} \card{Y_{j}}, 
\end{align*}
which is a contradiction. 
\end{claimproof}
We have thus shown that for each $t$, the partition $\set{X^t_1, X^t_2,\ldots, X^t_t}$ is optimal. \emph{While this shows that our algorithm is correct and thus completes the proof of the lemma, we can improve our algorithm as follows.} In Step 2, we need not go over all possible choices of $t$ between $\ceil{n/\gamma_2}$ and $\floor{n/\gamma_1}$; we will argue that an optimal solution is obtained at $t = \ceil{n/\gamma_2}$. That is, we will prove that 
\begin{equation*}
\sum_{1 \leq i < j \leq \ceil{n/\gamma_2}} \card{X^{\ceil{n/\gamma_2}}_i} \card{X^{\ceil{n/\gamma_2}}_j} = \min_{\substack{t \\ \ceil{n/\gamma_2} \leq t \leq \floor{n/\gamma_1}}} ~ \sum_{1 \leq i < j \leq t} \card{X^t_i} \card{X^t_j}. \tag{$\blacklozenge \blacklozenge \blacklozenge$}\label{eq:min}
\end{equation*}
Observe that the following claim implies Equation~(\ref{eq:min}). 
\begin{claim}
\label{claim:ceiling-minimum}
For each integer $t$ with $\ceil{n /\gamma_2} \leq t \leq \floor{n/ \gamma_1} - 1$, we have $\sum_{1 \leq i < j \leq t} \card{X^t_i} \card{X^t_j} \leq \sum_{1 \leq i < j \leq t + 1} \card{X^{t + 1}_i} \card{X^{t + 1}_j}$. 
\end{claim}
\begin{claimproof}
For each $t$ with $\ceil{n /\gamma_2} \leq t \leq \floor{n/ \gamma_1}$, let $h(t) = \sum_{1 \leq i < j \leq t} \card{X^t_i} \card{X^t_j}$. Now 
fix $t$ such that $\ceil{n /\gamma_2} \leq t \leq \floor{n/ \gamma_1} - 1$. We will show that $h(t) \leq h(t + 1)$. 
Observe first that $n \leq t \gamma_2$. Consider the partitions ${X^t_1, X^t_2,\ldots, X^t_t}$ and ${X^{t + 1}_1, X^{t + 1}_2,\ldots, X^{t + 1}_{t + 1}}$ of $V(G)$ that our algorithm constructs. Notice that we must have $\card{X^{t + 1}_{t + 1}} = \gamma_1$ for our algorithm would not assign a vertex to $X^{t + 1}_{t + 1}$ in Step 2.2.2 as $n \leq t \gamma_2$. For every $j \in  [t]$, as $\card{X^{t + 1}_j} \leq \card{X^t_j}$, we also assume without loss of generality that $X^{t + 1}_j = X^{t}_j \setminus X^{t + 1}_{t + 1}$.  

Now, to prove that $h(t) \leq h(t + 1)$, we calculate the difference $h(t + 1) - h(t)$ and show that this quantity is at least $0$. And to calculate $h(t + 1) - h(t)$, notice that we only need to account for the contribution of edges incident with $X^{t + 1}_{t + 1}$ to both $h(t)$ and $h(t + 1)$; as for the other edges, notice that as $X^{t + 1}_j = X^t_j \setminus X^{t + 1}_{t + 1}$ for every $j \in [t]$, each edge not incident with $X^{t + 1}_{t + 1}$ contributes equally to both $h(t)$ and $h(t + 1)$.  In what follows, we will use the fact that $X^{t + 1}_{t + 1}$ is the disjoint union of the $t $ sets $X^{t + 1}_{t + 1} \cap X^t_{j}$ for all $j \in [t]$, and therefore, $\gamma_1 = \card{X^{t + 1}_{t + 1}} = \sum_{j \in [t]} \card{X^{t + 1}_{t + 1} \cap X^t_{j}}$. 

Let us first compute the number of edges that contribute $1$ to $h(t + 1)$ and $0$ to $h(t)$. Notice that each such edge has exactly one endpoint in $X^{t + 1}_{t + 1}$. In particular, the edge is between $X^{t + 1}_{t + 1}$ and $X^{t + 1}_j$ for some $j \in [t]$. As noted above, $X^{t + 1}_{t + 1}$ is the disjoint union of $(X^{t + 1}_{t + 1} \cap X^t_j)$ and $\bigcup_{i \in [t] \setminus \set{j}} X^{t + 1}_{t + 1} \cap X^t_i$. Notice also that each edge between $X^{t + 1}_{t + 1} \cap X^t_j$ and $X^{t + 1}_j$ contributes $1$ to $h(t + 1)$ and $0$ to $h(t)$. And the number of such edges is $\card{X^{t + 1}_{t + 1} \cap X^t_j} \cdot \card{X^{t + 1}_j}$. %
As for the remaining edges between $X^{t + 1}_{t + 1}$ and $X^{t + 1}_j$, notice that each such edge is between $X^{t + 1}_{t + 1} \cap X^{t + 1}_i \subseteq X^{t}_i$ and $X^{t + 1}_j \subseteq X^t_j$ for some $i \in [t] \setminus \set{j}$, and each such edge contributes $1$ to both $h(t)$ and $h(t + 1)$. Thus the number of edges  that contribute $1$ to $h(t + 1)$ and $0$ to $h(t)$ is precisely 
\begin{align*}\sum_{j \in [t]} \card{X^{t + 1}_{t + 1} \cap X^t_j} \cdot \card{X^{t + 1}_j} &\geq \gamma_1 \cdot \sum_{j \in [t]} \card{X^{t + 1}_{t + 1} \cap X^t_j} ~~~~~~~~~~~~~~~~~~~~~~~~~ (\text{because }\card{X^{t + 1}_j} \geq \gamma_1)\\
&= \gamma_1 \cdot \card{X^{t + 1}_{t + 1}} \\
&= \gamma_1^2. 
\end{align*}  
Let us now compute the number of edges that contribute $1$ to $h(t)$ and $0$ to $h(t + 1)$. Notice that each such edge has both its endpoints in $X^{t + 1}_{t + 1}$. To compute the  contribution of such edges to $h(t)$, notice that for every distinct pair $i, j \in [t]$, each edge between $X^{t + 1}_{t + 1} \cap X^{t}_i$ and $X^{t + 1}_{t + 1} \cap X^{t}_j$ contributes exactly $1$ to $h(t)$, and the number of such edges is exactly $\card{X^{t + 1}_{t + 1} \cap X^{t}_i} \cdot \card{X^{t + 1}_{t + 1} \cap X^{t}_j}$. The only remaining edges with both endpoints in $X^{t + 1}_{t + 1}$ are those with both their endpoints in $X^{t + 1}_{t + 1} \cap X^{t}_i$ for some $i \in [t]$, and each such edge contributes exactly $0$ to $h(t)$. Thus the number of edges  that contribute $0$ to $h(t + 1)$ and $1$ to $h(t)$ is precisely 
\begin{align*}
    \sum_{1 \leq i < j \leq t} \card{X^{t + 1}_{t + 1} \cap X^{t}_i} \cdot \card{X^{t + 1}_{t + 1} \cap X^{t}_j}
    &\leq \frac{1}{2}\Lb{\sum_{j \in [t]} \card{X^{t + 1}_{t + 1} \cap X^{t}_j}}^2 \\
    &=\frac{1}{2} \Lb{\card{X^{t + 1}_{t + 1}}}^2 \\
    &=(1/2)\gamma_1^2,
\end{align*}
where the first inequality holds because $\sum_{1 \leq i < j \leq t} 2x_i x_j \leq \Lb{\sum_{j \in [t]} x_j}^2$ for any non-negative numbers $x_1, x_2,\ldots, x_t$. We have thus shown that $h(t + 1) - h(t) \geq \gamma_1^2 - (1/2) \gamma_1^2 \geq 0$, and therefore $h(t + 1) \geq h(t)$. 
\end{claimproof}
Claim~\ref{claim:ceiling-minimum} implies Equation~(\ref{eq:min}), which in turn implies that in Step 2 of our algorithm, we need to compute $\sum_{1 \leq i < j \leq t} \card{X^{t}_i} \card{X^{t}_j}$ for only $t = \ceil{n/\gamma_2}$. 
\end{proof}

Observe that we can use the algorithm of Lemma~\ref{lem:deletion-fpt-clique} to solve \cccd\ when the input graph is a cluster graph; we simply invoke the algorithm of Lemma~\ref{lem:deletion-fpt-clique} on each connected component of $G$. We thus have the following corollary. 
\begin{corollary}
\label{cor:deletion-fpt-cluster}
\cccd\ is polynomial-time solvable on cluster graphs. 
\end{corollary}

 In light of Corollary~\ref{cor:deletion-fpt-cluster}, it is straightforward to design an \FPT\ algorithm for \cccd\ on general graphs. Recall that a graph $G$ is a cluster graph if and only if $G$ does not contain $P_3$ as an induced subgraph. So, given an instance $(G, k, \gamma_1, \gamma_2)$ of \cccd, we branch on induced $P_3$s until the graph becomes $P_3$-free and then invoke Corollary~\ref{cor:deletion-fpt-cluster}. We now formally describe our algorithm, which we call \algocccd. 

\subparagraph*{\algocccd.} Given an instance $(G, k,\gamma_1, \gamma_2)$ of \cccd, we do as follows.
\begin{description}
    \item[Step 1:] If $G$ is a cluster graph, then we use the algorithm of Corollary~\ref{cor:deletion-fpt-cluster} to solve the problem. 
    \item[Step 2:] If Step 1 is not applicable and $k > 0$, then we greedily find an induced $P_3$, say $uvw$ (i.e., $uv, vw \in E(G)$ and $uw \notin E(G)$) and recursively call \algocccd\ on the two instances $(G-uv, k - 1, \gamma_1, \gamma_2)$ and $(G-vw, k - 1, \gamma_1, \gamma_2)$. 
    \item[Step 3:] If the previous two steps are not applicable, then $G$ is not a cluster graph and $k = 0$, and we return that $(G, k, \gamma_1, \gamma_2)$ is a no-instance. 
\end{description}

Observe that \algocccd\ runs in time \algobcdruntime. In each execution of Step 2, we make $2$ recursive calls, and we recurse only until $k$ becomes $0$, and hence we make at most $2^k$ recursive calls. All the other steps of the algorithm take only polynomial time. This completes the proof of Theorem~\ref{thm:deletion-cap-fpt}.

%% file: Full-version-FPT-decision-completion-and-editing.tex
\subsection{FPT Algorithms for BCC and BCE Using Correspondence Between Solutions and Integer Partitions}
\label{sec:fpt-completion-editing}
In this subsection, we design algorithms for both \bcc\ and \bce\ that run in time \algobccruntime. These algorithms rely on a correspondence between a solution (for either problem) and a family of integer partitions. 
First, consider \bcc. Even though Theorem~\ref{thm:bcc-branching-algorithm} guarantees a faster  algorithm for \bcc\ with a running time of \algobcdruntime, we present another algorithm here, as the arguments are instructive. 
Notice that for a yes-instance $(G, k, \eta)$ of \bcc, every solution $F \subseteq \binom{V(G)}{2} \setminus E(G)$ for $(G, k, \eta)$ ``merges'' certain  components of $G$ together to form components of $G + F$. 
So, to solve \bcc, we only need to decide which components of $G$ are merged  with which other components. %
Observe that an algorithm that runs in time $2^{\cO(k \log k)} n^{\cO(1)}$ is then quite straightforward: We simply need to go over all possible partitions of the connected components of $G$,  and check if we can merge each part into a clique by adding at most $k$ edges in total, and if the resulting cluster graph is $\eta$-balanced. 
To design a \algobccruntime\ time  algorithm, we argue that we can naturally associate each solution $F$ with a partition of $\card{V(F)}$, i.e., the number of vertices that $F$ modifies, and then we leverage the fact that the number of partitions of a positive integer $\ell$ is $2^{\cO(\sqrt{\ell})}$.  

\subparagraph*{Partitions of an integer.} For a positive integer $\ell$, by a partition of $\ell$, we mean writing $\ell$ as a sum of positive integers where we ignore the order of the summands. 
For example, $1 + 4 + 6 + 6$ is a partition of $17$. 
As mentioned, the order of the summands is immaterial; so $4 + 6 + 1 + 6$, $1 + 6 + 4 + 6$, and $6 + 6 + 1 + 4$ are all the same partition of $17$. 
For convenience, we represent a  partition $x_1 + x_2 + \cdots + x_t$ of $\ell$ by the multiset $\set{x_1, x_2,\ldots, x_t}$; for example, we say that $\set{1, 4, 6, 6}$ is a partition of $17$. Consider $\ell \in \mathbb{N}$. Following standard convention, we use $p(\ell)$ to denote the number of partitions of $\ell$. 
We will use a result due to Hardy and Ramanujan~\cite{hardy1918asymptotic} that bounds $p(\ell)$, and the fact that we can enumerate all partitions of a given integer $\ell$ in time proportional to $p(\ell)$ (see, for example, \cite{knuth2014art}).  

\begin{proposition}[follows from \cite{hardy1918asymptotic}]
\label{prop:partition-number}
There exists a constant $C$ such that $p(\ell) \leq 2^{C \sqrt{\ell}}$ for every positive integer $\ell$. 
Moreover, there is an algorithm that, given a positive integer $\ell$ as input, runs in time $2^{\cO(\sqrt{\ell})}$, and enumerates all partitions of $\ell$. 
\end{proposition}

\subparagraph*{Idea behind our algorithm for \bcc.} Consider a yes-instance $(G, k, \eta)$ of \bcc, and let $F \subseteq \binom{V(G)}{2} \setminus E(G)$ be a solution for $(G, k, \eta)$. 
Let $G_1, G_2,\ldots, G_{t}$ be the components of $G + F$ that were modified by $F$; and for each $i \in [t]$, let $G_{i1}, G_{i2},\ldots, G_{i r_i}$ be the components of $G$ that were merged together to form the component $G_i$ of $G + F$. 
Now, since $\card{F} \leq k$, we have $\card{V(F)} \leq 2k$. 
That is, $F$ modifies at most $2k$ vertices of $G$. 
And observe that $F$ modifies a vertex $v$ if and only if $F$ modifies every vertex in the connected component of $G +F$ that contains $v$. 
We thus have  $\card{G_1} + \card{G_2} + \cdots + \card{G_t} = \card{V(F)}$. In other words, $\set{\card{G_1}, \card{G_2},\ldots, \card{G_t}}$ is a partition of $\card{V(F)}$. 
Also, for each $i \in [t]$, we have $\card{G_{i1}} + \card{G_{i2}} + \cdots + \card{G_{i r_i}} = \card{G_i}$; that is, $\set{\card{G_{i1}}, \card{G_{i2}}, \ldots, \card{G_{i r_i}}}$ is a partition of $\card{G_i}$. 
Based on these observations, we design our algorithm, which works as follows. 
We guess $\card{V(F)}$ (at most $2k$ guesses); for each such guess, we guess the partition $\set{\card{G_1}, \card{G_2},\ldots, \card{G_t}}$; and for each such combination of guesses, we guess the partition $\set{\card{G_{i1}}, \card{G_{i2}}, \ldots, \card{G_{i r_i}}}$ for each $i \in [t]$; and we check if there is indeed a solution that is consistent with our guesses. 
Bounding the running time requires a careful analysis, for which we use Proposition~\ref{prop:partition-number}. 

\subsubsection*{Notation and Preliminary Observations}

To formally describe our algorithms, we first introduce the following terminology and state a few preparatory results.

\subparagraph*{The sum of pairwise products function $\spp$. } 
For each partition $X = \set{x_1, x_2,\ldots, x_t}$ of a positive integer, we define $\spp(X)$ as follows: if $X$ is a singleton set, i.e., $t = 1$, then $\spp(X) = 0$; otherwise, $\spp(X) = \sum_{\substack{i, j \in [t] \\ i \neq j}} x_i x_j$. 
That is, $\spp(X)$ is the sum of  pairwise products of the elements of $X$. 
For example, for the partition $X = \set{1, 4, 6, 6}$ of 17, we have $\spp(X) = (1 \times 4) + (1 \times 6) + (1 \times 6) + (4 \times 6) + (4 \times 6) + (6 \times 6) = 4 + 6 + 6 + 24 + 24 + 36 = 100$. 

\subparagraph*{The completion of a cluster graph w.r.t. a multiset of partitions.} Consider a cluster graph $G$ and a positive integer $\ell$. 
Let $X = \set{x_1, x_2,\ldots, x_t}$ be a partition of $\ell$. 
We say that $X$ is $G$-valid if $G$ contains $t$ distinct connected components $G_1, G_2,\ldots, G_t$ such that $\card{G_i} = x_i$ for every $i \in [t]$. 
Now, consider $G, \ell$ and $X = \set{x_1, x_2,\ldots, x_t}$ as before. 
Let $X_1, X_2,\ldots,X_t$ be such that for each $i \in [t]$,  $X_i$ is a partition of $x_i$. 
Notice then that the multiset\footnote{When dealing with the union of multisets $A$ and $B$, by $A \cup B$, we mean the ``additive union'' of $A$ and $B$---that is, for each element $x \in A \cup B$, the multiplicity of $x$ in $A \cup B$ is the sum of the multiplicity of $x$ in $A$ and the multiplicity of $x$ in $B$. For example, if $A = \set{x, y, y, z, z, z}$ and $B = \set{x, x, x, y, w}$, then $A \cup B = \set{x, x, x, x, y, y, y, z, z, z, w}$.} $X' = \bigcup_{i \in [t]}X_i$ is also a partition of $\ell$. 
For $G, \ell, X = \set{x_1, x_2,\ldots, x_t}, X_1, X_2,\ldots, X_t$ and $X'$ as before such that $X'$ is $G$-valid, 
 we define a \emph{completion of $G$ with respect to $(X, \set{X_1, X_2,\ldots, X_t})$} as follows. 
 For each $i \in [t]$, let $X_i = \set{x_{i1}, x_{i2},\ldots, x_{i r_i}}$. As $X' = \bigcup_{i \in [t]} X_i$ is $G$-valid, $G$ contains $\sum_{i \in [t]} r_i$ distinct connected components $G_{11}, G_{12},\ldots, G_{1 r_1}, G_{21}, G_{22},\ldots, G_{2 r_2},\ldots,G_{t1},$ $ G_{t2},\ldots, G_{t r_t}$ such that for every $i \in [t]$, we have $\card{G_{ij}} = x_{ij}$ for every $j \in [r_i]$. 
 Fix such a choice of $\sum_{i \in [t]} r_i$ components $G_{11}, G_{12},\ldots, G_{1 r_1}, G_{21}, G_{22},\ldots, G_{2 r_2},\ldots,G_{t1}, G_{t2},\ldots, G_{t r_t}$. 
 By a completion of $G$ with respect to $(X, \set{X_1, X_2,\ldots, X_t})$, we mean the graph obtained from $G$ by turning the subgraph induced by $\bigcup_{j \in [r_i]} V(G_{ij})$ into a clique for every $i \in [t]$. 
 Notice that as $G$ is a cluster graph, each $G_{ij}$ is a clique, and hence to turn $\bigcup_{j \in [r_i]} V(G_{ij})$ into a clique, we need to add exactly zero edges if $r_i = 1$, and otherwise, we need to add exactly  $\sum_{\substack{j, j' \in [r_i] \\ j \neq j'}} \card{G_{ij}} \card{G_{ij'}} = \sum_{\substack{j, j' \in [r] \\ j \neq j'}} x_{ij} x_{ij'}$ edges. 
 In either case, we need to add exactly $\spp(X_i)$ edges; recall that $\spp(X_i) = 0$ if $X_i = \set{x_{i1}}$ and $\spp(X_i) = \sum_{\substack{j, j' \in [r_i] \\ j \neq j'}} x_{ij} x_{ij'}$ otherwise. 
Thus, a completion of $G$ with respect to $(X, \set{X_1, X_2,\ldots, X_t})$ requires the addition of exactly $\sum_{i \in [t]} \spp(X_i)$ edges to $G$. 
Notice that a completion of $G$ with respect to $(X, \set{X_1, X_2,\ldots,X_t})$ is also a cluster graph. 
Notice also that such a completion of $G$  need not be unique as there may be multiple choices for the components $G_{11}, G_{12},\ldots, G_{1 r_1}, G_{21}, G_{22},\ldots, G_{2 r_2},\ldots,G_{t1}, G_{t2},\ldots, G_{t r_t}$. But notice that if $G'$ and $G''$ are two distinct completions of $G$ with respect to $(X, \set{X_1, X_2,\ldots,X_t})$, then $G'$ and $G''$ both have the same number of components. Moreover, for any $j \in \mathbb{N}$, $G'$ and $G''$ have exactly the same number of components of size exactly $j$; that is, if $G'_1, G'_2,\ldots, G'_s$ are the components of $G'$ and $G''_1, G''_2,\ldots, G''_s$ the components of $G''$, then $\{\card{G'_1}, \card{G'_2},\ldots, \card{G'_s}\} = \{\card{G''_1}, \card{G''_2},\ldots, \card{G''_s}\}$. In particular, $G'$ is $\eta$-balanced if and only if $G''$ is $\eta$-balanced. 
We summarise this discussion in the following observations. 

\begin{observation}
\label{obs:completion-fpt-correctness}
Consider a cluster graph $G$, $\ell \in \mathbb{N}$, a partition $X = \set{x_1, x_2,\ldots, x_t}$ of $\ell$, and $X_1, X_2,\ldots, X_t$, where $X_i$ is a partition of $x_i$ for every $i \in [t]$. Let $X'$ be $G$-valid, where $X' = \bigcup_{i \in [t]} X_i$.  
Then, \emph{every} completion of $G$ with respect to $(X, \set{X_1, X_2,\ldots, X_t})$ requires the addition of exactly $\sum_{i \in [t]} \spp(X_i)$ edges to $G$. 
In other words, each completion of $G$ with respect to $(X, \set{X_1, X_2,\ldots, X_t})$ is precisely a supergraph $G + F$ of $G$, where $F \subseteq \binom{V(G)}{2} \setminus E(G)$, with the following properties: (i) $G + F$ is a cluster graph, (ii) $\card{F} = \sum_{i \in [t]} \spp(X_i)$ and (iii) $\card{V(F)} = \ell$. 
\end{observation}

\begin{observation}
\label{obs:completion-fpt-runtime}
Given $G, \ell, X = \set{x_1, x_2,\ldots, x_t}$ and $X'$, where $G$ is a cluster graph, $\ell \in \mathbb{N}$, $X = \set{x_1, x_2,\ldots, x_t}$ is a partition of $\ell$, $X_i$ is a partition of $x_i$ for every $i \in [t]$, and $X' = \bigcup_{i \in [t]} X_i$, we can perform the following operations in polynomial time. 
\begin{enumerate}
\item Check if $X'$ is $G$-valid; this only requires checking for every $x \in X'$ if $G$ contains at least $\mul(x, X')$ distinct components, each of size exactly $x$, where $\mul(x, X')$ is the multiplicity of $x$ in $X'$. 
\item Compute $\sum_{i \in [t]} \spp(X_i)$. 
\item Construct a completion of $G$ with respect to $(X, \set{X_1, X_2,\ldots, X_t})$ (provided $X'$ is $G$-valid), and check whether it is $\eta$-balanced. 
\end{enumerate}
\end{observation}

We will also need the following lemma to bound the running time of our algorithms.   %
\begin{lemma}
\label{lem:completion-bound}
For $\ell \in \mathbb{N}$, the number of choices for the pair $(X, \set{X_1, X_2,\ldots, X_t})$, where $X = \set{x_1, x_2,\ldots, x_t}$ is a partition of $\ell$ and $X_i$ is a partition of $x_i$ for every $i \in [t]$, is $2^{\cO(\ell)}$. 
\end{lemma}
\begin{proof}
 Recall that for every $\ell \in \mathbb{N}$, $p(\ell)$ denotes the number of partitions of $\ell$. 
 Recall also  that by Proposition~\ref{prop:partition-number}, there exists a constant $C$ such that $p(\ell) \leq 2^{C \sqrt{\ell}}$ for every $\ell \in \mathbb{N}$. 
 
 Now, fix $\ell \in \mathbb{N}$. 
 For  convenience, for each partition $X$ of $\ell$, where $X = \set{x_1, x_2,\ldots, x_t}$, let us denote the product $p(x_1) p(x_2) \cdots p(x_t)$ by $\alpha_{\ell}(X)$. 
 Now, among all the partitions of $\ell$, let $X^* = \set{x^*_1, x^*_2,\ldots,x^*_{t^*}}$ be a partition for which the function $\alpha_{\ell}$ attains the maximum value; that is,  $\alpha_{\ell}(X^*) = \max_{X}  \alpha_{\ell}(X)$, where the maximum is over all partitions $X$ of $\ell$.  

 Now, for each partition $X$ of $\ell$, where $X = \set{x_1, x_2,\ldots, x_t}$, notice that there are at most $p(x_1) p(x_2) \cdots p(x_t) = \alpha_{\ell}(X)$ choices for the multiset $\set{X_1, X_2,\ldots, X_t}$, where $X_i$ is a partition of $x_i$ for every $i \in [t]$. 
 Hence the number of choices for the pair $(X, \set{X_1, X_2,\ldots, X_t})$ is at most 
\begin{align*}
  \sum_{\substack{X \\ X \text{ is a partition of } \ell}} \alpha_{\ell}(X) &\leq \sum_{\substack{X \\ X \text{ is a partition of } \ell}} \alpha_{\ell}(X^*) ~~~~~~~ ~~~~~
    ~~~  (\text{By the definition of } X^*)\\
    &=\sum_{\substack{X \\ X \text{ is a partition of } \ell}} p(x^*_1) p(x^*_2) \cdots p(x^*_{t^*}) \\
    &=p(\ell)p(x^*_1) p(x^*_2) \cdots p(x^*_{t^*}) \\
    &\leq 2^{C \sqrt{\ell}} \cdot 2^{C \sqrt{x^*_1}} \cdot 2^{C \sqrt{x^*_2}} \cdots 2^{C \sqrt{x^*_{t^*}}} ~~~~~~~ ~~~
     \text{(By Proposition~\ref{prop:partition-number})}\\
    &=2^{C(\sqrt{\ell} + \sqrt{x^*_1} + \sqrt{x^*_2} + \cdots + \sqrt{x^*_{t^*}})} \\
    &\leq 2^{C(\ell + x^*_1 + x^*_2 + \cdots + x^*_{t^*})} \\&=2^{\cO(\ell)}, 
\end{align*}
where the last equality holds because $x^*_1 + x^*_2 + \cdots + x^*_{t^*}= \ell$. 
\end{proof}
\subsubsection*{The Main Technical Lemma for BCC}
We now prove the following lemma, which will establish the correctness of our algorithm. 
Recall that in an instance $(G, k, \eta)$ of \bcc, the input graph $G$ is a cluster graph. 
\begin{lemma}
\label{lem:completion-fpt-correctness}
Consider an instance $(G, k, \eta)$ of \bcc. 
If $(G, k, \eta)$ is a yes-instance, then either $G$ is $\eta$-balanced or there exist an integer $\ell \in [2k] \setminus \set{1}$ and a partition $X = \set{x_1, x_2,\ldots, x_t}$ of $\ell$ such that there exists a multiset $\set{X_1, X_2,\ldots, X_t}$ of partitions, where $X_i$ is a partition of $x_i$ for each $i \in [t]$, with the following properties: 
\begin{enumerate}
\item $X'$ is $G$-valid, where $X' = \bigcup_{i \in [t]} X_i$; 

\item $\sum_{i \in [t]}\spp(X_i) \leq k$; and

\item a completion of $G$ with respect to $(X, \set{X_1, X_2,\ldots, X_t})$ is $\eta$-balanced. 
\end{enumerate}
\end{lemma}

\begin{proof}
Consider an instance $(G, k, \eta)$ of \bcc. 
Assume that $(G, k, \eta)$ is a yes-instance. 
If $G$ is $\eta$-balanced, then the lemma trivially holds. 
So, assume that $G$ is not $\eta$-balanced. 
Let $F \subseteq \binom{V(G)}{2} \setminus E(G)$ be a solution for $(G, k, \eta)$. 
Let $\ell = \card{V(F)}$; that is, $\ell$ is the number of vertices that $F$ modifies. 
As $G$ is not $\eta$-balanced, $F \neq \emptyset$, and hence $\ell = \card{V(F)} \geq 2$. 
Also, as $\card{F} \leq k$, we have $\ell  = \card{V(F)} \leq 2k$. Thus $\ell \in [2k] \setminus \set{1}$. 

We now show that $F$ corresponds to a partition $X = \set{x_1, x_2,\ldots, x_{t}}$ of $\ell$ and a multiset $\set{X_1, X_2,\ldots, X_t}$ of partitions that satisfy the properties required by the lemma. 
Let $G_1, G_2,\ldots,G_t$ be the connected components of $G + F$ that were modified by $F$; and for each $i \in [t]$, let $G_{i1}, G_{i2},\ldots, G_{ir_i}$ be the connected components of $G$ such that $V(G_i) = \bigcup_{j \in [r_i]} V(G_{ij})$. 
 That is, for each $i \in [t]$, $F$ ``merges'' the components $G_{i1}, G_{i2},\ldots, G_{ir_i}$ together to form the component $G_{i}$ of $G + F$. 
 Thus, $F = \bigcup_{i \in [t]} F_i$, where for each $i \in [t]$, $F_i$ is the set containing all possible edges between $V(G_{ij})$ and $V(G_{i'j'})$ for every pair of distinct $j, j' \in [r_i]$, i.e., $F_i = \{uv ~| \exists j, j' \in [r_i] \text{ with } j \neq j' \text{ such that } u \in V(G_{ij}), v \in V(G_{ij'})\}$. 
 Notice that $\card{F_i} = \sum_{\substack{j, j' \in [r_i] \\ j \neq j'}} \card{G_{ij}} \card{G_{ij'}}$, and as the sets $F_1, F_2,\ldots, F_t$ are pairwise disjoint, we have $\card{F} = \sum_{i \in [t]}\card{F_i}.$  
 Also, as the set of vertices that are modified by $F$ (i.e., $V(F)$) is precisely $\bigcup_{i \in [t]} V(G_i)$, we have $\card{\bigcup_{i \in [t]} V(G_i)} = \sum_{i \in [t]} \card{G_i} = \card{V(F)} = \ell$. 
 
 It is now straightforward to define the partitions $X$ and $X_1, X_2,\ldots, X_t$, as required by the statement of the lemma. 
 For each $i \in [t]$, let $x_i = \card{G_i}$ and $x_{ij} = \card{G_{ij}}$ for every $j \in [r_i]$. Thus $X = \set{\card{G_1}, \card{G_2},\ldots, \card{G_t}} = \set{x_1, x_2,\ldots, x_t}$ is a partition of $\sum_{i \in [t]} \card{G_i} = \card{V(F)} = \ell$;  and for each $i \in [t]$, $X_i = \set{\card{G_{i1}}, \card{G_{i2}}, \card{G_{i r_i}}} = \set{x_{i1}, x_{i2},\ldots, x_{i r_i}}$ is a partition of $x_i$. 
 As $G$ contains the distinct components $G_{11}, G_{12},\ldots, G_{1 r_1}, G_{21}, G_{22},\ldots, G_{2 r_2},\ldots,G_{t1}, G_{t2},\ldots, G_{t r_t}$, $X'$ is $G$-valid, where $X' = \bigcup_{i \in [t]} X_i$. 
 Also, observe that the graph $G + F$ is precisely a completion of $G$ with respect to $(X, \set{X_1, X_2,\ldots, X_t})$; recall that we obtain $G + F$ from $G$ by turning each $\bigcup_{j \in [r_i]} V(G_{ij})$ into a clique for every $i \in [t]$. 
 And as $G + F$ is $\eta$-balanced, we can conclude that a completion of $G$ with respect to $(X; X_1, X_2,\ldots, X_t)$ is $\eta$-balanced. 
 Finally, as observed earlier, for each $i \in [t]$, we have $\card{F_i} = \sum_{\substack{j, j' \in [r_i] \\ j \neq j'}} \card{G_{ij}} \card{G_{ij'}} = \sum_{\substack{j, j' \in [r_i] \\ j \neq j'}} x_{ij} x_{ij'} = \spp(X_i)$; and $k \geq \card{F} = \sum_{i \in [t]} \card{F_i} = \sum_{i \in [t]} \spp(X_i)$. This completes the proof of the lemma. 
\end{proof}

\subsubsection*{Algorithm for BCC}
We are now ready to describe our algorithm, which we call \algobcc. 

\subparagraph*{\algobcc.}Given an instance $(G, k, \eta)$ of \bcc\ as input, we proceed as follows. 
Recall that $G$ is a cluster graph. 
\begin{description}
\item[Step 1.] If $G$ is  $\eta$-balanced, then we return that $(G, k, \eta)$ is a yes-instance, and terminate. 

\item[Step 2.] If $k \leq 0$, then we return that $(G, k, \eta)$ is a no-instance, and terminate. 

\item[Step 3.] We use the algorithm of Proposition~\ref{prop:partition-number} to generate all partitions of $\ell$ for all $\ell \in [2k]$.

\item[Step 4.] For each $\ell \in [2k] \setminus \set{1}$, we do as follows.

\item[Step 4.1.] For each partition $X$ of $\ell$, we do as follows. 

\item[Step 4.1.1.] Let $X = \set{x_1, x_2,\ldots, x_t}$. For each multiset  $\set{X_1, X_2,\ldots, X_t}$, where $X_i$ is a partition of $x_i$ for every $i \in [t]$, we do as follows. 

\item[Step 4.1.1.1.] We  consider the partition $X' = \bigcup_{i \in [t]} X_i$ of $\ell$. If $X'$ is $G$-valid, $\sum_{i \in [t]}\spp(X_i) \leq k$, and a completion of $G$ with respect to $(X, \set{X_1, X_2,\ldots, X_t})$ is $\eta$-balanced, then we return that $(G, k, \eta)$ is a yes-instance, and terminate. 

\item[Step 5.] We return that $(G, k, \eta)$ is a no-instance, and terminate. 
\end{description}
The correctness of the algorithm follows from Observation~\ref{obs:completion-fpt-correctness} and Lemma~\ref{lem:completion-fpt-correctness}. We now analyse the running time. 
\begin{lemma}
\label{lem:completion-fpt-runtime}
\algobcc\ runs in time \algobccruntime. 
\end{lemma}
\begin{proof}
Observe that Steps 1, 2, 5  take only polynomial time. 
So does each execution of Step 4.1.1.1 (Observation~\ref{obs:completion-fpt-runtime}). 
By Proposition~\ref{prop:partition-number}, Step 3 takes time $2^{\cO(\sqrt{k})}$. 
Now, in Step 4, we only need to go over at most $2k - 1$ choices of $\ell$; and we can assume without loss of generality that $k \leq \binom{n}{2}$. 
That is, there are only $n^{\cO(1)}$ choices for $\ell$. 
So to prove the lemma, it is enough to prove that for each choice of $\ell$, the number of choices for $(X, \set{X_1, X_2,\ldots, X_t})$ is at most $2^{\cO(\ell)}$; Lemma~\ref{lem:completion-bound} proves precisely this. The lemma now follows from the fact that we only consider $\ell \leq 2k$ in Step 4. 
\end{proof}

We have thus proved the following result. 
\begin{theorem}
\label{thm:completion-fpt}
\bcc\ admits an algorithm that runs in time \algobccruntime. 
\end{theorem}

\subsubsection*{Towards our algorithm for BCE}
We now deal with \bce, the editing version. Specifically, we prove the following theorem. 
\begin{theorem}
\label{thm:editing-fpt}
\bce\ admits an algorithm that runs in time \algobccruntime. 
\end{theorem}
To prove this theorem, we first argue that we can solve \bce\ in time \algobccruntime\ when the input graph is a cluster graph. Specifically, we prove the following theorem. 
\begin{theorem}
\label{thm:editing-fpt-cluster}
\bce\ on cluster graphs admits an algorithm that runs in time \algobccruntime. 
\end{theorem}

Assuming Theorem~\ref{thm:editing-fpt-cluster}, let us first complete the proof of Theorem~\ref{thm:editing-fpt}. 
\begin{proof}[Proof of Theorem~\ref{thm:editing-fpt}]
Recall that a graph $G$ is a cluster graph if and only if $G$ does not contain $P_3$ as an induced subgraph. Given an instance $(G, k, \eta)$ of \bce, our algorithm works as follows. 
\begin{description}
    \item[Step 1] If $G$ is a cluster graph, then we use the algorithm of Theorem~\ref{thm:editing-fpt-cluster} to solve the problem. 
    \item[Step 2]  If Step 1 is not applicable and $k > 0$, then we branch on induced $P_3$s; more precisely, we greedily find an induced $P_3$, say $uvw$, (i.e., $u, v, w \in V(G)$ with $uv, vw \in E(G)$ and $uw \notin E(G)$) and recursively call our algorithm on the three instances $(G - uv, k - 1, \eta)$, $(G - vw, k - 1, \eta)$ and $(G + uw, k - 1, \eta)$. 
    \item[Step 3] If Steps 1 and 2 are not applicable, then $G$ is not a cluster graph and $k = 0$, and we return that $(G, k, \eta)$ is a no-instance of \bce.
\end{description} 
The correctness of this algorithm follows from the fact that we must modify every induced $P_3$ $uvw$ in $G$ by either deleting from $G$ one of the edges $uv$ and $vw$ or by adding the edge $uw$ to $G$. As for the running time, Step 3 runs in polynomial time; we make $3$ recursive calls in each execution of Step 2, and we recurse only until $k = 0$, resulting in at most $3^k$ recursive calls; by Theorem~\ref{thm:editing-fpt-cluster}, Step 1 takes time \algobccruntime. Thus the overall running time is bounded by \algobccruntime. 
\end{proof}

We now prove Theorem~\ref{thm:editing-fpt-cluster}, for which we adopt a familiar strategy: We associate solutions with partitions of an integer and use the bound for the number of partitions of an integer. 

\subparagraph*{Idea behind our algorithm for \bce\ on cluster graphs.} Consider an instance $(G, k, \eta)$ of \bce, where $G$ is a cluster graph. Let $F \subseteq \binom{V(G)}{2}$ be a solution for $(G, k, \eta)$, and let $F_1 \subseteq F$ be the edges that we delete from $G$ and $F_2 \subseteq F$ be the set of edges that we add to $G$. We may think of the cluster graph $G \triangle F$ as obtained from $G$ by a two-step process: first by deleting from $G$ the edges in $F_1$ and then by adding to $G - F_1$ the edges in $F_2$. Let $G' = G - F_1$. Then $G'$ must be a cluster graph, and we have  $G \triangle F = (G - F_1) + F_2 = G' + F_2$. Let $\ell_1 = \card{V(F_1)}$ and $\ell_2 = \card{V(F_2)}$. Similar to what we did in our algorithm for \bcc, we can argue that $F_1$  corresponds to a partition $Y = \set{y_1, y_2,\ldots, y_{s}}$  of $\ell_1$ and a multiset $\set{Y_1, Y_2,\ldots, Y_s}$,  and $F_2$ corresponds to a partition $X = \set{x_1, x_2,\ldots, x_{t}}$ of $\ell_2$ and a corresponding  multiset $\set{X_1, X_2,\ldots, X_t}$. Informally, each $Y_i = \set{y_{i1}, y_{i2},\ldots, y_{i q_i}}$ is a partition of $y_i$ and $Y_i$ corresponds to a component of $G$ of size exactly $y_i$ that gets split into components of sizes $y_{i1}, y_{i2},\ldots, y_{q_i}$ when we delete $F_1$. Similarly, each $X_i = \set{x_{i1}, x_{i2},\ldots, x_{i r_i}}$ corresponds to a component of $G' + F_2 = G \triangle F$ of size $x_i$ that was formed by  merging $r_i$  components of $G'$ of sizes $x_{i1}, x_{i2},\ldots, x_{i r_i}$. \lipicsEnd

 To formalise the above idea, we now define a deletion of cluster graph with respect to a multiset of partitions, similar to a completion of a cluster graph with respect to a multiset of partitions that we defined earlier. 

\subparagraph*{Deletion of a cluster graph w.r.t. a multiset of partitions.}
Consider a cluster graph $G$ and a positive integer $\ell$. Let $Y = \set{y_1, y_2,\ldots, y_s}$ be a partition of $\ell$. Recall that we say that $Y$ is $G$-valid if $G$ contains $s$ distinct components $G_1, G_2,\ldots, G_s$ such that $\card{G_j} = y_j$ for every $j \in [s]$. Now, consider $G, \ell$ and $X = \set{y_1, y_2,\ldots, y_s}$ such that $Y$ is $G$-valid. Let $Y_1, Y_2,\ldots, Y_s$ be such that for each $j \in [s]$, $Y_j = \set{y_{j1}, y_{j2}, \ldots, y_{j q_j}}$ is a partition of $y_j$. Since $Y$ is $G$-valid, $G$ contains $s$ distinct components $G_1, G_2,\ldots, G_s$ with $\card{G_j} = y_j$ for each $j \in [s]$. Fix such a set of $s$ components; recall that each $G_j$ is a clique. By a deletion of $G$ with respect to $(Y, \set{Y_1, Y_2,\ldots, Y_s})$, we mean the subgraph of $G$ obtained by deleting edges so that each $G_j$ is turned into a cluster graph with exactly $q_j$ connected components of sizes $y_{j1}, y_{j2},\ldots, q_{j q_j}$. Notice that if $q_j = 1$, i.e., $Y_j = \set{y_j}$, then $G_j$ itself is the required cluster graph and we do not need to delete any edges; otherwise $q_j \geq 2$, and we need to delete exactly $\sum_{\substack{i, i' \in [q_j] \\ i \neq i'}} y_{ji} y_{ji'}$ edges. In either case, we need to delete exactly $\spp(Y_j)$ edges to turn $G_j$ into a cluster graph with the required component sizes. Thus a deletion of $G$ with respect to $(Y, \set{Y_1, Y_2,\ldots, Y_s})$ requires the deletion of exactly $\sum_{j \in [s]} \spp(Y_j)$ edges. 

\subsubsection*{The Main Technical Lemma for \bce\ on cluster graphs}
\begin{lemma}
\label{lem:editing-fpt-correctness}
Consider an instance $(G, k, \eta)$ of \bce, where $G$ is a cluster graph. 
If $(G, k, \eta)$ is a yes-instance, then one of the following statements holds. 
\begin{enumerate}
    \item The instance $(G, k, \eta)$ is a yes-instance of \bcc. 
    \item The instance $(G, k, \eta)$ is a yes-instance of \bcd. 
    \item There exist integers $\ell_1, \ell_2 \in [2k] \setminus \set{1}$ and  partitions $Y = \set{y_1, y_2,\ldots, y_{s}}$ of $\ell_1$ and $X = \set{x_1, x_2,\ldots, x_t}$ of $\ell_2$ such that there exist  multisets $\set{Y_1, Y_2,\ldots, Y_s}$ and $\set{X_1, X_2,\ldots, X_t}$ of partitions, where  $Y_j$ is a partition of $y_j$ for every $j \in [s]$ and $X_i$ is a partition of $x_i$ for each $i \in [t]$, with the following properties: 
\begin{enumerate}
\item $Y$ is $G$-valid; 

\item $X'$ is $G'$-valid, where $X' = \bigcup_{i \in [t]} X_i$ and $G'$ is a deletion of $G$ with respect to $(Y, \set{Y_1, Y_2,\ldots, Y_s})$; 

\item a completion of $G'$ with respect to $(X, \set{X_1, X_2,\ldots, X_t})$ is $\eta$-balanced; and

\item $\sum_{j \in [s]}\spp(Y_j) + \sum_{i \in [t]}\spp(X_i) \leq k$.
\end{enumerate}
\end{enumerate}
\end{lemma}

\begin{proof}[Proof Sketch]
Assume that $(G, k, \eta)$ is a yes-instance of \bce, where $G$ is a cluster graph.  
If $(G, k, \eta)$ is a yes-instance of \bcc\ or a yes-instance of \bcd, then the lemma trivially holds. 
So, assume that neither of those is true. 
Let $F \subseteq \binom{V(G)}{2}$ be a solution for $(G, k, \eta)$, and let $F_1 = F \cap E(G)$ and $F_2 = F \setminus F_1$. That is, $F = F_1 \cup F_2$ and in particular, $F_2 = F \cap (\binom{V(G)}{2} \setminus E(G))$. 
 
Let $\ell_1 = \card{V(F_1)}$ and $\ell_2 = \card{V(F_2)}$; that is, $\ell_1$ is the number of vertices that $F_1$ modifies and $\ell_2$ is the number of vertices that $F_2$ modifies. 
As $G$ is not a yes-instance of \bcd, $F_1 \neq \emptyset$, and as $(G, k, \eta)$ is not a yes-instance of \bcc, $F_2 \neq \emptyset$, and hence $\ell_1 = \card{V(F_1)} \geq 2$ and $\ell_2 = \card{V(F_2)} \geq 2$ 
Also, as $\card{F_1} \leq \card{F} \leq k$ and $\card{F_2} \leq \card{F} \leq k$, we have $\ell_1  = \card{V(F_1)} \leq 2k$ and  $\ell_2  = \card{V(F_2)} \leq 2k$. Thus $\ell_1, \ell_2 \in [2k] \setminus \set{1}$. 

We may think of the cluster graph $G \triangle F$ as obtained from $G$ by a two-step process: first by deleting from $G$ the edges in $F_1$ and then by adding to $G - F_1$ the edges in $F_2$. Let $G' = G - F_1$. Then $G'$ is a cluster graph, and we have  $G \triangle F = (G - F_1) + F_2 = G' + F_2$. 
We can argue that $F_1$  corresponds to a partition $Y = \set{y_1, y_2,\ldots, y_{s}}$  of $\ell_1$ and the corresponding multiset $\set{Y_1, Y_2,\ldots, Y_s}$,  and $F_2$ corresponds to a partition $X = \set{x_1, x_2,\ldots, x_{t}}$ of $\ell_2$ and the corresponding  multiset $\set{X_1, X_2,\ldots, X_t}$ of partitions that satisfy the properties required by the lemma. 
Intuitively, the  partition $Y_j = \set{y_{j 1}, y_{j 2},\ldots, y_{j q_j}}$ of $y_j$ corresponds to a component of $G$ of size exactly $y_j$ that was split into components of sizes $y_{j 1}, y_{j 2},\ldots, y_{j q_j}$ by deleting the edges in $F_1$. Consequently, $G'$ is a deletion of $G$ with respect to $(Y, \set{Y_1, Y_2,\ldots, Y_s})$. Each $X_i = \set{x_{i1}, x_{i2},\ldots, x_{ir_i}}$ corresponds to the components of $G'$ of sizes $x_{i1}, x_{i2},\ldots, x_{ir_i}$ that were merged into a single component of $G' + F_2 = G \triangle F$ of size $x_i$. 
That is, $G' + F_2 = G \triangle F$ is a completion of $G'$ with respect to $(X, \set{X_1, X_2,\ldots, X_t})$. 
These arguments will also imply  that  $\sum_{j \in [s]} \spp(Y_j) = \card{F_1}$ and $\sum_{i \in [t]} \spp(X_i) = \card{F_2}$, and thus $\sum_{j \in [s]} \spp(Y_j) + \sum_{i \in [t]} \spp(X_i) = \card{F_1} + \card{F_2} = \card{F} \leq k$. 
\end{proof}

\subsubsection*{Algorithm for BCE on cluster graphs}

We now design an algorithm for \bce\ on cluster graphs, which we call \algobcec. The correctness of the algorithm follows from Lemma~\ref{lem:editing-fpt-correctness}. 

\subparagraph*{\algobcec.}Given an instance $(G, k, \eta)$ of \bce\ as input, where $G$ is a cluster graph, we proceed as follows.  
\begin{description}
\item[Step 1.] If $(G, k, \eta)$ is a yes-instance of \bcc\ or a yes-instance of \bcd, we return that $(G, k, \eta)$ is a yes-instance of \bce, and terminate. To do this, we use the algorithms of Theorems~\ref{thm:completion-fpt} and \ref{thm:deletion-fpt}. 

\item[Step 2.] If $k \leq 0$, then we return that $(G, k, \eta)$ is a no-instance, and terminate. 

\item[Step 3.] We use the algorithm of Proposition~\ref{prop:partition-number} to generate all partitions of $\ell$ for all $\ell \in [2k]$.

\item[Step 4.] For every choice of  $\ell_1, \ell_2 \in [2k] \setminus \set{1}$, we do as follows.

\item[Step 4.1.] For each choice of partitions $Y$ of $\ell_1$ and $X$ of $\ell_2$, we do as follows. 

\item[Step 4.1.1.] Let $Y = \set{y_1, y_2,\ldots, y_s}$ and $X = \set{x_1, x_2,\ldots, x_t}$. For each choice of multisets $Y = \set{Y_1, Y_2,\ldots, Y_s}$ and $X = \set{X_1, X_2,\ldots, X_t}$, where $Y_j$ is a partition of $y_j$ for every $j \in [s]$ and $X_i$ is a partition of $x_i$ for every $i \in [t]$, we do as follows. 

\item[Step 4.1.1.1.] If $Y$ is $G$-valid, then we construct the graph $G'$, where $G'$ is a deletion of $G$ with respect to $(Y, \set{Y_1, Y_2,\ldots, Y_s})$, and we consider the partition $X' = \bigcup_{i \in [t]} X_i$ of $\ell_2$. If $X'$ is $G'$-valid, a completion of $G'$ with respect to $(X, \set{X_1, X_2,\ldots, X_t})$ is $\eta$-balanced and $\sum_{j \in [s]}\spp(Y_j) + \sum_{i \in [t]}\spp(X_i) \leq k$, then we return that $(G, k, \eta)$ is a yes-instance, and terminate. 

\item[Step 5.] We return that $(G, k, \eta)$ is a no-instance, and terminate. 
\end{description}

We now analyse the running time of \algobcec. 

\begin{lemma}
\label{lem:editing-fpt-runtime}
\algobcec\ runs in time \algobccruntime. 
\end{lemma}
\begin{proof}
Observe that Steps 2 and 5 take only polynomial time. 
So does each execution of Step 4.1.1.1. 
By Theorems~\ref{thm:completion-fpt} and \ref{thm:deletion-fpt}, Step 1 takes time \algobccruntime. By Proposition~\ref{prop:partition-number}, Step 3 takes time $2^{\cO(\sqrt{k})}$. 
Now, in Step 4, we only need to go over at most $2k - 1$ choices of $\ell_1$ and at most $2k - 1$ choices of $\ell_2$; and we can assume without loss of generality that $k \leq \binom{n}{2}$. 
That is, there are only $n^{\cO(1)}$ choices for the pair $(\ell_1, \ell_2)$. In Step 4.1, we go over all possible choices of $(Y, X)$. As $Y$ is a partition of $\ell$, by Proposition~\ref{prop:partition-number}, $Y$ has $p(\ell_1) = 2^{\cO(\sqrt{\ell_1})}$ choices. Similarly,$X$ has $p(\ell_2) = 2^{\cO(\sqrt{\ell_2})}$ choices. Thus, there are  $2^{\cO(\sqrt{\ell_1})} \cdot 2^{\cO(\sqrt{\ell_2})} = 2^{\cO(\sqrt{\ell_1} + \sqrt{\ell_2})} = 2^{\cO(\sqrt{k})}$ choices for the pair $(Y, X)$.  
 Finally, by Lemma~\ref{lem:completion-bound}, for each choice of $(\ell_1, \ell_2)$, the number of choices for $((Y, \set{Y_1, Y_2,\ldots, Y_s}), (X, \set{X_1, X_2,\ldots, X_t}))$  is at most $2^{\cO(\ell_1)} \cdot 2^{\cO(\ell_2)} = 2^{\cO(k)}$. Thus, the overall running time of the algorithm is bounded by \algobccruntime.
\end{proof}
This completes the proof of Theorem~\ref{thm:editing-fpt-cluster}.

\input{algebraic-algorithms}

%% file: algebraic-algorithms.tex
\subsection{Faster Algorithms for BCC and BCE}
\label{sec:fast} 
In Section~\ref{sec:fpt-completion-editing}, we designed algorithms for \bcc\ and \bce\ that run in time \algobccruntime. By refining the ideas that we used for these algorithms, 
we now design algorithms for \bcc\ and \bce\ that run in time \fastbccruntime.

\subparagraph*{The high level idea.} We use \bcc\ to illustrate the idea behind our algorithms; this idea extends to \bce\ as well.  %
Consider a yes-instance $(G, k, \eta)$ of \bcc, and suppose $F \subseteq \binom{V(G)}{2} \setminus E(G)$ is the solution that we are looking for; and $G_1, G_2,\ldots, G_t$ are the components of the $\eta$-balanced cluster graph $G + F$ that were modified by $F$, and for each $i \in [t]$, $G_{i1}, G_{i2},\ldots, G_{ir_i}$ are the components of $G$ that were merged together to form the component $G_i$ of $G$.  In our previous algorithm for \bcc\ (\algobcc), we guessed $\card{V(F)}$; then we guessed $X = \set{\card{G_1}, \card{G_2},\ldots, \card{G_t}}$, and for each $i \in [t]$, we guessed $X_i = \set{\card{G_{i1}}, \card{G_{i2}}, \ldots, \card{G_{i r_i}}}$. As $X$ is a partition of $\card{V(F)} \in [2k]$, and as the number of partitions of any  integer $\ell$ is $2^{o(\ell)}$ (Proposition~\ref{prop:partition-number}),  we can guess $X$ in time $2^{o(k)}$. But  guessing each $X_i$ is what led to the $2^{\cO(k)}$ factor in the running time (Lemma~\ref{lem:completion-bound}); we now design a faster algorithm by avoiding this time-consuming guess. To do this, we guess $\card{V(F)}$ and $X$ as before, and we guess $X' = \bigcup_{i \in [t]} X_i$ instead of guessing each $X_i$ separately. Notice that as $X$ and $X'$ are partitions of $\card{V(F)} \in [2k]$, we only make $2k \cdot 2^{\cO(\sqrt{k})} \cdot 2^{\cO(\sqrt{k})} = 2^{o(k)}$ guesses. Once we make our guesses, we check in time $2^{2k} = 4^k$ if we can merge the components by adding at most $k$ edges; this is the more non-trivial step of our algorithm. 

 To check if we can merge the components, we think of the components $G_1, G_2,\ldots, G_t$ of $G + F$ as bins, and for all $i \in [t]$ the components $G_{i1}, G_{i2},\ldots, G_{ir_i}$  of $G$ as balls, where the size a ball (or the capacity of a bin) is the size of the corresponding component. And in the balls-and-bins parlance, merging the components of $G$ to form the components of $G + F$ simply means assigning balls to bins so that the total size of balls assigned to each bin does not exceed the capacity of that bin. Also, notice that to assign a ball of size $\card{G_{ij}}$ to a bin of capacity $\card{G_i}$, we need to add exactly $\card{G_{ij}} (\card{G_i} -  \card{G_{ij}})$ edges incident with $V(G_{ij})$; we think of this quantity as the cost of assigning this ball to this bin. So guessing the partitions $X$ and $X'$ means guessing the capacities of the bins and the sizes of the balls. And for each $(X, X')$, we simply need to check if we can assign the balls to the bins subject to capacity and cost constraints; we discuss below how we can do this in time $4^k n^{\cO(1)}$. Notice that as $X$ and $X'$ are partitions of $\card{V(F)} \leq 2k$, we have $\card{X}, \card{X'} \leq 2k$; that is, we have at most $2k$ balls and at most $2k$ bins. \lipicsEnd

 To exploit these ideas and design our algorithms for \bcc\ and \bce, we first define an auxiliary problem called \binbfull\ (\binb) and design an algorithm for \binb. We then show that we can reduce \bcc\ and \bce\ to \binb. 

\subsubsection{An Algorithm for  \binbfull}
 Consider a set of $s$ balls (indexed by $1, 2,\ldots, s$) with sizes $b_1, b_2,\ldots, b_s \in \mathbb{N}$, a set of $t$ bins (indexed by $1, 2,\ldots, t$) with capacities $x_1, x_2, \ldots, x_t \in \mathbb{N}$. We are interested in assigning these balls to the bins in such a way that for each bin $j$, the sum of the sizes of balls assigned to bin $j$ does not exceed the capacity of that bin. Formally, consider $B \subseteq [s]$. We define the volume of $B$, denoted by $\vol(B)$, to be the sum of the sizes of the balls in $B$, i.e., $\vol(B) = \sum_{i \in B} b_i$ if $B \neq \emptyset$ and  $\vol(B) = 0$ if $B = \emptyset$.  
By an assignment of the balls in $B$ to the $t$ bins (or simply an assignment of $B$, for short), we mean a function $\fn{\beta}{B}{[t]}$ such that $\vol(\beta^{-1}(j)) \leq x_j$ for every $j \in [t]$. %
If $\beta(i) = j$ for $i \in [s]$ and $j \in [t]$, then we say that  ball $i$ is assigned to bin $j$ (or that the $i$th ball is assigned to the $j$th bin) under the assignment $\beta$. %
Now, consider a cost function $\fn{\cost}{[s] \times [t]}{\mathbb{N}}$, where for every $i \in [s]$ and $j \in [t]$, $\cost(i, j)$ is the cost of assigning the $i$th ball to the $j$th bin. 
For $B' \subseteq B \subseteq [s]$, and an assignment $\fn{\beta}{B}{[t]}$, the cost of $B'$ under the assignment $\beta$ is $\sum_{i \in B'} \cost(i, \beta(i))$, %
and we denote this quantity by $\cost_{\beta}(B')$; when the assignment $\beta$ is clear from the context, we may omit the subscript and simply write $\cost(B')$. 

We now define the following problem. 

\defproblem{\binbfull\ (\binb)}{A set of $s$ balls (indexed by $1, 2,\ldots, s$) with sizes $b_1, b_2,\ldots, b_s \in \mathbb{N}$, a set of $t$ bins (indexed by $1, 2,\ldots, t$) with capacities $x_1, x_2, \ldots, x_t \in \mathbb{N}$, a budget $W \in \mathbb{N}$, and a cost function $\fn{\cost}{[s] \times [t]}{\mathbb{N} \cup \set{0}}$ where $\cost(i, j)$ is the cost of assigning the $i$th ball to the $j$th bin for every $i \in [s]$ and $j \in [t]$.}{Decide if there exists an assignment of the balls in $[s]$ to the $t$ bins such that the cost of the assignment is at most $W$.}
Observe that we can design a straightforward dynamic programming algorithm for \binb\ that runs in time $\cO^{\star}(3^s)$ (assuming $\cost(\cdot, \cdot)$ and $W$ are encoded in unary):  for each subset $B \subseteq [s]$, $j \in [t]$ and $q \in [W]_0$, we simply need to check if there is an assignment of $B$ to the first $j$ bins with cost at most $q$.  To execute this, we need to go over all subsets $B$ of $[s]$, and check if there exist a subset $B' \subseteq B$ and $q' \leq q$ such that we can assign $B'$ to the $j$th bin with cost at most $q'$ and assign $B \setminus B'$ to the first $j - 1$ bins with cost at most $q - q'$, which will take time $\cO^{\star}(\sum_{B \subseteq [s]} 2^{\card{B}}) = \cO^{\star}(\sum_{i = 0}^s \binom{s}{i} 2^i) = \cO^{\star}(3^s)$. We now show that we can execute this idea in time $\cO^{\star}(2^s)$ by using fast polynomial multiplication.\footnote{We can also use the fast subset convolution (FSC) algorithm of Bj{\"o}rklund et al.~\cite{DBLP:conf/stoc/BjorklundHKK07} to accelerate this DP to run in time $\cO^{\star}(2^s)$. But instead of using FSC, we resort to fast Fourier transform (FFT)-based polynomial multiplication to achieve the same running time. While both tools guarantee the running time of $\cO^{\star}(2^s)$, FFT is arguably preferable to FSC when it comes to implementation.  
As noted by Cygan and Pilipczuk~\cite{DBLP:journals/tcs/CyganP10}, FFT is widely known in the computer science community, with a number of libraries readily available for its implementation~\cite{gambron2020comparison}. Also, FSC relies on M{\"o}bius transforms, which requires repeated additions and subtractions, and this can lead to large rounding errors in floating point arithmetic~\cite{doi:10.1080/03610926.2014.894070}.} Specifically, we prove the following theorem.\footnote{\binbfull\ is a special case of the the {\sc Generalised Assignment Problem (GAP)}, which is often phrased in terms of assigning jobs to machines. While GAP has been studied intensively in the approximation algorithms framework, where the typical goals have been  minimising the cost and makespan (the maximum time taken by any machine to complete the jobs assigned to it), we are unaware of an exact exponential time algorithm for GAP, along the lines of our algorithm in Theorem~\ref{thm:binb}. We refer the reader to Williamson and Shmoys~\cite[Chapter 11]{DBLP:books/daglib/0030297} for approximation algorithms for GAP and its variants; and to Kundakcioglu and Alizamir~\cite{DBLP:reference/opt/KundakciogluA09} for an overview of literature on GAP.} 
\begin{theorem}
\label{thm:binb}
\binbfull\ admits an algorithm that runs in time \algobinbruntime. 
\end{theorem}

To prove Theorem~\ref{thm:binb}, we rely on the fact that we can efficiently multiply two polynomials.  
\begin{proposition}[\cite{DBLP:conf/issac/Moenck76}]
\label{prop:polynomial-multiplication}
The product of two polynomials of degree $d$ can be computed in time $\cO(d \log d)$. 
\end{proposition}

We now introduce the following notation and terminology, which we borrow from~\cite{DBLP:journals/tcs/CyganP10,DBLP:conf/ijcai/Gupta00T21}.  
\subparagraph*{Notation and Terminology.} Consider $B \subseteq [s]$. The characteristic vector of $B$, denoted by $\chi(B)$,  is the $s$-length binary string whose $i$th bit is $1$ if and only if $i \in B$. The Hamming weight of a binary string is the number of $1$s in that string. Notice that we can interpret a binary string as the binary encoding of an integer and vice-versa. For example, we interpret the string $(1, 0, 0, 1, 0)$ to be the number 18 as the binary encoding of 18 is exactly $10010$. For $B \subseteq [s]$, consider the monomial $z^{\chi(B)}$, where $z$ is an indeterminate and $\chi(B)$ is the characteristic vector of $B$, interpreted as the binary encoding of an integer. The Hamming weight of a monomial $z^{i}$ is the Hamming weight of the binary encoding (interpreted as a binary string) of $i$. Consider a polynomial $P(z)$ and an integer $i$. The Hamming projection of $P(z)$ to $i$, denoted by $\ca{H}_{i}(P(z))$, is the sum of the monomials of $P$ of Hamming weight exactly $i$. That is, $\ca{H}_{i}(P(z))$ is the polynomial obtained from $P(z)$ by removing all those monomials of Hamming weight not equal to $i$. In particular, if $P(z)$ has no monomial of Hamming weight exactly $i$, then $\ca{H}_{i}(P(z))$ is the zero polynomial. Finally, we define the representative polynomial of $P(z)$, denoted by $\ca{R}(P(z))$, to be the polynomial obtained from $P(z)$ by replacing the coefficient of each (non-zero) monomial with $1$. For example, consider the polynomial $P(z) = z^8 + 2z^6 + z^5 + 3z^2 + 9$. Then $\ca{H}_{2}(P(z)) = 2z^6 + z^5$ as the binary encodings of 6 and 5 respectively are 0110 and 0101, and they both have Hamming weight 2, whereas the binary encodings of $8, 2$ and $0$ (the monomial $9$ is $9z^0$) respectively are 1000, 0010 and 0000, none of which has Hamming weight 2. And $\ca{R}(P(z)) = z^8 + z^6 + z^5 + z^2 + 1$. 

\subparagraph*{Outline of our algorithm for \binb.}
\begin{sloppypar}
Consider an instance of \binb\ with $s$ balls, $t$ bins and budget $W$. The idea is to encode all possible assignments of balls in $[s]$ to the $t$ bins as a polynomial, each monomial of which will correspond to an  assignment. More specifically, for each $i \in [s]_0, j \in [t], q \in [W]_0$, we define a polynomial $P_{i, j, q}(z)$, where each monomial of $P_{i, j, q}(z)$ will  correspond to an  assignment of a subset of exactly $i$ balls to the first $j$ bins so that the total cost of this assignment is at most $q$ (if at least one such  assignment exists, in which case each monomial will have Hamming weight exactly $i$; otherwise, $P_{i, j, q}(z)$ will just be the zero polynomial). In particular, for $B \subseteq [s]$ with $\card{B} = i$, the polynomial $P_{i, j, q}(z)$ will contain the monomial $z^{\chi(B)}$ if and only if there is an assignment of $B$ to the first $j$ bins with cost at most $q$. We compute these polynomials iteratively by going from smaller to larger values of $j$.  The polynomial that we ultimately need to compute is $P_{s, t, W}(z)$, and we will have the guarantee that the given instance of \binb\ is a yes-instance if and only if $P_{s, t, W}(z)$ contains the monomial $z^{\chi([s])}$.   \lipicsEnd 
\end{sloppypar}

Before proceeding to the proof of Theorem~\ref{thm:binb}, we state one crucial result that we will rely on to establish the correctness of our algorithm. In our algorithm, we only check if there is an assignment of the balls to the bins; we do not guess which subset of balls is assigned to which bin. We must nonetheless ensure that the subsets of balls assigned to different bins are indeed disjoint. For that, we use the following result, which hardwires  this disjointness requirement into the the polynomial $P_{i, j, q}(z)$ that we will define.    
\begin{proposition}[\cite{DBLP:journals/tcs/CyganP10}]
\label{prop:disjoint-monomial}
Let $B_1, B_2 \subseteq [s]$. Then $B_1$ and $B_2$ are disjoint if and only if the Hamming weight of the monomial $z^{\chi(B_1) + \chi(B_2)}$ is $\card{B_1} + \card{B_2}$.
\end{proposition}

We are now ready to prove Theorem~\ref{thm:binb}. 
\begin{proof}[Proof of Theorem~\ref{thm:binb}]

Consider an instance of \binb, where we are given  $s$ balls (indexed by $1, 2,\ldots, s$) with sizes $b_1, b_2,\ldots, b_s$, $t$ bins (indexed by $1, 2,\ldots, t$) with capacities $x_1, x_2,\ldots, x_t$, a cost function $\fn{\cost}{[s] \times [t]}{\mathbb{N} \cup \set{0}}$ and a budget $W \in \mathbb{N}$. 

For each $i \in [s]_0, j \in [t], q \in [W]_0$, we will define a polynomial $P_{i, j, q}(z)$, where each monomial of $P_{i, j, q}(z)$ will correspond to a (partial) assignment of a set of exactly $i$ balls to the first $j$ bins such that the cost of the assignment is at most $q$. To that end, we will first define an auxiliary polynomial $A_{i, j, q}(z)$, each monomial of which will correspond to an assignment of a set of exactly $i$ balls to the $j$th bin such that the  cost of the assignment is at most $q$. 
Formally, for every $i \in [s]_0, j \in [t], q \in [W]_0$, we define 
\[
A_{i, j, q}(z) = \sum_{\substack{B \subseteq [s] \\ \card{B} = i \\ \vol(B) \leq x_j \\ \sum_{r \in B} \cost(r, j) \leq q}} z^{\chi(B)}. 
\]
Notice that $A_{i, j, q}(z)$ contains the monomial $z^{\chi(B)}$ if and only if $\card{B} = i$ and the balls in $B$ can be assigned to the $j$th bin with cost at most $q$. We now define $P_{i, j, q}(z)$ for every $i \in [s]_0, j \in [t], q \in [W]_0$ as follows. We define $P_{i, 1, q}(z) = A_{i, 1, q}(z)$. For $j > 1$, notice that there exists an assignment of $i$ balls to the first $j$ bins with cost at most $q$ if and only if there exist an assignment of exactly $i'$ balls to $j$th bin with cost at most $q'$ and an assignment of $i - i'$ balls to the first $j - 1$ bins with cost at most $q - q'$, for some $i' \in [i]_0$ and $q' \in [q]_0$. We thus define   
\begin{equation*}
P_{i, j, q}(z) = \ca{R}\Lb{\ca{H}_{i}\Lb{\sum_{(i', q') \in [i]_0 \times [q]_0} A_{i', j, q'}(z) \cdot P_{i - i', j - 1, q - q'}(z)}} \tag{$\star \star \star$}\label{eq:poly},
\end{equation*}
where $\ca{H}_{i}(\cdot)$ is the Hamming projection to $i$ and $\ca{R}(\cdot)$ is the representative polynomial. The $\ca{H}_i(\cdot)$ operator ensures that $P_{i, j, q}(z)$ contains only monomials of Hamming weight exactly $i$ (which correspond to assignments of exactly $i$ balls). And the $\ca{R}(\cdot)$ operator ensures that the coefficients of the polynomials we multiply always remain $1$. In particular, the coefficient of each (non-zero) monomial of $P_{i, j, q}(z)$ (and $A_{i, j, q}(z)$) is $1$.  

We first compute $A_{i, j, q}(z)$ for every $i \in [s]_0, j \in [t]$ and $q \in [W]_0$ by going over all subsets $B \subseteq [s]$. We then compute $P_{i, j, q}(z)$ using Equation~(\ref{eq:poly}). In particular, we use the algorithm of Proposition~\ref{prop:polynomial-multiplication} to compute the product $A_{i', j, q'}(z) \cdot P_{i - i', j - 1, q - q'}(z)$ in Equation~(\ref{eq:poly}) for every choice of $(i', q')$. Finally, we return that the given instance of \binb\ is a yes-instance if and only if the polynomial $P_{s, t, W}(z)$ contains the monomial $z^{\chi([s])}$. (Notice that if the given instance of \binb\ is a yes-instance, then we can find an assignment of $[s]$ with cost at most $W$ by backtracking from $P_{s, t, W}(z)$.)

Let us now analyse the running time of our algorithm. Since $i \in [s]_0, j \in [t]$ and $q \in [W]_0$, we compute $\cO(s \cdot t \cdot W)$ polynomials. Each polynomial has degree at most $2^s$ as the number represented by any $s$-length binary string is at most $2^s$. Also, to compute $P_{i, j, q}(z)$ using Equation~(\ref{eq:poly}), we go over $\cO(s \cdot W)$ choices of $(i', q')$, and for each such choice, multiply two polynomials of degree at most $2^s$; by Proposition~\ref{prop:polynomial-multiplication}, we can multiply two polynomials of degree at most $2^s$ in time $\cO(2^s \cdot s)$. Finally, notice that we can apply the operators $\ca{R}(\cdot)$ and $\ca{H}_i(\cdot)$ in time proportional to $2^s$. Hence the total running time is bounded by $\cO((s \cdot t \cdot W) \cdot  (s \cdot W) \cdot (2^s \cdot s)) = \algobinbruntime$. 

Finally, to establish the correctness of our algorithm, we prove the following claim. 

\begin{claim}
\label{claim:binb-correctness}
Consider $i \in [s], j \in [t]$ and $q \in [W]_0$. 
For every $B \subseteq [s]$ with $\card{B} = i$, there exists an assignment $\fn{\beta}{B}{[j]}$ with $\cost_{\beta}(B) \leq q$ if and only if the polynomial $P_{i, j, q}(z)$ contains the monomial $z^{\chi(B)}$.  
\end{claim}

Assuming for now that Claim~\ref{claim:binb-correctness} holds, let us complete the proof of the theorem. 
Recall that our algorithm returns yes if and only if the polynomial $P_{s, t, W}(z)$ contains the monomial $z^{\chi([s])}$, which by Claim~\ref{claim:binb-correctness}, holds if and only if there is an assignment of $[s]$ with cost at most $W$. We can thus conclude that our algorithm returns yes if and only if there is an assignment of $[s]$ with cost at most $W$. 

\begin{claimproof}[Proof of Claim~\ref{claim:binb-correctness}]
Fix $B \subseteq [s]$ with $\card{B} = i$. We prove the claim by induction on $j$. 

Consider $j = 1$. Observe that there exists an assignment $\beta$ that assigns $B$ to the first bin with cost at most $q$ if and only if $\vol(B) \leq x_1$ and $\cost_{\beta}(B) = \sum_{r \in B} \cost(r, 1) \leq q$, which, by the definition of $A_{i, 1, q}(z)$,  holds if and only if the polynomial $A_{i, 1, q}(z) = P_{i, 1, q}(z)$ contains the monomial $z^{\chi(B)}$. 

Assume now that the claim holds for $j = j' - 1$. We prove that it holds for $j = j'$. 

Suppose that there exists an assignment $\fn{\beta}{B}{[j']}$ with $\cost(B) \leq q$. Let $\hat B \subseteq B$ be the set of balls that are assigned to bin $j'$ under $\beta$. Let $\hat i = \card{\hat B}$ and $\hat q = \cost_{\beta}(\hat B)$.  %
First, as the $\hat i$ balls in $\hat B$ are assigned to bin $j'$ under $\beta$ with cost $\hat q$, the polynomial $A_{\hat i, j', \hat q}(z)$ contains the monomial $z^{\chi(\hat B)}$.
Second, by  the definition of $\hat B$, the $i - \hat i = \card{B \setminus \hat B}$ balls in $B \setminus \hat B$ are assigned to the first $j' - 1$ balls under $\beta$ with cost  $\cost(B) - \cost(\hat B) \leq q - \hat q$. 
Hence, by the induction hypothesis, the polynomial $P_{i - \hat i, j' - 1, q - \hat q}(z)$ contains the monomial $z^{\chi(B \setminus \hat B)}$.  
Therefore the product $A_{\hat i, j', \hat q}(z) \cdot P_{i - \hat i, j' - 1, q - \hat q}(z)$ contains the monomial $z^{\chi(\hat B) + \chi(B \setminus \hat B)}$. 
As $\card{\hat B} + \card{B \setminus \hat B} = i$ and $\cost(\hat B) + \cost(B \setminus \hat B) = \cost(B) \leq q$, the sum on the right hand side of Equation~(\ref{eq:poly}) contains the summand  $A_{\hat i, j', \hat q}(z) \cdot P_{i - \hat i, j' - 1, q - \hat q}(z)$, and hence the monomial $z^{\chi(\hat B) + \chi(B \setminus \hat B)}$. 
Finally, as the sets $\hat B$ and $B \setminus \hat B$ are disjoint, by Proposition~\ref{prop:disjoint-monomial}, the monomial $z^{\chi(\hat B) + \chi(B \setminus \hat B)}$ has Hamming weight exactly $\card{\hat B} + \card{B \setminus \hat B} = \card{B} = i$. And applying the operators $\ca{H}_i(\cdot)$ and $\ca{R}(\cdot)$ does not change the Hamming weight of this monomial. 
Finally, notice again that as the sets $\hat B$ and $B \setminus \hat B$ are disjoint, $\chi(\hat B) + \chi(B \setminus \hat B)$ (interpreted as a binary sum) is precisely $\chi(\hat B \cup (B \setminus \hat B)) = \chi(B)$. 
We can thus conclude that $P_{i, j', q}(z)$ contains the monomial $z^{\chi(B)}$. 

Conversely, assume that the polynomial $P_{i, j', q}(z)$ contains the monomial $z^{\chi(B)}$. 
Then, by Equation~(\ref{eq:poly}), there exist $i' \in [i]_0$ and $q' \in [q]_0$ such that the product $A_{i', j', q'}(z) \cdot P_{i - i', j' - 1, q - q'}(z)$ contains the monomial $z^{\chi(B)}$. For convenience, let $i'' = i - i'$ and $q'' = q - q'$. 
Let $z^a$ be a  monomial of $A_{i', j', q'}(z)$ and $z^p$ be a monomial of $P_{i'', j - 1, q''}$ such that $z^a \cdot z^p = z^{a + p} = z^{\chi(B)}$.  
Notice that every monomial in $A_{i', j', q'}(z)$ has Hamming weight exactly $i'$ and every monomial in $P_{i', j' - 1, q''}(z)$ has Hamming weight exactly $i''$. 
We can therefore conclude that the monomials $z^a$ and $z^p$ have Hamming weights $i'$ and $i''$, respectively. 

We will now show that the monomials $z^a$ and $z^p$ respectively correspond to assignments of  subsets $B' \subseteq B$ and $B'' \subseteq [B]$ with $B' \cup B'' = B$ and $B' \cap B'' = \emptyset$. First, by the definition of $A_{i', j', q'}(z)$, each monomial of $A_{i', j', q'}(z)$ corresponds to an assignment of a set of exactly $i'$ balls to bin $j'$ with cost at most $q'$. 
Let $B' \subseteq [s]$ with $\card{B'} = i'$ be the set of exactly $i'$ balls such that the monomial $z^a$ of $A_{i', j', q'}(z)$ corresponds to the assignment of $B'$ to bin $j'$ (with cost at most $q'$); let this assignment be denoted by $\beta'$. 
In particular, we have $z^a = z^{\chi(B')}$. Second, let $B''  \subseteq [s]$ be such that $\chi(B'')$ (interpreted as the binary encoding of an integer) is precisely $p$. That is, $z^p = z^{\chi(B'')}$. Notice that such a set $B''$ exists as the monomial $z^p$ (and in particular, the binary encoding of $p$) has Hamming weight $i'' \leq s$. Because of the same reason, we can conclude that $\card{B''} = i''$. We thus have $B', B'' \subseteq [s]$ with $\card{B'} = i'$, $\card{B''} = i'' = i - i'$ and  $z^{\chi(B')} \cdot z^{\chi(B'')} = z^{\chi(B') + \chi(B'')} = z^{\chi(B)}$. As  the monomial $z^{\chi(B)} = z^{\chi(B') + \chi(B'')}$ has Hamming weight $\card{B} = i = i' + i'' = \card{B'} + \card{B''}$, by Proposition~\ref{prop:disjoint-monomial}, we can conclude that the sets $B'$ and $B''$ are disjoint.  Observe then  that we must have $B', B'' \subseteq B$, for otherwise we would not have $\chi(B') + \chi(B'') = \chi(B)$. As $\card{B''} = i'' = i - i' = \card{B} - \card{B'}$, we can conclude further that $B'' = B \setminus B'$. 
We have thus argued that the polynomial $P_{i'', j' - 1, q''}$ contains the monomial $z^{p} = z^{\chi(B'')}$. By the induction hypothesis, there exists an assignment $\fn{\beta''}{B''}{[j' - 1]}$ of $B''$ with $\cost_{\beta''}(B'') \leq q''$. 

Let $\fn{\gamma}{B}[j']$ be the assignment of $B$ obtained by combining $\beta'$ and $\beta''$ as follows: We have $\gamma(\Tilde{i}) = \beta'(\Tilde{i}) = j'$ if $\Tilde{i} \in B'$ and $\gamma(\Tilde{i}) = \beta''(\Tilde{i})$ if $\Tilde{i} \in B''$. As $B'$ and $B''$ are disjoint and $B' \cup B'' = B$, the assignment $\gamma$ is well-defined. Notice that $\cost_{\gamma}(B) = \cost_{\beta'}(B') + \cost_{\beta''}(B'') \leq q' + q'' = q' + (q - q'') = q$. Thus, $\gamma$ is the required assignment of $B$ to the first $j'$ bins with cost at most $q$.  
\end{claimproof}
This completes the proof of Theorem~\ref{thm:binb}. 
\end{proof}

\subsubsection{An Algorithm for an Annotated Variant of Cluster Modification}
To show that that \bcc\ and \bce\ admit algorithms that runs in time \fastbccruntime, we do as follows. We first define an annotated variant of {\sc Cluster Completion} called \annocmfull\ (\annocm), and show that solving \bcc\ and \bce\ amounts to solving $2^{o(k)}$ many instances of \annocm, each instance of which can be reduced to an instance of \binb\ with $s \leq 2k$ balls, $t \leq 2k$ bins and budget $W = 2k$. Theorem~\ref{thm:binb} will then yield \fastbccruntime\ time algorithms for \bcc\ and \bce. 

\subparagraph*{Notation and Terminology.} 
For multisets $A$ and $B$, by the set difference $A \setminus B$, we mean the ``additive set difference'' of $A$ and $B$; that is, $A \setminus B \subseteq A$ and for each $x \in A$, $\mul(x, A \setminus B) = \max \set{0, \mul(x, A) - \mul(x, B)}$. For example, if $A = \set{x, y, y, y, z, z, z}$ and $B = \set{x, x, y, w}$, then $A \setminus B = \set{y, y, z, z, z}$. For a multiset $X$ of non-negative integers and $\eta \geq 0$, we say that $X$ is $\eta$-balanced if $\card{x - y} \leq \eta$ for every $x, y \in X$. %
For a graph $G$, by the set of component sizes of $G$, denoted by $\CS(G)$, we mean the multiset of positive integers where for each $j \in \mathbb{N}$, the multiplicity of $j$ in $\CS(G)$ is precisely the number of components of $G$ of size exactly $j$; that is, $\CS(G)$ is the multiset  $\set{\card{H} ~|~ H \text{ is a connected component of } G}$. In particular, a graph $G$ is $\eta$-balanced if and only if $\CS(G)$ is $\eta$-balanced. Consider a cluster graph $G$ and an integer $\ell \in \mathbb{N}$. Recall that for a partition $X' = \set{x'_1, x'_2,\ldots, x'_s}$ of $\ell$, we say that $X'$ is $G$-valid if $G$ contains $s$ distinct components $H_1, H_2,\ldots, H_s$ such that $\card{H_i} = x'_i$ for every $i \in [s]$. For partitions $X$ and $X'$ of $\ell$ such that $X'$ is $G$-valid, by a modification of $G$ with respect to $(X, X')$, we mean a cluster graph $\hat G$ obtained from $G$ by adding edges so  that $\CS(\hat G) = (\CS(G) \setminus X') \cup X$; that is, we obtain $\hat G$ from $G$ by merging components so that $\card{X'}$ components of $G$ of sizes from $X'$ are ``replaced'' by $\card{X}$ components of sizes from $X$.  Notice that a modification of $G$ w.r.t. $(X, X')$ need not exist, and that a modification need not be unique if one exists. Consider a modification $\hat G$ of $G$ w.r.t. $(X, X')$ (assuming one exists). Notice that we can identify   $\hat G$ with a function that maps $X'$ to $X$: Each $x' \in X'$ is mapped to $x \in X$ if and only if the component of $G$ corresponding to  $x'$ is merged with other components to form a component of $\hat G$ of size exactly $x$, and we call this function a witness for the modification $\hat G$ and denote it by $\wit_{\hat G}$, or simply $\wit$ when $\hat G$ is clear from the context. That is, the witness for $\hat G$ is the function $\fn{\wit}{X'}{X}$, where for each $x \in X$, a component of $\hat G$ of size $x$ replaces exactly $\card{\wit^{-1}(x)}$ components of $G$ of sizes from the multiset $\wit^{-1}(x)$; we thus have $\sum_{x' \in \wit^{-1}(x)} x' = x$ for every $x \in X$. \lipicsEnd 

\begin{observation}
\label{obs:witness}Consider a cluster graph $G$, partitions $X = \set{x_1, x_2,\ldots, x_t}$ and $X' = \set{x'_1, x'_2,\ldots, x'_s}$ of $\ell \in \mathbb{N}$ such that $X'$ is $G$-valid, and a modification $\hat G$ of $G$ w.r.t. $(X, X')$ (assuming a modification exists). 
Corresponding to each $x'_i \in X'$ and $x_j \in X$ such that $\wit(x'_i) = x_j$, a component $H$ of $G$ with $\card{H} = x'_i$ is merged with some other components of $G$ to form a component $\hat H$ of $\hat G$ with $\card{\hat H} = x_j$. To form $\hat H$, notice that we need to add exactly $\card{\hat H} - \card{H} = x_j - x'_i$ edges incident with each vertex of $H$; in total, we need to add exactly $\card{H}(\card{\hat H} - \card{H}) = x'_i(x_j - x'_i)$ edges incident with $V(H)$. Thus, $\card{E(\hat G) \setminus E(G)}$, i.e., the total number of edges we need to add to $G$ to form $\hat G$ is precisely $ (1/2) \sum_{i \in [s]} \Lb{x'_i \cdot (\wit(x'_i) - x'_i)}$; the $1/2$ accounts for the fact that each edge in $E(\hat G) \setminus E(G)$ is counted twice in the sum $\sum_{i \in [s]} \Lb{x'_i \cdot (\wit(x'_i) - x'_i)}$. 
\end{observation}

We now define the following problem. 

\defproblem{\annocmfull\ (\annocm)}{A cluster graph, a non-negative integer $k$ and two partitions $X$ and $X'$ of an integer $\ell \in \mathbb{N}$ such that $X'$ is $G$-valid.}{Decide if there exists a modification $\hat G$ of $G$ with respect to $(X, X')$ such that $\card{E(\hat G) \setminus E(G)} \leq k$.}

\begin{theorem}
\label{thm:annocm}
\annocm\ admits an algorithm that runs in time $2^\ell \cdot n^{\cO(1)}$. 
\end{theorem}
\begin{proof}
We will show that we can reduce \annocm\ to \binb. The theorem will then follow from  Theorem~\ref{thm:binb}. 

Given an instance $(G, k, X, X')$ of \annocm, where $X' = \set{x'_1, x'_2,\ldots, x'_s}$ and $X = \set{x_1, x_2,\ldots, x_t}$ are partitions of $\ell \in \mathbb{N}$, we construct an instance of \binb\ as follows: We take  $s$ balls indexed by $1, 2,\ldots, s$ with sizes $x'_1, x'_2,\ldots, x'_s$, and $t$ bins indexed by $1, 2,\ldots, t$ with capacities $x_1, x_2,\ldots, x_t$, a budget $W = 2k$ and a cost function $\fn{\cost}{[s]}{[t]}$ where $\cost(i, j) = x'_i(x_j - x'_i)$ if $x_j \geq x'_i$ and $\cost(i, j) = 2k + 1$ otherwise. 

Notice that we can construct this instance of \binb\ in polynomial time. Also, as $X'$ and $X$ are partitions of $\ell$, we have $s = \card{X'} \leq \ell$ and $t = \card{X} \leq \ell$; as $X'$ is $G$-valid, we also have $\ell \leq n$; finally, we can assume without loss of generality that $k \leq \binom{n}{2}$.  As we have $s, t \leq \ell \leq n$ and $W = 2k \leq 2 \binom{n}{2}$, Theorem~\ref{thm:binb} implies that we can  solve this  instance of \binb\  in time $2^{\ell} n^{\cO(1)}$. We  return that the given instance of \annocm\ is a yes-instance if and only if the constructed instance of \binb\ is a yes-instance. 

To complete the proof, we only need to show that the two instances are equivalent. This follows from the fact that every modification $\hat G$ of $G$ w.r.t. $(X, X')$ with $\card{E(\hat G) \setminus E(G)} \leq k$ corresponds to an assignment $\fn{\beta}{[s]}{[t]}$ with cost at most $W = 2k$ and vice versa: For every $i \in [s]$ and $j \in [t]$, we will have $\beta(i) = j$ if and only if $\wit_{\hat G}(x'_i) = x_j$. As it is straightforward to verify that this correspondence between $\hat G$ and $\beta$ indeed shows the equivalence between the two instances, we omit the technical details. 
\end{proof}

\subsubsection{A Faster Algorithm for BCC}

We now prove the following lemma, which shows that solving \bcc\ amounts to solving $2^{o(k)}$ many instances of \annocm. 

\begin{lemma}
\label{lem:completion-fast-fpt-correctness}
Consider an instance $(G, k, \eta)$ of \bcc. 
If $(G, k, \eta)$ is a yes-instance, then either $G$ is $\eta$-balanced or there exist an integer $\ell \in [2k] \setminus \set{1}$ and partitions $X$ and $X'$ of $\ell$ such that $X'$ is $G$-valid, the multiset $(\CS(G) \setminus X') \cup X$ is $\eta$-balanced, and $(G, k, X, X')$ is a yes-instance of \annocm. 
\end{lemma}

\begin{proof}
Consider an instance $(G, k, \eta)$ of \bcc. 
Assume that $(G, k, \eta)$ is a yes-instance. 
If $G$ is $\eta$-balanced, then the lemma trivially holds; recall that  $G$ is a cluster graph.  
So, assume that $G$ is not $\eta$-balanced. We will show that $\ell$, $X$ and $X'$, as required by the statement of lemma, exist by invoking Lemma~\ref{lem:completion-fpt-correctness}. 

By Lemma~\ref{lem:completion-fpt-correctness}, there exist $\ell \in [2k] \setminus \set{1}$ and a partition $X = \set{x_1, x_2,\ldots, x_{t}}$ of $\ell$ such that for each $i \in [t]$, there exists a partition $X_i$ of $x_i$ with the following properties: (1) $\bigcup_{i \in [t]}X_i$ is $G$-valid, (2) $\sum_{i \in [t]}\spp(X_i) \leq k$, and (3) a completion  of $G$ w.r.t. $(X, \set{X_1, X_2,\ldots, X_t})$ is $\eta$-balanced. 

Let $\hat G$ be a completion of $G$ w.r.t. $(X, \set{X_1, X_2,\ldots, X_t})$. Then, by Observation~\ref{obs:completion-fpt-correctness}  and property (3) above, $\hat G$ is $\eta$-balanced; by Observation~\ref{obs:completion-fpt-correctness}  and property (2) above, $\card{E(\hat G) \setminus E(G)} = \sum_{i \in [t]}\spp(X_i) \leq k$. 
We take $X' = \bigcup_{i \in [t]}X_i$. Then, by the definition of a modification of $G$ w.r.t. $(X, X')$, $\hat G$ is a modification of $G$ w.r.t. $(X, X')$. And as $\card{E(\hat G) \leq E(G)} \leq k$, we can conclude that $(G, k, X, X')$ is a yes-instance of \annocm. Again, as $\hat G$ is a modification of $G$ w.r.t. $(X, X')$, we also have $(\CS(G) \setminus X') \cup X = \CS(\hat G)$, which is $\eta$-balanced as $\hat G$ is an $\eta$-balanced graph. Thus, $\ell, X$ and $X'$ satisfy all the properties required by the statement of the lemma.  
\end{proof}

We are now ready to describe our algorithm for \bcc, which we call \fastalgobcc. 

\subparagraph*{\fastalgobcc.} Given an instance $(G, k, \eta)$ of \bcc\ as input, we proceed as follows. 
Recall that $G$ is a cluster graph. 
\begin{description}
\item[Step 1.] If $G$ is $\eta$-balanced, then we return that $(G, k, \eta)$ is a yes-instance, and terminate. 

\item[Step 2.] If $k \leq 0$, then we return that $(G, k, \eta)$ is a no-instance, and terminate. 

\item[Step 3.] We use the algorithm of Proposition~\ref{prop:partition-number} to generate all partitions of $\ell$  for all $\ell \in [2k]$.

\item[Step 4.] For each $\ell \in [2k] \setminus \set{1}$, we do as follows.

\item[Step 4.1.] For each pair of partitions $X$ and $X'$ of $\ell$ such that $(\CS(G) \setminus X') \cup X$ is $\eta$-balanced, we do as follows. 

\item[Step 4.1.1.] We consider the instance $(G, k, X, X')$ of \annocm, and use the algorithm of Theorem~\ref{thm:annocm} to solve this instance. If $(G, k, X, X')$ is a yes-instance of \annocm, then we return that $(G, k, \eta)$ is a yes-instance of \bcc, and terminate. 

\item[Step 5.] We return that $(G, k, \eta)$ is a no-instance, and terminate. 
\end{description}

The correctness of \fastalgobcc\  follows from Lemma~\ref{lem:completion-fast-fpt-correctness}. To analyse its running time, observe that the most time-consuming step is Step 4.1.1., each execution of which requires time $2^{\ell} n^{\cO(1)}$ (by Theorem~\ref{thm:annocm}). Also, by Proposition~\ref{prop:partition-number}, the number of partitions of an integer $\ell$ is $2^{\cO(\sqrt{\ell})}$, and hence  the number of pairs $(X, X')$ that we consider in Step 4.1 is  $2^{\cO(\sqrt{\ell})} \cdot 2^{\cO(\sqrt{\ell})} = 2^{o(\ell)}$. Finally, we go over $2k - 1$ choices of $\ell$ in Step 4, and as we only consider $\ell \leq 2k$, we can conclude that the total running time is bounded by $2k \cdot 2^{o(2k)} \cdot 2^{2k} n^{\cO(1)} = \fastbccruntime$. 
We thus have the following result. 
\begin{theorem}
\label{thm:completion-fast-fpt}
\bcc\ admits an algorithm that runs in time \fastbccruntime. 
\end{theorem}

\subsubsection{A Faster Algorithm for BCE}

We now turn to \bce. To design an algorithm for \bce, notice that as we did earlier, we only need to design an algorithm for \bce\ when the input graph is a cluster graph, as we can branch on induced $P_3$s  otherwise. Specifically, we prove the following theorem. 
\begin{theorem}
\label{thm:editing-fast-fpt-cluster}
\bce\ on cluster graphs admits an algorithm that runs in time \fastbccruntime. 
\end{theorem}

Assuming Theorem~\ref{thm:editing-fast-fpt-cluster}, let us first prove the following result. 

\begin{theorem}
\label{thm:editing-fast-fpt}
\bce\ admits an algorithm that runs in time \fastbccruntime. 
\end{theorem}
\begin{proof}[Proof Sketch]
The proof is nearly identical to the proof of Theorem~\ref{thm:editing-fpt}. 

Given an instance $(G, k, \eta)$ of \bce, our algorithm works as follows.  If $G$ is a cluster graph, then we use the algorithm of Theorem~\ref{thm:editing-fast-fpt-cluster} to solve the problem. Otherwise, we find an induced $P_3$, say $uvw$ and recursively call our algorithm on the three instances $(G - uv, k - 1, \eta), (G - vw, k - 1, \eta)$ and $(G + uw, k - 1, \eta)$. 

As for the running time, notice that each time we find a $P_3$, we make $3$ recursive calls, until $k$ becomes $0$ or the graph becomes a cluster graph $G$ (whichever happens earlier). Hence the total number of recursive calls is at most $\sum_{i = 0}^k 3^i$. Now, on any branch of the computation, as we first make $i \in [k]_0$ edge modifications until the graph becomes a cluster graph and only then invoke the algorithm of Theorem~\ref{thm:editing-fast-fpt-cluster} with parameter $k - i$, the latter step will only take time $4^{k - i + o(k - i)} n^{\cO(1)}$. Thus the total running time is bounded by $\sum_{i = 0}^k 3^i \cdot 4^{k - i + o(k - i)} n^{\cO(1)} \leq \fastbccruntime$.  
\end{proof}

Now we only have to prove Theorem~\ref{thm:editing-fast-fpt-cluster}. Before moving to a formal proof of Theorem~\ref{thm:editing-fast-fpt-cluster}, we  highlight some of the key ideas first.  

\subparagraph*{Ideas we use in our algorithm for \bce\ on cluster graphs.} For the most part, we use a combination of arguments we already used in our previous algorithm for \bce\ on cluster graphs (\algobcec) and our  \fastbccruntime\ time algorithm for \bcc\ (\fastalgobcc). But there are subtle differences. First of all, observe that a cluster graph $G$ is completely specified by its component sizes, i.e., the multiset $\CS(G)$;  any two cluster graphs $H$ and $H'$ are isomorphic if and only if $\CS(H) = \CS(H')$. We will rely on this fact. 
Now, suppose $(G, k, \eta)$ is a yes-instance of \bce, where $G$ is a cluster graph, and suppose $F \subseteq \binom{V(G)}{2}$ is the solution that we are looking for. Let $F_1 \subseteq F$ be the set of edges that we delete from $G$ and $F_2 = F \setminus F_1$ be the set of edges that we add to $G$. Then $G - F_1$ is a cluster graph. Also, $F_1$ corresponds to a pair of partitions $(Y, Y')$ of $\ell_1 = \card{V(F_1)}$:  each $y \in Y$ corresponds to a component of $G$ of size $y$ that was split into smaller components when we deleted $F_1$, and each $y' \in Y'$ corresponds to one of the ``new'' components of $G - F_1$ of size $y'$ that was formed when we deleted $F_1$. Thus $\CS(G - F_1) = (\CS(G) \setminus Y) \cup Y'$. Once again, 
we obtain $G'$ from $G$ by deleting the edges of $F_1$, or equivalently, we obtain $G$ from $G - F_1$ by adding the edges of $F_1$. In other words, $(G - F_1, \card{F_1}, Y, Y')$ is a yes-instance of \annocm.  But then, so is $(G', \card{F_1}, Y, Y')$ for \emph{any cluster}  graph $G'$ with $\CS(G') = \CS(G - F_1)$. %
The rest of the arguments are identical to those in \fastalgobcc. We obtain $G \triangle F$ by adding the edges of $F_2$ to the cluster graph $G - F_1$. That is, for appropriate partitions $X$ and $X'$ of $\card{V(F_2)}$, $(G - F_1, \card{F_2}, X, X')$ is a yes-instance of \annocm, and therefore, $(G', \card{F_2}, X, X')$ is a yes-instance of 
\annocm. 
So in our algorithm,  we guess $\card{F_1}, \card{F_2}, \card{V(F_1)}, \card{V(F_2)}, Y, Y'$ and $X, X'$ and we simply check if $(G', \card{F_1}, Y, Y')$ is a yes-instance of \annocm, and if $(G', \card{F_2}, X, X')$ is a yes-instance of \annocm, where $G'$ is an arbitrary cluster graph with $\CS(G') = \CS(G) \setminus Y) \cup Y'$. \lipicsEnd

We formalise the above ideas in the following lemma, which will establish the correctness our algorithm for \bce\ on cluster graphs. 

\begin{lemma}
\label{lem:editing-fast-fpt-correctness}
Consider an instance $(G, k, \eta)$ of \bce, where $G$ is a cluster graph. 
If $(G, k, \eta)$ is a yes-instance, then one of the following statements holds. 
\begin{enumerate}
    \item The instance $(G, k, \eta)$ is a yes-instance of \bcc. 
    \item The instance $(G, k, \eta)$ is a yes-instance of \bcd. 
    \item There exist integers $k_1, k_2, \ell_1, \ell_2$ with $k_1 + k_2 \leq k$ and $\ell_1 \in [2k_1] \setminus \set{1}, \ell_2 \in [2k_2] \setminus \set{1}$  such that there exist partitions $Y$ and $Y'$ of $
    \ell_1$ and $X$ and $X'$ of $\ell_2$ with the following properties: %
\begin{enumerate}
\item $Y$ is $G$-valid; 

\item $(G', k_1, Y, Y')$ is a yes-instance of \annocm, where $G'$ is a cluster graph with $\CS(G') = (\CS(G) \setminus Y) \cup Y'$;

\item $X'$ is $G'$-valid; 

\item the multiset $(\CS(G') \setminus X') \cup X$ is $\eta$-balanced; 

\item $(G', k_2, X, X')$ is a yes-instance of \annocm. 
\end{enumerate}
\end{enumerate}
\end{lemma}

\begin{proof}[Proof Sketch]
The proof closely follows the proof of Lemma~\ref{lem:editing-fpt-correctness}. We only highlight the most salient arguments. Assume that $(G, k, \eta)$ is a yes-instance of \bce, and let $F$ be a solution for $(G, k, \eta)$. Let $F_1 \subseteq F$ be the set of edges that we delete from $G$ and $F_2 = F \setminus F_1$ be the set of edges that we add to $G$. Then $G' = G - F_1$ is a cluster graph. 

We take $k_i = \card{F_i}$ and $\ell_i = \card{V(F_i)}$ for each $i \in [2]$. We associate $F_1$ with a pair of partitions $Y$ and $Y'$ of $\ell_1$; the partition $Y$ corresponds to the components of $G$ that are  split into smaller components when we delete  $F_1$, and the partition $Y'$ corresponds to the ``new'' components of $G' = G - F_1$ that are formed when we delete $F_1$. In other words, we obtain $G$ from $G'$ by adding edges, and in particular, $G$ is a modification of $G'$ w.r.t. $(Y, Y')$. As $E(G) \setminus E(G') = F_1$ and $\card{F_1} = k_1$, $(G', k_1, Y, Y')$ is a yes-instance of \annocm. We  associate $F_2$ with a pair of partitions $X$ and $X'$ of $\ell_2$. The partition $X'$ corresponds to the components of $G'$ that are merged together to form the components of $G' + F_2 = G \triangle F$; the partition $X$ corresponds to the ``new'' components of $G' + F_2$ that are formed by merging together components of $G'$. In other words, $G' + F_2 = G \triangle F$ is a modification of $G$ w.r.t. $(X, X')$, and thus $\CS(G \triangle F) = (\CS(G' + F_2) \setminus X') \cup X$. As $E(G \triangle F) \setminus E(G') = F_2$ and $\card{F_2} = k_2$, $(G', k_2, X, X')$ is a yes-instance of \annocm, and as the graph $G \triangle F$ is $\eta$-balanced, the multiset $\CS(G \triangle F) = (\CS(G' + F_2) \setminus X') \cup X$ is $\eta$-balanced. 
\end{proof}

Based on Lemma~\ref{lem:editing-fast-fpt-correctness}, we design an algorithm for \bce\ on cluster graphs, which we call \fastalgobcec. 
\subparagraph*{\fastalgobcec.} Given an instance $(G, k, \eta)$ of \bce\ as input, where $G$ is a cluster graph, we proceed as follows.  

\begin{description}
\item[Step 1.] If $(G, k, \eta)$ is a yes-instance of \bcc\ or a yes-instance of \bcd, we return that $(G, k, \eta)$ is a yes-instance of \bce, and terminate. To do this, we use the algorithms of Theorems~\ref{thm:completion-fast-fpt} and \ref{thm:deletion-fpt}. 

\item[Step 2.] If $k \leq 0$, then we return that $(G, k, \eta)$ is a no-instance, and terminate. 

\item[Step 3.] We use the algorithm of Proposition~\ref{prop:partition-number} to generate all partitions of $\ell$ for all $\ell \in [2k]$.

\item[Step 4.] For every choice of $k_1, k_2 \in [k]$ with $k_1 + k_2 \leq k$ and $\ell_1 \in [2k_1] \setminus \set{1}, \ell_2 \in [2k_2] \setminus \set{1}$, we do as follows.

\item[Step 4.1.] For each pair of partitions $Y$ and $Y'$ of $\ell_1$ and each pair of partitions $X$ and $X'$ of $\ell_2$, such that $Y$ is $G$-valid, we do as follows. 

\item[Step 4.1.1.] We construct a cluster graph $G'$ with $\CS(G') = (\CS(G) \setminus Y) \cup Y'$. 

\item[Step 4.1.2] If $X'$ is $G'$-valid, the multiset $(\CS(G') \setminus X') \cup X$ is $\eta$-balanced, and both $(G', k_1, Y, Y')$ and $(G', k_2, X, X')$ are yes-instances of \annocm, then we return that $(G, k, \eta)$ is a yes-instance of \bce, and terminate. To check if $(G', k_1, Y, Y')$ and $(G', k_2, X, X')$ are yes-instances of \annocm, we use the algorithm of Theorem~\ref{thm:annocm}. 

\item[Step 5.] We return that $(G, k, \eta)$ is a no-instance, and terminate. 
\end{description}

The correctness of the algorithm follows from  Lemma~\ref{lem:editing-fast-fpt-correctness}. We now analyse the running time. 
\begin{lemma}
\label{lem:editing-fast-fpt-runtime}
\fastalgobcec\ runs in time \fastbccruntime. 
\end{lemma}
\begin{proof}
Observe that Steps 2 and 5 take only polynomial time. 
By Theorems~\ref{thm:completion-fast-fpt} and \ref{thm:deletion-fpt}, Step 1 takes time at most \fastbccruntime. By Proposition~\ref{prop:partition-number}, Step 3 takes time $2^{\cO(\sqrt{k})} = 2^{o(k)}$.  
Now, in Step 4, we need to go over at most $k + 1$ choices of $k_i$ and at most $2k_i - 1$ choices of $\ell_i$ for each $i \in [2]$; and we can assume without loss of generality that $k \leq \binom{n}{2}$. 
That is, there are only $n^{\cO(1)}$ choices for the tuple $(k_1, k_2, \ell_1, \ell_2)$. In Step 4.1, we go over all possible choices of $(Y, Y', X, X')$. As $Y$ and $Y'$ are partitions of $\ell_1 \leq 2k_1 \leq 2k$, by Proposition~\ref{prop:partition-number}, each of them  has at most $p(2k) = 2^{\cO(\sqrt{k})}$ choices;  similarly, each of $X$ and $X'$ also has at most $p(2k) = 2^{\cO(\sqrt{k})}$ choices. Thus, there are  $2^{\cO(\sqrt{k})} \cdot 2^{\cO(\sqrt{k})} \cdot 2^{\cO(\sqrt{k})} \cdot 2^{\cO(\sqrt{k})} = 2^{\cO(\sqrt{k})} = 2^{o(k)}$ choices for the tuple $(Y, Y', X, X')$.  
 Finally, in Step 4.1.2, for each choice of $(Y, Y', X, X')$, we invoke the algorithm of Theorem~\ref{thm:annocm} twice; on the instance $(G', k_1, Y, Y')$, the algorithm takes time $2^{\ell_1} n^{\cO(1)} \leq 2^{2k_1} n^{\cO(1)} = 4^{k_1} n^{\cO(1)}$; similarly, on the instance $(G', k_1, Y, Y')$, it takes time at most $4^{k_2} n^{\cO(1)}$. Thus each execution of Step 4.1.2 takes time $4^{k_1} n^{\cO(1)} + 4^{k_2} n^{\cO(1)} \leq 4^{k_1 + k_2} n^{\cO(1)} = 4^{k} n^{\cO(1)}$. Therefore, the overall running time of the algorithm is bounded by \fastbccruntime.
\end{proof}
This completes the proof of Theorem~\ref{thm:editing-fast-fpt-cluster}.

%% file: counting.tex
\section{FPT Algorithms for the Counting Versions of BCC, BCD and BCE}\label{sec:counting}
In this section, we show that the \emph{counting} versions of \bcc, \bcd, and \bce\ admit \FPT\ algorithms. %
We show that same ideas that we used in Section~\ref{sec:fpt-completion-editing} can be adapted to work for the counting versions of the problems as well. In particular, the idea of associating solutions with sets of integer partitions adapt nicely to the counting setting---we can enumerate these sets of partitions and count the number of solutions that correspond to each set of partitions. We formally define the counting version of \bcc\ as follows. 

\defproblem{\countbccfull\ (\countbcc)}{A cluster graph $G$ and non-negative integers $k$ and $\eta$.}{Output the number of sets $F \subseteq \binom{V(G)}{2} \setminus E(G)$ such that $\card{F} \leq k$ and $G + F$ is an $\eta$-balanced cluster graph.}

We define the counting versions of \bcd\ and \bce\ analogously. We show that \countbcc\ can be solved in time \algobccruntime, and \countbcd\ and \countbce\ can be solved in time \countbceruntime. 

\subparagraph*{Notation.} Recall the following notation. For a multiset $X$ and $x \in X$,  $\mul(x, X)$ denotes the multiplicity of $x$ in $X$. And $\US(X)$ denotes the underlying set of $X$. We say that a multiset $X$ of non-negative integers is $\eta$-balanced if $\card{x - y} \leq \eta$ for every $x, y \in X$. 
For a graph $G$, $\CC(G)$ denotes the set of connected components of $G$, and  $\CS(G)$ denotes the multiset of the component sizes of $G$. In this section, whenever the graph $G$ is clear from the context, we use $\beta(j)$ to denote the number of connected components of $G$ of size exactly $j$, i.e., $\beta(j) = \mul(j, \CS(G))$. For a decision problem $\Pi$ and its natural counting version $\# \Pi$, we use $\nsol_{\Pi}(I)$ to denote the number of solutions for $\Pi$ on the input instance $I$; for example, $\nsol_{BCC}(G, k, \eta)$. And whenever the problem is clear from the context, we may drop the subscript and simply write $\nsol(I)$. 

\subsection{FPT Algorithm for (a Colour-Restricted Variant of) \#BCC}

We begin with \countbcc. But instead of solving \countbcc\ directly, we solve a more general version of the problem in which the each component of the input graph has one of $q + 1$ colours $0, 1,\ldots, q$, and certain colour-based constraints  disallow adding edges between certain components. This colour-restricted version will be useful in designing our algorithm for \countbce\ as well. 

We now formally introduce the colour-restricted variant of \bcc. 
Consider a graph $G$. Recall that $\CC(G)$ denotes the set of connected components of $G$. 
 Consider a colouring function $\fn{\mu}{\CC(G)}{[q]_0}$ that assigns colours to the connected components of $G$. Notice that $\mu$ assigns colours to the components of $G$ and not to individual vertices. But with a slight abuse of notation, for a vertex $v$, we define $\mu(v) = \mu(H)$, where $H$ is the connected component of $G$ that contains $v$. 
 Consider $G$ and $\fn{\mu}{\CC(G)}{[q]_{0}}$. For $F \subseteq \binom{V(G)}{2}$, we say that $F$ is \emph{colour-restricted} if for every $uv \in F$, either $\mu(u) = \mu(v) = 0$ or $\mu(u) \neq \mu(v)$. That is, $F$ is colour-restricted if and only if for every $uv \in F$, either both $u$ and $v$ have colour $0$, or $u$ and $v$ have different colours. 

 \defproblem{\crbccfull\ (\crbcc)}{A cluster graph $G$ along with a colouring function $\fn{\mu}{\CC(G)}{[q]_0}$, and a non-negative integer $k$ }{Decide if there exists $F \subseteq \binom{V(G)}{2} \setminus E(G)$ such that $F$ is colour-restricted, $\card{F} \leq k$ and $G + F$ is an $\eta$-balanced cluster graph.}

Observe that \crbcc\ is identical to \bcc, except that the solution has to be colour-restricted. In particular, \bcc\ is a special case of \crbcc\ in which all components are of colour $0$. Recall that to solve \bcc\ (and its counting version \countbcc), we relied on the fact that each solution $F$ for \bcc\ could naturally be associated with a multiset of integer partitions in which each partition  describes the sizes of the components of $G$ that are merged together to form a component of $G + F$. But this is not quite the case with \crbcc, as the colours of the components have to be taken into account. Nonetheless, notice that we can still associate each solution $F$ for \crbcc\ with multisets of (size, colour) pairs that describe the size and colour of individual components of $G$ that are merged together to form a component of $G + F$. That is, each solution $F$ for \crbcc\ can be associated with a pair $(X, \set{X_1, X_2,\ldots, X_t})$, where $X = \set{x_1, x_2,\ldots, x_t}$ is a partition of $\card{V(F)}$ and for every $i \in [t]$, $X_i = \set{(x_{i1}, \tau_{i1}), (x_{i2}, \tau_{i2}),\ldots, (x_{i r_i}, \tau_{i r_i})}$ is such that $\set{x_{11}, x_{12}, \ldots, x_{i r_i}}$ is a partition of $x_i$, and $\tau_{ij} \in [q]_0$; for each $i \in [t]$, $x_{i1}, x_{i2},\ldots, x_{i r_i}$ are the sizes of the $r_i$ components of $G$ that are merged together by $F$ to form a component of $G + F$ of size $x_i$, and the corresponding $\tau_{ij}$s are the colours of those $r_i$ components. We extend the earlier definitions of various  terms involving integer partitions---for example, definitions of a partition $X$ being $G$-valid, a completion of a cluster graph w.r.t. a pair $(X, \set{X_1, X_2,\ldots, X_t})$ etc.---to include the colours. Similarly, by suitably adapting the proof of Lemma~\ref{lem:completion-bound}, we derive the following result that bounds the number of choices for the pair $(X, \set{X_1, X_2,\ldots, X_t})$. 

\begin{lemma}\label{lem:colours-bound}
    Consider positive integers $\ell$ and $q$. The number of choices for the pair $(X, \set{X_1, X_2,\ldots, X_t})$, where $X = \set{x_1, x_2,\ldots, x_t}$ is a partition of $\ell$, and for every $i \in [t]$, $X_i = \set{(x_{i1}, \tau_{i1}), (x_{i2}, \tau_{i2}),\ldots, (x_{i r_i}, \tau_{i r_i})}$ is such that $\set{x_{11}, x_{12}, \ldots, x_{i r_i}}$ is a partition of $x_i$, and $\tau_{ij} \in [q]_0$ for every $j \in [r_i]$, is $2^{\cO(\ell \log q)}$. 
\end{lemma}

\begin{proof}
 Recall that for every $\ell \in \mathbb{N}$, $p(\ell)$ denotes the number of partitions of $\ell$. 
 Recall also  that by Proposition~\ref{prop:partition-number}, there exists a constant $C$ such that $p(\ell) \leq 2^{C \sqrt{\ell}}$ for every $\ell \in \mathbb{N}$. Among all the partitions of $\ell$, let $X^* = \set{x^*_1, x^*_2,\ldots,x^*_{t^*}}$ be a partition that maximises the product  $p(x_1) (q + 1)^{x_1} \cdot p(x_2) (q + 1)^{x_2} \cdots p(x_{t^*}) (q + 1)^{x_{t^*}}$.  
 
 Now, for each fixed partition $x_{i1} + x_{i2} + \cdots + x_{ir_i}$ of $x_i$, notice that the number of choices for the multiset $X_i = \set{(x_{i1}, \tau_{i1}), (x_{i2}, \tau_{i2}),\ldots, (x_{i r_i}, \tau_{i r_i})}$, where $\tau_{ij} \in [q]_0$, is precisely $(q + 1)^{r_i} \leq (q + 1)^{x_i}$; this is true because each $\tau_{ij}$ has $q + 1$ choices, and as $\set{x_{i1}, x_{i2},\ldots, x_{i r_i}}$ is a partition of $x_i$, we have  $r_i \leq x_i$. Hence for each $x_i$, the number of choices for the multiset $X_i$ is at most $p(x_i) (q + 1)^{x_i}$. Therefore, for each partition $X$ of $\ell$, where $X = \set{x_1, x_2,\ldots, x_t}$, there are at most $p(x_1) (q + 1)^{x_1} \cdot p(x_2) (q + 1)^{x_2} \cdots p(x_t) (q + 1)^{x_t}$ choices for the multiset $\set{X_1, X_2,\ldots, X_t}$. Finally, by summing over all partitions of $\ell$, we get that the number of choices for $(X, \set{X_1, X_2,\ldots, X_t})$ is at most 
 \begin{align*}
 \sum_{X} \prod_{i \in [t]} p(x_i) (q + 1)^{x_i}  &\leq \sum_{X} \prod_{j \in [t^*]} p(x^*_j) (q + 1)^{x^*_j} \\
 &= p(\ell) \prod_{j \in [t^*]} p(x^*_j) (q + 1)^{x^*_j} \\
 &\leq 2^{C \sqrt{\ell}} \prod_{j \in [t^*]} 2^{C\sqrt{x^*_j}} (q + 1)^{x^*_j} \\ 
 &= 2^{C(\sqrt{\ell} + \sqrt{x^*_1} + \sqrt{x^*_2} + \cdots + \sqrt{x^*_{t^*}})} \cdot (q + 1)^{x^*_1 + x^*_2 + \cdots x^*_{t^*}} \\
 &\leq 2^{C(\ell + x^*_1 + x^*_2 + \cdots x^*_{t^*})} \cdot (q + 1)^{x^*_1 + x^*_2 + \cdots x^*_{t^*}} \\
 &=2^{\cO(\ell)} \cdot (q + 1)^{\ell}. 
 \end{align*}
\end{proof}

Consider a graph $G$ and $\fn{\mu}{\CC(G)}{[q]_0}$. For $j \in \mathbb{N}$ and $\tau \in [q]_0$, let $\beta(j, \tau)$ be the number of components of $G$ of size $j$ and colour $\tau$. Also, let $\ncom(G, \mu, (X, \set{X_1, X_2,\ldots, X_t}))$ be the number of completions of $G$ w.r.t. $(X, \set{X_1, X_2,\ldots, X_t})$, which we compute next. 

\begin{lemma}\label{lem:colours-completion}
Consider a cluster graph $G$, a colouring function $\fn{\mu}{\CC(G)}{[q]_0}$ and a positive integer $\ell$. Let $(X, \set{X_1, X_2,\ldots, X_t})$ be as follows: (a) $X = \set{x_1, x_2,\ldots, x_t}$ is a partition of $\ell$, (b) for every $i \in [t]$, $X_i = \set{(x_{i1}, \tau_{i1}), (x_{i2}, \tau_{i2}),\ldots, (x_{i r_i}, \tau_{i r_i})}$ is such that $\set{x_{11}, x_{12}, \ldots, x_{i r_i}}$ is a partition of $x_i$, and $\tau_{ij} \in [q]_0$ for every $j \in [r_i]$, and (c) $X'$ is $G$-valid, where $X' = \bigcup_{i \in [t]} X_i$.  %
For $i \in [t], j \in \mathbb{N}$ and $\tau \in [q]_0$, let $d(i, j, \tau)$ be defined as follows: $d(1, j, \tau) = 0$ and $d(i, j, \tau) = \sum_{s = 1}^{i - 1} \mul((j, \tau), X_s)$ for $i > 1$. For $\hat X \in \set{X_1, X_2,\ldots, X_t}$, let $m(\hat X)$ denote the multiplicity of $\hat X$ in $\set{X_1, X_2,\ldots, X_t}$. 
Then, 
\[
\ncom(G, \mu, (X, \set{X_1, X_2,\ldots, X_t})) = \frac{\prod_{i \in [t]} \prod_{(j, \tau) \in \US(X_i)} \binom{\beta(j, \tau) - d(i, j, \tau)}{\mul((j, \tau), X_i)}}{\prod_{\hat X \in \US(\set{X_1, X_2,\ldots, X_t})} m(\hat X)!}. 
\]
Moreover, given $G, \mu$ and $(X, \set{X_1, X_2,\ldots, X_t})$, we can compute $\ncom(G, \mu, (X, \set{X_1, X_2,\ldots, X_t}))$ in polynomial time. 
\end{lemma}
\begin{proof}
Observe that $\ncom(G, \mu, (X, \set{X_1, X_2,\ldots, X_t}))$ is precisely the number of ways in which we can choose $t$ pairwise-disjoint sets $\ca{D}_{1}, \ca{D}_2,\ldots, \ca{D}_t$ of connected components of $G$ such that for each $i \in [t]$,  $\ca{D}_i$ consists of $r_i$ distinct components $H_{i1}, H_{i2},\ldots, H_{i r_i}$ with $\card{H_{ij}} = x_{ij}$ and $\mu(H_{ij}) = \tau_{ij}$ for every $j \in [r_i]$. For $i = 1$ to $t$ in this order, let us count the number of ways in which we can choose $\ca{D}_i$. Observe that for each $(j, \tau) \in \US(X_1)$,  $\ca{D}_1$ consists of $\mul((j, \tau), X_1)$  components of size $j$ and colour $\tau$. And $G$ has $\beta(j, \tau)$ components of $G$ of size $j$ and colour $\tau$. Therefore, we can choose $\mul((j, \tau), X_1)$ many such components in $\binom{\beta(j, \tau)}{\mul((j, \tau), X_1)}$ ways, and hence the number of choices for $\ca{D}_1$ is $\prod_{(j, \tau) \in \US(X_1)} \binom{\beta(j, \tau)}{\mul((j, \tau), X_1)}$. Now, consider $i > 1$ and $(j, \tau) \in \US(X_i)$.  Out of the $\beta(j, \tau)$ components of size $j$ and colour $\tau$, for each $s < i$, we choose exactly $\mul((j, \tau), X_s)$ components; so when we come to $i$, we will have already chosen $\sum_{s = 1}^{i - 1} \mul((j, \tau), X_s) = d(i, j, \tau)$ such components, and we will be left with $\beta(j, \tau) - d(i, j, \tau)$ components. Hence we can choose $\mul((j, \tau), X_i)$ of those components in $\binom{\beta(j, \tau) - d(i, j, \tau)}{\mul((j, \tau), X_i)}$ ways, which implies that the number of choices for $\ca{D}_i$ is precisely $\prod_{(j, \tau) \in \US(X_i)} \binom{\beta(j, \tau) - d(i, j, \tau)}{\mul((j, \tau), X_i)}$. Thus the number of choices for the ordered sequence $(\ca{D}_1, \ca{D}_2, \ldots, \ca{D}_t)$ is $\prod_{i \in [t]} \prod_{(j, \tau) \in \US(X_i)} \binom{\beta(j, \tau) - d(i, j, \tau)}{\mul((j, \tau), X_i)}$. 

Now, notice that the order of the sets $\ca{D}_1, \ca{D}_2,\ldots, \ca{D}_t$ is immaterial. And for each $\hat X \in \US(\set{X_1, X_2,\ldots, X_t})$, $\set{\ca{D}_1, \ca{D}_2,\ldots, \ca{D}_t}$ contains $m_{\hat X}$ identical sets; we think of $\ca{D}_i$ and $\ca{D}_{i'}$ as identical if $X_i = X_{i'}$. Any permutation of these $m_{\hat X}$ sets would all lead to the same set $\set{\ca{D}_1, \ca{D}_2,\ldots, \ca{D}_t}$, and there are $m_{\hat X}!$ such permutations. Thus, to obtain the number of choices for $\set{\ca{D}_1, \ca{D}_2,\ldots, \ca{D}_t}$, we must divide the product $\prod_{i \in [t]} \prod_{(j, \tau) \in \US(X_i)} \binom{\beta(j, \tau) - d(i, j, \tau)}{\mul((j, \tau), X_i)}$ by $\prod_{\hat X \in \US(\set{X_1, X_2,\ldots, X_t})} m_{\hat X}!$. Finally, as all the numbers involved in both these products can be encoded using $n^{\cO(1)}$ bits, we can compute $\ncom(G, \mu, (X, \set{X_1, X_2,\ldots, X_t}))$ in polynomial time. 
\end{proof}

Consider an instance $(G, \mu, k, \eta)$ of \crbcc. To solve \crbcc\ and its counting variant \countcrbcc,  we can use the correspondence between solutions for $(G, \mu, k, \eta)$ and completions of $G$ w.r.t. pairs $(X, \set{X_1, X_2,\ldots, X_t})$ with appropriate properties.  Among other properties, we would in particular need the $X_i$s to be colour-restricted; for $X_i = \set{(x_{i1}, \tau_{i1}), (x_{i2}, \tau_{i2}),\ldots, (x_{i r_i}, \tau_{i r_i})}$, where each $\tau_{ij} \in [q]_0$, we say that $X_i$ is colour-restricted if for each $p \in [q]$, there exists at most one index $j \in [r_i]$ such that $\tau_{ij} = p$. Notice that if $X_1, X_2,\ldots, X_t$ are colour-restricted, then $F$ is colour-restricted for any $F \subseteq \binom{V(G)}{2} \setminus E(G)$ such that $G + F$ is a completion of $G$ w.r.t. $(X, \set{X_1, X_2,\ldots, X_t})$. 
Now, as in the case of \bcc, given an instance $(G, \mu, k, \eta)$ of \crbcc, by going over all possible choices for $(X, \{X_1, X_2,\ldots, X_t\})$ and summing up $\ncom(G, \mu, (X, \set{X_1, X_2,\ldots, X_t}))$ we can count the number of solutions for the instance $(G, \mu, k, \eta)$. Thus,  Lemmas~\ref{lem:colours-bound} and \ref{lem:colours-completion} together imply the following result. 

\begin{lemma}\label{lem:colours-counting}
\countcrbcc\ admits an algorithm that runs in time \colourruntime. 
\end{lemma}

As \bcc\ is a special case of \crbcc\ with just one colour (i.e., $q + 1 = 1$), Lemma~\ref{lem:colours-counting} immediately implies the following result. 

\begin{theorem}
    \countbcc\ admits an algorithm that runs in time \algobccruntime. 
\end{theorem}

\subsection{FPT Algorithms for \#BCD  and \#BCE}

We now turn to the counting versions of \bcd\ and \bce, and design algorithms for these problems that run in time \countbceruntime. These algorithms, as in the case of our algorithm for \countbcc, work by enumerating all representative solutions and counting the number of solutions that correspond to each representative. In particular, we reduce the problem to $2^{\cO(k \log k)}$ many instances of appropriately defined variants of \countbcc, \countbcd\ or \countbce, where the input graphs are cluster graphs. And then we can easily associate the representative solutions with a set of integer partitions, just as we did in our algorithms for (the decision versions of) \bcc\ and \bce. 

\subparagraph*{Notation.} %
Consider a multiset $Y = \set{y_1, y_2,\ldots, y_s}$ of positive integers. By $\tsum(Y)$, we denote the quantity $\sum_{i \in [s]} y_j$. In particular, if $Y$ is a partition of $\ell$ for some positive integer $\ell$, then $\tsum(Y) = \ell$. For a clique $G$ on $y_1$ vertices and a partition $Y_1$ of $y_1$, let $\ndel(y_1, Y_1)$ denote the number of deletions of $G$ with respect to $(\set{y_1}, Y_1)$. For a cluster graph $G$, a partition $Y = \set{y_1, y_2,\ldots, y_s}$ (of some positive integer $\ell$) such that $Y$ is $G$-valid, and a partition $Y_i$ of $y_i$ for each $i \in [s]$, let $\ndel(G, (Y, \set{Y_1, Y_2,\ldots, Y_s}))$ denote the number of deletions of $G$ with respect to $(Y, \set{Y_1, Y_2,\ldots, Y_{s}})$. \lipicsEnd

 We begin by proving the following lemma, which counts the number of deletions of a clique with respect to an integer partition. 

\begin{lemma}\label{lem:count-deletion-clique}
Let $G$ be a clique on $y_1$ vertices, and let $Y_1 = \set{y_{11}, y_{12}, \ldots, y_{1 q_1}}$ be a partition of $y_1$. For $i \in [q_1]$, let $b_i$ be defined as follows: $b_1 = 0$ and for $i > 1$, $b_i = \sum_{t = 1}^{i - 1} y_{1 t}$. For $y \in \US(Y_1)$, let $m(y) = \mul(y, Y_1)$.  Then,  
\[
\ndel(\set{y_1}, Y_1) = \frac{\prod_{i = 1}^{q_1} \binom{y_1 - b_i}{y_{1i}}}{\prod_{y \in \US(Y_1)} m(y)!}. 
\]
\end{lemma}
\begin{proof}
Observe that $\ndel(\set{y_1}, Y_1)$ is precisely the number of ways in which we can partition $V(G)$ into exactly $q_1$ parts, say $D_1, D_2,\ldots, D_{q_1}$ such that $\card{D_{i}} = y_{1i}$ for every $i \in [q_1]$. Now, the number of choices for $D_1$ is precisely the number of ways in which we can choose $y_{11}$ vertices from $y_1$ vertices, which is precisely $\binom{y_1}{y_{11}}$. Having chosen $D_1$, the number of ways in which we can choose $D_2$ is precisely the number of ways in which we can choose $y_{12}$ vertices from $y_1 - y_{11} = y_1 - b_2$ vertices, which is precisely $\binom{y_1 - b_2}{y_{12}}$. More generally, having chosen $D_1, D_2,\ldots, D_{i - 1}$, the number of choices for $D_i$ is precisely the number of ways in which we can choose $y_{1i}$ vertices from $y_1 - \sum_{t = 1}^{i - 1}{y_{1t}} = y_1 - b_i$ vertices, which is precisely $\binom{y_1 - b_i}{y_{1i}}$. Therefore, the number of choices for the tuple $(D_1, D_2,\ldots, D_{q_1})$ is $\prod_{i = 1}^{q_1} \binom{y_1 - b_i}{y_{1i}}$. Now, as $\set{D_1, D_2,\ldots, D_{q_1}}$ is a partition of $V(G)$, the order of the sets in the partition is immaterial. For each $y \in \US(Y_1)$, the partition $\set{D_1, D_2,\ldots, D_{q_1}}$ contains exactly $\mul(y, Y_1) = m(y)$ sets of size $y$; any permutation of these $m(y)$ sets would all lead to the same partition of $V(G)$, and there are $m(y) !$ such permutations. Thus, to obtain, $\ndel(\set{y_1}, Y_1)$, we must divide the product $\prod_{i = 1}^{q_1} \binom{y_1 - b_i}{y_{1i}}$ by $\prod_{y \in \US(Y_1)} m(y) !$, which proves the lemma. 
\end{proof} 

We now extend Lemma~\ref{lem:count-deletion-clique} to count the number of deletions of a cluster graph w.r.t. a multiset of integer-partitions. 
\begin{lemma}\label{Lem:count-deletion-cluster}
Consider a cluster graph $G$, a positive integer $\ell$, and a pair $(Y, \set{Y_1, Y_2,\ldots, Y_s})$, where $Y = \set{y_1, y_2,\ldots, y_s}$ is a partition of $\ell$, and $Y_i$ is a partition of $y_i$ for each $i \in [s]$, and $Y$ is $G$-valid. 
\begin{enumerate}  
\item Let $\US(\set{Y_1, Y_2,\ldots, Y_s} = \set{Y'_1, Y'_2,\ldots, Y'_r}$.  
\item For $j \in [r]$, let $m(Y'_j) = \mul(Y'_j, \set{Y_1, Y_2,\ldots, Y_s})$, and $y'_j = \tsum(Y'_j)$. 
\item Consider the $r$ integers $c_1, c_2,\ldots, c_r$ defined as follows: $c_1 = 0$, and for $j > 1$, $c_j = \sum_{\substack{i \in [j - 1] \\ y'_i = y'_j}} m(Y'_i)$. 
\end{enumerate}
Then, 
\[
\ndel(G, (Y, \set{Y_1, Y_2,\ldots, Y_s})) = \Lb{\prod_{j \in [r]} \binom{\beta(y'_j) - c_j}{m(j)}} \Lb{\prod_{i \in [s]} \ndel(\set{y_i}, Y_i)}.
\]
\end{lemma}
\begin{proof}
To prove the lemma, observe that $\ndel(G, (Y, \set{Y_1, Y_2,\ldots, Y_s}))$ is precisely product $\zeta_1 \cdot \zeta_2$, where $\zeta_1$ is the number of choices for $s$ components $H_1, H_2,\ldots, H_s$ such that each $H_i$ corresponds to the partition $Y_i$, and $\zeta_2$ is the number of ways in which the $s$ components can be split w.r.t. to the partitions $Y_1, Y_2,\ldots, Y_s$. Observe that $\zeta_2$ is precisely the product of the number of deletions of each $H_i$ w.r.t. $(\set{y_i}, Y_i)$, i.e., $\zeta_2 = \prod_{i \in [s]} \ndel(\set{y_i}, Y_i)$. 

We now calculate $\zeta_1$, which is the number of choices for the components $H_1, H_2,\ldots, H_s$. We ought to be careful while counting the number of choices, as 
we cannot treat two components $H_i$ and $H_{i'}$ as interchangeable simply because $\card{H_i} = \card{H_{i'}}$; we also need to consider the corresponding partitions $Y_i$ and $Y_{i'}$ that respectively determine how $H_i$ and $H_{i'}$ will be split into smaller components. On the other hand, if the partitions $Y_i$ and $Y_{i'}$ are identical, then $H_i$ and $H_{i'}$ are indeed interchangeable. 
Thus, for each $Y'_j  \in \set{Y'_1, Y'_2,\ldots, Y'_r} = \US(\set{Y_1, Y_2,\ldots, Y_s})$, we need to count the number of ways in which we can choose $\mul(Y'_j, \set{Y_1, Y_2,\ldots, Y_s}) = m(Y'_j)$ components of size $y'_j = \tsum(Y'_j)$.  Now, notice that for distinct $j, j' \in [r]$, we have $Y'_{j} \neq Y'_{j'}$, but we may still have $y'_j = \tsum(Y'_j) = \tsum(Y'_{j'}) = y'_{j'}$. So we will have to account for this possibility as well. 

For each $j = 1$ to $r$ in this order, we count the number of ways in which we can choose $m(Y'_j)$ components of size $y'_{j}$. Notice that for $j = 1$, there are $\binom{\beta(y'_j)}{m(Y'_j)}$ choices for the required components. Consider $j > 1$. Out of the $\beta(y'_j)$ components of $G$ of size $y'_j$, for each $i < j$ such that $y'_i = y'_j$, we  have already chosen exactly $m(Y'_i)$ components; so when we come to $j$, we will have already chosen $\sum_{\substack{i \in [j - 1] \\ y'_i = y'_j}} m(Y'_i) = c_j$ components and we will be left with $\beta(y'_j) - c_j$ components of size $y'_j$. Hence we can choose $m(Y'_j)$ of those components in $\binom{\beta(y'_j) - c_j}{m(Y'_j)}$ ways. Thus the total number of choices for the $s$ components is precisely $\prod_{j \in [r]} \binom{\beta(y'_j) - c_j}{m(Y'_j)}$.  
\end{proof}

\subsection*{An algorithm for \#BCD}
To design our algorithm for \countbcd, we rely on some of the arguments that we used in our kernel for \bcd\ to bound the number of vertices that belong to the non-clique components. Specifically, we rely on the fact that the main reduction rule that we used to bound the non-clique part of the graph (Reduction Rule~\ref{rule:deletion-clique-or-indset}) does not change the number of solutions, and the other rules (Reduction Rules~\ref{rule:deletion-sanity-check-2} and \ref{rule:non-adjacent}) simply identify no-instances. By exhaustively applying these three rules, we can derive the following result. 
\begin{lemma}\label{lem:deletion-Z}
There is an algorithm that given an instance $(G, k, \eta)$ of \countbcd, runs in polynomial time, and either correctly reports that $(G, k, \eta)$ is a no-instance of \bcd, or  returns an instance $(G', k', \eta)$ of \countbcd\ and a set $Z \subseteq V(G')$ such that $\nsol(G, k, \eta) = \nsol(G', k', \eta)$, $\card{Z} = \cO(k^{4})$, $G' - Z$ is a cluster graph, and there is no edge between $Z$ and $V(G') \setminus Z$ (that is, $Z$ contains all the non-clique components of $G'$). 
\end{lemma}

Consider an instance $(G, k, \eta)$ of \countbcd. From now on, we assume that we have run the algorithm of Lemma~\ref{lem:deletion-Z} on our input instance and found the set $Z$; let us denote the resulting instance by $(G, k, \eta)$ as well. Now, to find $\nsol(G, k, \eta)$, we guess how each solution $F \subseteq E(G)$ intersects with $Z$. Notice that each edge of $G$ is either an edge of $G[Z]$ or an edge of $G - Z$, and hence each solution $F \subseteq E(G)$ for $(G, k, \eta)$ is the disjoint union of $F_0 = F \cap E(G[Z])$ and $F_1 = F \cap E(G - Z)$. Thus the components of $G - F$ are precisely the components of $G[Z] - F_0$ and the components of $(G - Z) - F_1$. That is, $\CC(G - F) = \CC(G[Z] - F_0) \cup \CC((G - Z) - F_1)$, and hence $\CS(G - F) = \CS(G[Z] - F_0) \cup \CS((G - Z) - F_1)$; as $G - F$ is $\eta$-balanced, we can conclude that $\CS(G[Z] - F_0) \cup \CS((G - Z) - F_1)$ is $\eta$-balanced. Finally, the facts that $G - F$ is an $\eta$-balanced cluster graph and $F_0 \cap F_1 = \emptyset$ together imply that $G[Z] - F_0$ and $(G - Z) - F_1$ are $\eta$-balanced cluster graphs. In particular, $F_1$ is a solution for the instance $(G - Z, k - \card{F_0}, \eta)$ of \bcd. And recall that by Lemma~\ref{lem:deletion-Z}, $G - Z$ is a cluster graph, and therefore, we can associate $F_1$ with a pair $(Y, \set{Y_1, Y_2,\ldots, Y_s})$, where $Y$ is a partition of $\card{V(F_1)}$. And by going over all possible $(Y, \set{Y_1, Y_2,\ldots, Y_s})$, we can count the number of solutions.  %
We formalise this below.    

Formally, for $F_0 \subseteq E(G[Z])$ and a pair $(Y, \set{Y_1, Y_2,\ldots, Y_s})$, we say that the pair $(F_0, (Y, \set{Y_1, Y_2,\ldots, Y_s}))$ is good for $(G, k, \eta)$ if the following conditions hold:
\begin{itemize}
\item $F_0$ is a solution for the instance $(G[Z], k, \eta)$ of \bcd;
\item $k - \card{F_0} \geq 1$;
\item $(Y, \set{Y_1, Y_2,\ldots, Y_s})$ is such that $Y = \set{y_1, y_2,\ldots, y_s}$ is a partition of $\ell$, for some $\ell \in [2(k - \card{F_0})] \setminus \set{1}$, and $Y_i$ is a partition of $y_i$ for every $i \in [s]$;
\item $Y$ is $(G - Z)$-valid;
\item $\CS(G[Z] - F_0) \cup ((\CS(G - Z) \setminus Y) \cup \hat Y)$ is $\eta$-balanced, 
where $\hat Y = \bigcup_{i \in [s]}Y_i$; 
\item $\sum_{i \in [s]} \spp(Y_i) \leq k - \card{F_0}$.
\end{itemize}

Notice that for each $F_0$ and $(Y, \set{Y_1, Y_2,\ldots, Y_s})$, we can check in polynomial time if $(F_0, Y, \set{Y_1, Y_2,\ldots, Y_s})$ is good for $(G, k, \eta)$. Notice also that each good $(F_0, Y, \set{Y_1, Y_2,\ldots, Y_s})$ only accounts for solutions of the form $F = F_0 \cup F_1$, where $F_1 \neq \emptyset$. We also need to account for the number of solutions that are fully contained in $E(G[Z])$. Formally, let $f(G, k, \eta, Z)$ be the number of sets $F_0 \subseteq E(G[Z])$ such that $F_0$ is a solution for the instance $(G, k, \eta)$ of \bcd. 

Putting all this together, we have 
\begin{equation*}
\nsol(G, k, \eta) = f(G, k, \eta, Z) + \sum_{(F_0, Y, \set{Y_1, Y_2,\ldots, Y_s})} \ndel(G - Z, (Y, \set{Y_1, Y_2,\ldots, Y_s}),
\tag{\textsl{Eq. BCD}}\label{eq:bcd}
\end{equation*}
where the summation is over all good $(F_0, Y, \set{Y_1, Y_2,\ldots, Y_s})$. 

The algorithm for \countbcd\ is now straightforward. We compute $\nsol(G, k, \eta)$ using Equation (\ref{eq:bcd}). To do this, we first compute $f(G, k, \eta, Z)$ by going over all subsets $F_0 \subseteq E(G[Z])$ of size at most $k$, and checking if $F_0$ is indeed a solution for $(G, k, \eta)$; this takes time $\sum_{i = 0}^k\binom{\card{Z}}{i} n^{\cO(1)} = \binom{\cO(k^4)}{k} n^{\cO(1)} = 2^{\cO(k \log k)} n^{\cO(1)}$. We then go over all pairs $(F_0, (Y, \set{Y_1, Y_2,\ldots, Y_s}))$ that are good for $(G, k, \eta)$, and compute $\ndel(G - Z, (Y, \set{Y_1, Y_2,\ldots, Y_s})$ using Lemma~\ref{Lem:count-deletion-cluster}. As $F_0 \subseteq E(G[Z])$ and $\card{E(G[Z])} \leq \binom{\card{Z}}{2}$, $F_0$ has at most $\sum_{i = 0}^k \binom{\binom{\card{Z}}{2}}{i} = 2^{\cO(k \log k)}$ choices and $(Y, \{Y_1, Y_2,\ldots, Y_s\})$ has $2^{\cO(k)}$ choices (by Lemma~\ref{lem:completion-bound}), this step also takes time $2^{\cO(k \log k)} n^{\cO(1)}$. Notice that all the arithmetic operations we need to perform to compute $\ndel(G - Z, (Y, \set{Y_1, Y_2,\ldots, Y_s}))$ only involve numbers that can be encoded using $n^{\cO(1)}$ bits. We can therefore compute $\nsol(G, k, \eta)$ in time \countbceruntime. We thus have the following result. 

\begin{theorem}\label{thm:deletion-count}
    \countbcd\ admits an algorithm that runs in time \countbceruntime. 
\end{theorem}

\subsubsection*{An algorithm for \#BCE}

To design an algorithm for \countbce, we follow the same strategy that we used for \countbcd. But there are differences. We first discuss them, and outline our main arguments. 

\subparagraph*{Idea behind our algorithm for \countbce.} Consider an instance $(G, k, \eta)$ of \countbce. As in the case of \countbcd, we start with a subset of vertices $Z$ with similar properties: $\card{Z} = k^{\cO(1)}$, $G - Z$ is a cluster graph and there are no edges between $Z$ and $V(G) \setminus Z$. Every solution $F \subseteq \binom{V(G)}{2}$ is then the disjoint union of three (possibly empty) parts, $F_0$, $F_1$ and $F_2$, where $F_0$ is the intersection of $F$ with $G[Z]$, $F_1$ is the set of edges that we delete (apart from those in $F_0$) and $F_2$ is the set of edges that we add (again, apart from those in $F_0$). Notice that $G \triangle F_0$ must be a cluster graph. 
Now, to calculate $\nsol(G, k, \eta)$, we go over all possible choices for $F_0$, and count the number of possibilities for $F_1$ and $F_2$ that are consistent with $F_0$. And as before, we associate $F_1$ and $F_2$ with sets of partitions of $\card{V(F_1)}$ and $\card{V(F_2)}$. 
But quite unlike in the case of \countbcd, where for each $F_0$, the rest of the solution was ``untouched'' by $Z$ (meaning completely contained in $G - Z$), we do not have that guarantee for \countbce, as we  can now add edges to $G$ (i.e., $F_2$) and some of those edges could be incident with vertices in $Z$. So while counting the number of solutions we must ensure that no two components of $G[Z] \triangle F_0$ are merged together when we add the edges in $F_2$. A component of $G[Z] \triangle F_0$, however, may be merged with a component of $G - Z$. Now, as we did in our algorithm for (the decision version of) \bce, we think of $G \triangle F$ as being obtained from $G \triangle F_0$ by a two step process: first by deleting the edges of $F_1$, and then adding the edges of $F_2$. Notice that $F_1$ and $F_2$ are disjoint. So in this two step process,  $F_2$ does not merge together two components that were split by $F_1$; for example, consider any  component $H$ of $G - Z$ that was split into smaller components $H_1, H_2,\ldots, H_r$ when we deleted $F_1$; notice that $F_2$ does not merge two $H_i$s together as $F_1 \cap F_2 = \emptyset$.  We must take this fact into account while counting the number of choices for $F_2$. To deal with all these restrictions on certain pairs of components that cannot be merged together, we reduce the problem to \crbcc\. %
\lipicsEnd

\paragraph*{Solving \countbce}

We are finally ready to design an algorithm for \countbce. The last ingredient we need for our algorithm is a set $Z \subseteq V(G)$ that contains all the non-clique components of $G$. To find such a $Z$, we turn to our kernelization algorithm for \bce, and in particular, Reduction Rules~\ref{rule:editing-edge}, \ref{rule:editing-non-edge} and \ref{rule:editing-number-of-visible}  that we used to bound the number of vertices that belong to the non-clique components of $G$. Reduction Rules~\ref{rule:editing-edge} and \ref{rule:editing-non-edge} respectively delete and add edges that belong to every solution, and therefore preserve the number of solutions, and Reduction Rule~\ref{rule:editing-number-of-visible} simply eliminates no-instances. By exhaustively applying these three rules, we can derive the following result. 

\begin{lemma}\label{lem:editing-Z}
There is an algorithm that given an instance $(G, k, \eta)$ of \countbce, runs in polynomial time, and either correctly reports that $(G, k, \eta)$  is a no-instance of \bce, or returns an instance $(G', k', \eta)$ of \countbce\ and a set $Z \subseteq V(G')$ such that $\nsol(G, k, \eta) = \nsol(G', k', \eta)$, $\card{Z} = \cO(k^{3})$, $G' - Z$ is a cluster graph, and there is no edge between $Z$ and $V(G') \setminus Z$ (that is, $Z$ contains all the non-clique components of $G'$). 
\end{lemma}

Consider an instance $(G, k, \eta)$ of \countbce. From now on, we assume that we have run the algorithm of Lemma~\ref{lem:editing-Z} on our input instance and found the set $Z$; let us denote the resulting instance by $(G, k, \eta)$ as well. %
Notice that each solution $F \subseteq \binom{V(G)}{2}$ for $(G, k, \eta)$ is the disjoint union of three sets $F_0$, $F_1$ and $F_2$, where $F_0 = F \cap \binom{Z}{2}$ is the intersection of $F$ with $G[Z]$, $F_1 = (F \cap E(G)) \setminus F_0 $ is the edges that we delete from $G$ (apart from those in $F_0$, and thus $F_1 \subseteq E(G - Z)$), and $F_2 = (F \cap \binom{V(G)}{2} \setminus E(G)) \setminus F_0$ is the set of edges that we add to $G$ (apart from those in $F_0$). 
As $F = F_0 \cup F_1 \cup F_2$ is a solution, and $F_1 \cup F_2$ does not modify any (non-)edge with both its endpoints in $Z$, we can conclude that $G[Z] \triangle F_0$ is cluster graph; this, along with the fact that $G - Z$ is a cluster graph, implies that $G \triangle F_0$ is indeed a cluster graph. Notice that we cannot, however,  conclude that $F_0$ is a solution for the instance $(G[Z], k, \eta)$, as $G[Z] \triangle F_0$ need not be $\eta$-balanced. %
To count the number of solutions $F = F_0 \cup F_1 \cup F_2$ for $(G, k, \eta)$, we will now group the solutions into four classes depending on whether or not $F_1$ or $F_2$ is empty, and we will show that the number of solutions in each class can be computed in time \countbceruntime.  
In what follows, we will write $\bm{Y}$ as a shorthand for $(Y, \set{Y_1, Y_2,\ldots, Y_s})$. 

\begin{enumerate}
\item
{\bf When $F_1 \neq \emptyset$ and $F_2 \neq \emptyset$.} Let $A_1$ be the number of solutions $F = F_0 \cup F_1 \cup F_2$ for $(G, k, \eta)$ such that $F_1 \neq \emptyset$ and $F_2 \neq \emptyset$. 
For $F_0 \subseteq \binom{Z}{2}$ and $\bm{Y}$, we say that the pair $(F_0, \bm{Y})$ is semi-good for $(G, k, \eta)$ if the following conditions hold:
\begin{itemize}
\item $G \triangle F_0$ is a cluster graph;
\item $k - \card{F_0} \geq 1$;
\item $(Y, \set{Y_1, Y_2,\ldots, Y_s})$ is such that $Y = \set{y_1, y_2,\ldots, y_s}$ is a partition of $\ell$, for some $\ell \in [2(k - \card{F_0})] \setminus \set{1}$, and $Y_i$ is a partition of $y_i$ for every $i \in [s]$;
\item $Y$ is $(G - Z)$-valid;
\item $\sum_{i \in [s]} \spp(Y_i) \leq k - \card{F_0}$.
\end{itemize}

Consider a pair $(F_0, \bm{Y})$ that is semi-good for $(G, k, \eta)$. Based on $(F_0, \bm{Y})$, we define an instance $I_{F_0, \bm{Y}} = (G', \mu', k', \eta')$ of \crbcc\ as follows. 
For $i \in [s]$, let $Y_i = \set{y_{i1}, y_{i2},\ldots, y_{ir_i}}$. We first define a cluster graph $G'$ and a colouring function $\fn{\mu'}{\CC(G')}{[s + 1]_0}$ obtained from $G \triangle F_0$ as follows. Fix $s$ components $H_1, H_2,\ldots, H_s$ of $G - Z$ such that $\card{H_i} = x_i$ for every $i \in [s]$; notice that the $H_i$s are components of $G \triangle F_0$ as well. Let $G'$ be a deletion of $G \triangle F_0$ w.r.t. $(Y, \set{Y_1, Y_2,\ldots, Y_s})$ in which each $H_i$ is split into $r_i$ components $H_{i1}, H_{i2},\ldots, H_{ir_i}$ such that $\card{H_{ij}} = y_{ij}$ for every $j \in [r_i]$. We will  define $\mu'$ in such a way that no two $H_{ij}$ and $H_{ij'}$ can be merged together, and no two components of $G'$ that are fully contained in $Z$ can be merged together either. With this in mind, we define $\mu'(H_{ij}) = i$ for every $i \in [s]$, $j \in [r_i]$, and for every component $H \neq H_{ij}$ for any $i \in [s]$, $j \in [r_i]$, we define $\mu'(H) = s + 1$ if  $V(H) \subseteq Z$ and $\mu'(H) = 0$ otherwise. Also, we define $k' = k - \card{F_0} - \sum_{i \in [s]}\spp(Y_i)$, and $\eta' = \eta$.  

It is straightforward to argue that $F = F_0 \cup F_1 \cup F_2$ is solution for $(G, k, \eta)$ with $F_1, F_2 \neq \emptyset$ if and only if $\card{V(F_1)}$ can be associated with an appropriately defined $\bm{Y}$ such that $(F_0, \bm{Y})$ is semi-good for $(G, k, \eta)$ and $F_2$ is a non-empty solution for the instance $I_{F_0, \bm{Y}}$.  We therefore have
\[
A_1 = \sum_{(F_0, \bm{Y})} \nsol^{\neq \emptyset}(I_{F_0, \bm{Y}}), 
\]
where the summation is over all pairs $(F_0, \bm{Y})$ that are semi-good for $(G, k, \eta)$, and $\nsol^{\neq \emptyset}(I_{F_0, \bm{Y}})$ is the number of non-empty solutions for the instance $I_{F_0, \bm{Y}}$ of \crbcc. 
Also, we can compute $A_1$ in time \countbceruntime. To see this, notice that as $s \leq \card{V(F_1)} \leq 2k$, the function $\mu$ uses $s + 2$ colours (including $0$). Therefore, by Lemma~\ref{lem:colours-counting}, we can compute $\nsol^{\neq \emptyset}(I_{F_0, \bm{Y}})$ in time \countbceruntime. And the number of choices for the pair $(F_0, \bm{Y})$ is $\sum_{j = 0}^k\binom{\binom{\card{Z}}{2}}{j} \cdot 2^{\cO(k)} = \card{Z}^{\cO(k)} \cdot 2^{\cO(k)} = 2^{\cO(k \log k)}$. We can therefore compute $A_1$ in time \countbceruntime.  

\item {\bf When $F_1 = \emptyset$ and $F_2 \neq \emptyset$.} Let $A_2$ be the number of solutions $F = F_0 \cup F_1 \cup F_2$ for $(G, k, \eta)$ such that $F_1 = \emptyset$ and $F_2 \neq \emptyset$. As in the previous case, we will show that computing $A_2$ amounts to solving an instance of \countcrbcc. We define an instance $I_{F_0} = (G'', \mu'', k'', \eta'')$ of \crbcc\ as follows: $G'' = G \triangle F_0$, $k'' = k - \card{F_0}$, $\eta'' = \eta$, and for a component $H$ of $G''$, $\mu''(H) = 1$ if $V(H) \subseteq Z$ and $\mu''(H) = 0$ otherwise. Again, it is straightforward to argue that $F = F_0 \cup F_2$ with $F_2 \neq \emptyset$ is a solution for $(G, k, \eta)$ if and only if $F_2$ is a non-empty solution for the instance $I_{F_0}$ of \crbcc. Therefore, 
\[
A_2 = \sum_{F_0} \nsol^{\neq \emptyset}(I_{F_0}),
\]
where the summation is over all $F_0 \subseteq \binom{Z}{2}$ such that $G \triangle F_0$ is a cluster graph and $\card{F_0} \leq k$. Now, as $\mu''$ uses only $2$ colours, by Lemma~\ref{lem:colours-counting}, we can compute $\nsol^{\neq \emptyset}(I_{F_0})$ in time \algobccruntime. Also, $F_0$ has $\sum_{j = 0}^k\binom{\binom{\card{Z}}{2}}{j} = 2^{\cO(k \log k)}$ choices. We can therefore compute $A_2$ in time \countbceruntime. 

\item {\bf When $F_1 \neq \emptyset$ and $F_2 = \emptyset$.} Let $A_3$ be the number of solutions $F = F_0 \cup F_1 \cup F_2$ for $(G, k, \eta)$ such that $F_1 \neq \emptyset$ and $F_2 = \emptyset$. Notice that $F = F_0 \cup F_1$ with $F_1 \neq \emptyset$ is a solution for $(G, k, \eta)$ if and only if $F_1$ is a non-empty solution for the instance $(G \triangle F_0, k - \card{F_0}, \eta)$ of \bcd\ such that $F_2 \cap E(G[Z] \triangle F_0) = \emptyset$. It is then straightforward to associate $F_1$ with $\bm{Y}$ such that $(G - Z) - F_1$ is a deletion of $G - Z$ w.r.t. $\bm{Y}$.  Therefore,
\[
A_3 = \sum_{(F_0, \bm{Y})} \ndel(G - Z, \bm{Y}),
\]
where the summation is over all pairs $(F_0, \bm{Y})$ such that 
\begin{itemize}
\item $F_0 \subseteq \binom{Z}{2}$ and $G \triangle F_0$ is a cluster graph,
\item $k - \card{F_0} \geq 1$,
\item $(Y, \set{Y_1, Y_2,\ldots, Y_s})$ is such that $Y = \set{y_1, y_2,\ldots, y_s}$ is a partition of $\ell$, for some $\ell \in [2(k - \card{F_0})] \setminus \set{1}$, and $Y_i$ is a partition of $y_i$ for every $i \in [s]$;
\item $Y$ is $(G - Z)$-valid;
\item $(\CS(G \triangle F_0)  \setminus Y) \cup \hat Y$ is $\eta$-balanced, 
where $\hat Y = \bigcup_{i \in [s]}Y_i$; 
\item $\sum_{i \in [s]} \spp(Y_i) \leq k - \card{F_0}$.
\end{itemize}
Again, we can argue that the number of choices for $(F_0, \bm{Y})$ is $2^{\cO(k \log k)}$, and we can thus compute $A_3$ in time \countbceruntime. 

\item {\bf When $F_1 = F_2 = \emptyset$.} Let $A_4$ be the number of solutions $F = F_0 \cup F_1 \cup F_2$ for $(G, k, \eta)$ such that $F_1 = \emptyset$ and $F_2 \neq \emptyset$. Notice in this case that $F_0$ must be a solution for the instance $(G, k, \eta)$ of \bce. Thus, by going over all possible choices for $F_0$ and checking if it is a solution, we can compute $A_4$. In particular, 
$A_4 = \sum_{F_0} 1$, 
where the summation is over all $F_0 \subseteq \binom{Z}{2}$ such that $F_0$ is a solution for $(G, k, \eta)$. As before, $F_0$ has $2^{\cO(k \log k)}$ choices, and therefore, we can compute $A_4$ in time \countbceruntime. 
\end{enumerate}

We have shown that we can compute $A_1, A_2, A_3, A_4$, and consequently $\nsol_{BCE}(G, k, \eta) = A_1 + A_2 + A_3 + A_4$, in time \countbceruntime. We thus have the following result. 

\begin{theorem}\label{thm:counting-editing}
\countbce\ admits an algorithm that runs in time \countbceruntime. 
\end{theorem}

%% file: nph_bcc.tex
\section{\NPH ness of \bccfull}
\label{sec:completion-nph}
To show that \bcc\ is \NPH, we use a reduction due to Froese et al.~\cite{DBLP:conf/aaai/FroeseKN22}, which shows that a related problem called {\sc Cluster Transformation by Edge Addition} is \NPH. In {\sc Cluster Transformation by Edge Addition}, we are given a cluster graph $G$ and a non-negative integer $k$, and we have to decide if we can add \emph{exactly} $k$ edges to $G$ in such a way that the resulting graph is also a cluster graph. 
Froese et al.~\cite{DBLP:conf/aaai/FroeseKN22} showed that {\sc Cluster Transformation by Edge Addition} is \NPH\ by designing a reduction from a problem called {\sc Numerical 3D-Matching}. We will argue that this reduction also shows the \NPH ness of \bcc. Below, we reproduce the reduction from~\cite{DBLP:conf/aaai/FroeseKN22} and outline its correctness. 

In {\sc Numerical 3D-Matching}, which is known to be strongly \NPH ~\cite{DBLP:journals/siamcomp/GareyJ75}, we are given positive integers $t$, $a_1, a_2,\ldots, a_n$, $b_1, b_2,\ldots, b_n$ and $c_1, c_2,\ldots, c_n$, and we have to decide if there exist bijections $\fn{\alpha, \beta, \gamma}{[n]}{[n]}$ such that $a_{\alpha(i)} + b_{\beta(i)} + c_{\gamma(i)} = t$ for every $i \in [n]$.  
We reduce {\sc Numerical 3D-Matching} to \bcc. Informally, the reduction works as follows. Given an instance of {\sc Numerical 3D-Matching}, we construct a cluster graph $G$ by introducing a ``small'' clique corresponding to each $a_i$, a ``medium'' clique corresponding to each $b_i$ and a ``large'' clique corresponding to each $c_i$. We set $k$ appropriately so that the only way we can turn $G$ into a $0$-balanced cluster graph by adding at most $k$ is by merging together one small clique, one medium clique and one large clique.  

Consider an instance $I = (t, (a_i)_{i \in [n]}, (b_i)_{i \in [n]}, (c_i)_{i \in [n]})$ of {\sc Numerical 3D-Matching}. As the problem is strongly \NPH, we assume without loss of generality that $a_i, b_i, c_i \leq n^d$ for some constant $d$. We also assume that $t > a_i, b_i, c_i$ for every $i \in [n]$ and that $\sum_{i \in [t]} a_i + b_i + c_i = nt$, as otherwise $I$ is clearly a no-instance. 

We construct an instance $(G, k, \eta)$ of \bcc\ as follows. First of all, we set $\eta = 0$. Now, let $A = n^{2d}$, $B = n^{3d}$ and $C = n^{7d}$. For $i \in [n]$, let $a'_i = a_i + A$, $b'_i = b_i + B$ and $c'_i = c_i + C$, and we  add three cliques of sizes $a'_i, b'_i$ and $c'_i$ to $G$; we call these cliques small, medium and large, respectively.  For convenience, we refer to these cliques by their sizes; for example, we may refer to the clique $a'_i$. Let $t' = t + A + B + C$, and we set
\[
k = n \binom{t'}{2} - \card{E(G)} = n \binom{t'}{2} - \sum_{i \in [n]} \Lb{ \binom{a'_i}{2} + \binom{b'_i}{2} + \binom{c'_i}{2} }. 
\]

We now argue that $I$ is a yes-instance of {\sc Numerical 3D-Matching} if and only if $(G, k, \eta)$ is a yes-instance of \bcc. The forward direction is straightforward. Assume  that $I$ is a yes-instance, and let $\fn{\alpha, \beta, \gamma}{[n]}{[n]}$ be bijections such that $a_{\alpha(i)} + b_{\beta(i)} + c_{\gamma(i)} = t$ for every $i \in [n]$. Then  the cluster graph $G'$ obtained from $G$ by merging the three cliques  $a'_{\alpha(i)}$,  $b'_{\beta(i)}$ and   $c'_{\gamma(i)}$ for each $i \in [n]$ is $0$-balanced. And it is straightforward to verify that $\card{E(G') \setminus E(G)} = k$.   

The backward direction is more involved. Assume that $(G', k, \eta)$ is a yes-instance of \bcc, and let $F \subseteq \binom{V(G)}{2}$ be a solution for $(G, k, \eta)$. As $\eta = 0$, the cluster graph $G + F$ is $0$-balanced. To show that $I$ is a yes-instance, we use the following sequence of arguments. 
\begin{enumerate}
    \item For $n \geq 3$ and $d \geq 1$, we have $n^{10d + 1} \leq k \leq 2n^{d + 1}$~\cite[Lemma 7]{DBLP:conf/aaai/FroeseKN22}.  
    
    \item\label{item:no-two-large} As a consequence, no two large cliques can be merged together, because each large clique has size $\Omega(C) = \Omega(n^{7d})$, and merging together two large cliques needs the addition of $\Omega(n^{14})$ edges~\cite[Lemma 8]{DBLP:conf/aaai/FroeseKN22}. 
    
    \item\label{item:at-least-n} This implies that the graph $G + F$ has at least $n$ components (one corresponding to each large clique of $G$).
    \item\label{item:small-medium} As $G + F$ is $0$-balanced, we must merge each small clique and each medium clique with a large clique. To see this, notice that the total number of vertices in all the small and medium cliques together is $\sum_{i \in [n]} (a'_i + b'_i) = \sum_{i \in [n]} (a_i + A + b_i + B) = nA + nB + \sum_{i \in [n]} (a_i + b_i) = \cO(n^{2d + 1} + n^{3d + 1} + n^{d + 1}) = \cO(n^{3d + 1})$,  whereas each large clique has size $\Omega(C) = \Omega(n^{7d})$. Item~\ref{item:no-two-large} now implies that each small clique is merged with exactly one large clique, and similarly each medium clique is merged with exactly one large clique. %
    \item\label{item:exactly-n} Items~\ref{item:at-least-n} and \ref{item:small-medium} together imply that $G + F$ has exactly $n$ components. 

    \item\label{item:exactly-t'} As $\card{G + F} = \card{G} = \sum_{i \in [n]} (a'_i + b'_i + c'_i) = \sum_{i \in [n]} (a_i + A + b_i + B + c_i + C) = n(A + B + C + t) = nt'$, and $G + F$ is $0$-balanced with exactly $n$ components, we can conclude that each component of $G + F$ has size exactly $t'$. %

    \item\label{item:compatible} But then, as $t' = t + A + B + C = t + n^{2d} + n^{3d} + n^{7d}$ and $t \leq 3n^{d}$, each component of $G + F$ must be formed by merging together one small, one medium and one large clique~\cite[Lemma 13]{DBLP:conf/aaai/FroeseKN22}. In other words, there exist bijections $\fn{\alpha', \beta', \gamma'}{[n]}{[n]}$ such that for each $i \in [n]$, we have $a'_{\alpha'(i)} + b'_{\beta'(i)} + c'_{\gamma'(i)} = t'$

    \item Based on Item~\ref{item:compatible}, it is  straightforward to argue that $I$ is a yes-instance~\cite[Proof of Theorem 6]{DBLP:conf/aaai/FroeseKN22}. For each $i \in [n]$, we have $a_{\alpha'(i)} + b_{\beta'(i)} + c_{\gamma'(i)} = a'_{\alpha'(i)} - A + b'_{\beta'(i)} - B + c'_{\gamma'(i)} - C = t' - A - B - C = t$. This shows that $I$ is a yes-instance of {\sc Numerical 3D-Matching}.    
\end{enumerate}

%% file: Full-version-MAIN.bbl
\begin{thebibliography}{10}

\bibitem{DBLP:journals/jda/Abu-Khzam17}
Faisal~N. Abu{-}Khzam.
\newblock On the complexity of multi-parameterized cluster editing.
\newblock {\em J. Discrete Algorithms}, 45:26--34, 2017.
\newblock \href {https://doi.org/10.1016/j.jda.2017.07.003}
  {\path{doi:10.1016/j.jda.2017.07.003}}.

\bibitem{DBLP:conf/aaai/Agrawal00023}
Akanksha Agrawal, Tanmay Inamdar, Saket Saurabh, and Jie Xue.
\newblock Clustering what matters: Optimal approximation for clustering with
  outliers.
\newblock In Brian Williams, Yiling Chen, and Jennifer Neville, editors, {\em
  Thirty-Seventh {AAAI} Conference on Artificial Intelligence, {AAAI} 2023,
  Thirty-Fifth Conference on Innovative Applications of Artificial
  Intelligence, {IAAI} 2023, Thirteenth Symposium on Educational Advances in
  Artificial Intelligence, {EAAI} 2023, Washington, DC, USA, February 7-14,
  2023}, pages 6666--6674. {AAAI} Press, 2023.
\newblock \href {https://doi.org/10.1609/aaai.v37i6.25818}
  {\path{doi:10.1609/aaai.v37i6.25818}}.

\bibitem{DBLP:journals/jacm/AilonCN08}
Nir Ailon, Moses Charikar, and Alantha Newman.
\newblock Aggregating inconsistent information: Ranking and clustering.
\newblock {\em J. {ACM}}, 55(5):23:1--23:27, 2008.
\newblock \href {https://doi.org/10.1145/1411509.1411513}
  {\path{doi:10.1145/1411509.1411513}}.

\bibitem{althoff2011balanced}
Tim Althoff, Adrian Ulges, and Andreas Dengel.
\newblock Balanced clustering for content-based image browsing.
\newblock In {\em Informatiktage}, pages 27--30, 2011.

\bibitem{DBLP:journals/mp/AprileDFH23}
Manuel Aprile, Matthew Drescher, Samuel Fiorini, and Tony Huynh.
\newblock A tight approximation algorithm for the cluster vertex deletion
  problem.
\newblock {\em Math. Program.}, 197(2):1069--1091, 2023.
\newblock URL: \url{https://doi.org/10.1007/s10107-021-01744-w}, \href
  {https://doi.org/10.1007/S10107-021-01744-W}
  {\path{doi:10.1007/S10107-021-01744-W}}.

\bibitem{DBLP:journals/mst/BandyapadhyayFGPS23}
Sayan Bandyapadhyay, Fedor~V. Fomin, Petr~A. Golovach, Nidhi Purohit, and
  Kirill Simonov.
\newblock Lossy kernelization of same-size clustering.
\newblock {\em Theory Comput. Syst.}, 67(4):785--824, 2023.
\newblock \href {https://doi.org/10.1007/s00224-023-10129-9}
  {\path{doi:10.1007/s00224-023-10129-9}}.

\bibitem{DBLP:journals/ml/BansalBC04}
Nikhil Bansal, Avrim Blum, and Shuchi Chawla.
\newblock Correlation clustering.
\newblock {\em Mach. Learn.}, 56(1-3):89--113, 2004.
\newblock \href {https://doi.org/10.1023/B:MACH.0000033116.57574.95}
  {\path{doi:10.1023/B:MACH.0000033116.57574.95}}.

\bibitem{DBLP:journals/jcb/Ben-DorSY99}
Amir Ben{-}Dor, Ron Shamir, and Zohar Yakhini.
\newblock Clustering gene expression patterns.
\newblock {\em J. Comput. Biol.}, 6(3/4):281--297, 1999.
\newblock \href {https://doi.org/10.1089/106652799318274}
  {\path{doi:10.1089/106652799318274}}.

\bibitem{DBLP:conf/stoc/BjorklundHKK07}
Andreas Bj{\"{o}}rklund, Thore Husfeldt, Petteri Kaski, and Mikko Koivisto.
\newblock Fourier meets m{\"{o}}bius: fast subset convolution.
\newblock In David~S. Johnson and Uriel Feige, editors, {\em Proceedings of the
  39th Annual {ACM} Symposium on Theory of Computing, San Diego, California,
  USA, June 11-13, 2007}, pages 67--74. {ACM}, 2007.
\newblock \href {https://doi.org/10.1145/1250790.1250801}
  {\path{doi:10.1145/1250790.1250801}}.

\bibitem{DBLP:journals/jda/Bocker12}
Sebastian B{\"{o}}cker.
\newblock A golden ratio parameterized algorithm for cluster editing.
\newblock {\em J. Discrete Algorithms}, 16:79--89, 2012.
\newblock \href {https://doi.org/10.1016/j.jda.2012.04.005}
  {\path{doi:10.1016/j.jda.2012.04.005}}.

\bibitem{DBLP:conf/cie/BockerB13}
Sebastian B{\"{o}}cker and Jan Baumbach.
\newblock Cluster editing.
\newblock In Paola Bonizzoni, Vasco Brattka, and Benedikt L{\"{o}}we, editors,
  {\em The Nature of Computation. Logic, Algorithms, Applications - 9th
  Conference on Computability in Europe, CiE 2013, Milan, Italy, July 1-5,
  2013. Proceedings}, volume 7921 of {\em Lecture Notes in Computer Science},
  pages 33--44. Springer, 2013.
\newblock \href {https://doi.org/10.1007/978-3-642-39053-1\_5}
  {\path{doi:10.1007/978-3-642-39053-1\_5}}.

\bibitem{DBLP:conf/apbc/BockerBBT08}
Sebastian B{\"{o}}cker, Sebastian Briesemeister, Quang Bao~Anh Bui, and Anke
  Tru{\ss}.
\newblock A fixed-parameter approach for weighted cluster editing.
\newblock In Alvis Brazma, Satoru Miyano, and Tatsuya Akutsu, editors, {\em
  Proceedings of the 6th Asia-Pacific Bioinformatics Conference, {APBC} 2008,
  14-17 January 2008, Kyoto, Japan}, volume~6 of {\em Advances in
  Bioinformatics and Computational Biology}, pages 211--220. Imperial College
  Press, 2008.
\newblock URL:
  \url{http://www.comp.nus.edu.sg/\%7Ewongls/psZ/apbc2008/apbc050a.pdf}.

\bibitem{DBLP:journals/algorithmica/BockerBK11}
Sebastian B{\"{o}}cker, Sebastian Briesemeister, and Gunnar~W. Klau.
\newblock Exact algorithms for cluster editing: Evaluation and experiments.
\newblock {\em Algorithmica}, 60(2):316--334, 2011.
\newblock \href {https://doi.org/10.1007/s00453-009-9339-7}
  {\path{doi:10.1007/s00453-009-9339-7}}.

\bibitem{bradley2000constrained}
Paul~S. Bradley, Kristin~P. Bennett, and Ayhan Demiriz.
\newblock Constrained {K}-means clustering.
\newblock {\em Microsoft Research, Redmond}, 20(0):0, 2000.

\bibitem{DBLP:conf/iwpec/CaoC10}
Yixin Cao and Jianer Chen.
\newblock Cluster editing: Kernelization based on edge cuts.
\newblock In Venkatesh Raman and Saket Saurabh, editors, {\em Parameterized and
  Exact Computation - 5th International Symposium, {IPEC} 2010, Chennai, India,
  December 13-15, 2010. Proceedings}, volume 6478 of {\em Lecture Notes in
  Computer Science}, pages 60--71. Springer, 2010.
\newblock \href {https://doi.org/10.1007/978-3-642-17493-3\_8}
  {\path{doi:10.1007/978-3-642-17493-3\_8}}.

\bibitem{DBLP:journals/jcss/CharikarGW05}
Moses Charikar, Venkatesan Guruswami, and Anthony Wirth.
\newblock Clustering with qualitative information.
\newblock {\em J. Comput. Syst. Sci.}, 71(3):360--383, 2005.
\newblock \href {https://doi.org/10.1016/j.jcss.2004.10.012}
  {\path{doi:10.1016/j.jcss.2004.10.012}}.

\bibitem{DBLP:journals/jcss/ChenM12}
Jianer Chen and Jie Meng.
\newblock A 2k kernel for the cluster editing problem.
\newblock {\em J. Comput. Syst. Sci.}, 78(1):211--220, 2012.
\newblock \href {https://doi.org/10.1016/j.jcss.2011.04.001}
  {\path{doi:10.1016/j.jcss.2011.04.001}}.

\bibitem{DBLP:journals/csr/CrespelleDFG23}
Christophe Crespelle, P{\aa}l~Gr{\o}n{\aa}s Drange, Fedor~V. Fomin, and Petr~A.
  Golovach.
\newblock A survey of parameterized algorithms and the complexity of edge
  modification.
\newblock {\em Comput. Sci. Rev.}, 48:100556, 2023.
\newblock \href {https://doi.org/10.1016/j.cosrev.2023.100556}
  {\path{doi:10.1016/j.cosrev.2023.100556}}.

\bibitem{DBLP:books/sp/CyganFKLMPPS15}
Marek Cygan, Fedor~V. Fomin, Lukasz Kowalik, Daniel Lokshtanov, D{\'{a}}niel
  Marx, Marcin Pilipczuk, Michal Pilipczuk, and Saket Saurabh.
\newblock {\em Parameterized Algorithms}.
\newblock Springer, 2015.

\bibitem{DBLP:journals/tcs/CyganP10}
Marek Cygan and Marcin Pilipczuk.
\newblock Exact and approximate bandwidth.
\newblock {\em Theor. Comput. Sci.}, 411(40-42):3701--3713, 2010.
\newblock URL: \url{https://doi.org/10.1016/j.tcs.2010.06.018}, \href
  {https://doi.org/10.1016/J.TCS.2010.06.018}
  {\path{doi:10.1016/J.TCS.2010.06.018}}.

\bibitem{DBLP:journals/eaai/EzugwuIOAAEA22}
Absalom~E. Ezugwu, Abiodun~M. Ikotun, Olaide~Nathaniel Oyelade, Laith~Mohammad
  Abualigah, Jeffrey~O. Agushaka, Christopher~I. Eke, and Andronicus~Ayobami
  Akinyelu.
\newblock A comprehensive survey of clustering algorithms: State-of-the-art
  machine learning applications, taxonomy, challenges, and future research
  prospects.
\newblock {\em Eng. Appl. Artif. Intell.}, 110:104743, 2022.
\newblock \href {https://doi.org/10.1016/j.engappai.2022.104743}
  {\path{doi:10.1016/j.engappai.2022.104743}}.

\bibitem{DBLP:conf/iwpec/Fellows06}
Michael~R. Fellows.
\newblock The lost continent of polynomial time: Preprocessing and
  kernelization.
\newblock In Hans~L. Bodlaender and Michael~A. Langston, editors, {\em
  Parameterized and Exact Computation, Second International Workshop, {IWPEC}
  2006, Z{\"{u}}rich, Switzerland, September 13-15, 2006, Proceedings}, volume
  4169 of {\em Lecture Notes in Computer Science}, pages 276--277. Springer,
  2006.
\newblock \href {https://doi.org/10.1007/11847250\_25}
  {\path{doi:10.1007/11847250\_25}}.

\bibitem{DBLP:conf/fct/FellowsLRS07}
Michael~R. Fellows, Michael~A. Langston, Frances~A. Rosamond, and Peter Shaw.
\newblock Efficient parameterized preprocessing for cluster editing.
\newblock In Erzs{\'{e}}bet Csuhaj{-}Varj{\'{u}} and Zolt{\'{a}}n {\'{E}}sik,
  editors, {\em Fundamentals of Computation Theory, 16th International
  Symposium, {FCT} 2007, Budapest, Hungary, August 27-30, 2007, Proceedings},
  volume 4639 of {\em Lecture Notes in Computer Science}, pages 312--321.
  Springer, 2007.
\newblock \href {https://doi.org/10.1007/978-3-540-74240-1\_27}
  {\path{doi:10.1007/978-3-540-74240-1\_27}}.

\bibitem{DBLP:journals/jcss/FominGP23}
Fedor~V. Fomin, Petr~A. Golovach, and Nidhi Purohit.
\newblock Parameterized complexity of categorical clustering with size
  constraints.
\newblock {\em J. Comput. Syst. Sci.}, 136:171--194, 2023.
\newblock \href {https://doi.org/10.1016/j.jcss.2023.03.006}
  {\path{doi:10.1016/j.jcss.2023.03.006}}.

\bibitem{DBLP:journals/jcss/FominKPPV14}
Fedor~V. Fomin, Stefan Kratsch, Marcin Pilipczuk, Michal Pilipczuk, and Yngve
  Villanger.
\newblock Tight bounds for parameterized complexity of cluster editing with a
  small number of clusters.
\newblock {\em J. Comput. Syst. Sci.}, 80(7):1430--1447, 2014.
\newblock \href {https://doi.org/10.1016/j.jcss.2014.04.015}
  {\path{doi:10.1016/j.jcss.2014.04.015}}.

\bibitem{DBLP:journals/siamcomp/FoxRSWW20}
Jacob Fox, Tim Roughgarden, C.~Seshadhri, Fan Wei, and Nicole Wein.
\newblock Finding cliques in social networks: {A} new distribution-free model.
\newblock {\em {SIAM} J. Comput.}, 49(2):448--464, 2020.
\newblock \href {https://doi.org/10.1137/18M1210459}
  {\path{doi:10.1137/18M1210459}}.

\bibitem{DBLP:conf/aaai/FroeseKN22}
Vincent Froese, Leon Kellerhals, and Rolf Niedermeier.
\newblock Modification-fair cluster editing.
\newblock In {\em Thirty-Sixth {AAAI} Conference on Artificial Intelligence,
  {AAAI} 2022, Thirty-Fourth Conference on Innovative Applications of
  Artificial Intelligence, {IAAI} 2022, The Twelveth Symposium on Educational
  Advances in Artificial Intelligence, {EAAI} 2022 Virtual Event, February 22 -
  March 1, 2022}, pages 6631--6638. {AAAI} Press, 2022.
\newblock \href {https://doi.org/10.1609/aaai.v36i6.20617}
  {\path{doi:10.1609/aaai.v36i6.20617}}.

\bibitem{gambron2020comparison}
Philippe Gambron and Sue Thorne.
\newblock Comparison of several fft libraries in c/c++.
\newblock Technical report, STFC, 2020.
\newblock URL: \url{https://epubs.stfc.ac.uk/work/45434573}, \href
  {https://doi.org/10.5286/raltr.2020003} {\path{doi:10.5286/raltr.2020003}}.

\bibitem{DBLP:journals/siamcomp/GareyJ75}
M.~R. Garey and David~S. Johnson.
\newblock Complexity results for multiprocessor scheduling under resource
  constraints.
\newblock {\em {SIAM} J. Comput.}, 4(4):397--411, 1975.
\newblock \href {https://doi.org/10.1137/0204035} {\path{doi:10.1137/0204035}}.

\bibitem{DBLP:journals/algorithmica/GrammGHN04}
Jens Gramm, Jiong Guo, Falk H{\"{u}}ffner, and Rolf Niedermeier.
\newblock Automated generation of search tree algorithms for hard graph
  modification problems.
\newblock {\em Algorithmica}, 39(4):321--347, 2004.
\newblock \href {https://doi.org/10.1007/s00453-004-1090-5}
  {\path{doi:10.1007/s00453-004-1090-5}}.

\bibitem{DBLP:journals/mst/GrammGHN05}
Jens Gramm, Jiong Guo, Falk H{\"{u}}ffner, and Rolf Niedermeier.
\newblock Graph-modeled data clustering: Exact algorithms for clique
  generation.
\newblock {\em Theory Comput. Syst.}, 38(4):373--392, 2005.
\newblock \href {https://doi.org/10.1007/s00224-004-1178-y}
  {\path{doi:10.1007/s00224-004-1178-y}}.

\bibitem{DBLP:journals/tcs/Guo09}
Jiong Guo.
\newblock A more effective linear kernelization for cluster editing.
\newblock {\em Theor. Comput. Sci.}, 410(8-10):718--726, 2009.
\newblock \href {https://doi.org/10.1016/j.tcs.2008.10.021}
  {\path{doi:10.1016/j.tcs.2008.10.021}}.

\bibitem{DBLP:conf/tamc/GuoHKZ08}
Jiong Guo, Falk H{\"{u}}ffner, Christian Komusiewicz, and Yong Zhang.
\newblock Improved algorithms for bicluster editing.
\newblock In Manindra Agrawal, Ding{-}Zhu Du, Zhenhua Duan, and Angsheng Li,
  editors, {\em Theory and Applications of Models of Computation, 5th
  International Conference, {TAMC} 2008, Xi'an, China, April 25-29, 2008.
  Proceedings}, volume 4978 of {\em Lecture Notes in Computer Science}, pages
  445--456. Springer, 2008.
\newblock \href {https://doi.org/10.1007/978-3-540-79228-4\_39}
  {\path{doi:10.1007/978-3-540-79228-4\_39}}.

\bibitem{DBLP:journals/siamdm/GuoKNU10}
Jiong Guo, Christian Komusiewicz, Rolf Niedermeier, and Johannes Uhlmann.
\newblock A more relaxed model for graph-based data clustering: s-plex cluster
  editing.
\newblock {\em {SIAM} J. Discret. Math.}, 24(4):1662--1683, 2010.
\newblock \href {https://doi.org/10.1137/090767285}
  {\path{doi:10.1137/090767285}}.

\bibitem{DBLP:conf/ijcai/Gupta00T21}
Sushmita Gupta, Pallavi Jain, Saket Saurabh, and Nimrod Talmon.
\newblock Even more effort towards improved bounds and fixed-parameter
  tractability for multiwinner rules.
\newblock In Zhi{-}Hua Zhou, editor, {\em Proceedings of the Thirtieth
  International Joint Conference on Artificial Intelligence, {IJCAI} 2021,
  Virtual Event / Montreal, Canada, 19-27 August 2021}, pages 217--223.
  ijcai.org, 2021.
\newblock URL: \url{https://doi.org/10.24963/ijcai.2021/31}, \href
  {https://doi.org/10.24963/IJCAI.2021/31} {\path{doi:10.24963/IJCAI.2021/31}}.

\bibitem{DBLP:journals/dam/GutinY23}
Gregory~Z. Gutin and Anders Yeo.
\newblock (1,1)-cluster editing is polynomial-time solvable.
\newblock {\em Discret. Appl. Math.}, 340:259--271, 2023.
\newblock URL: \url{https://doi.org/10.1016/j.dam.2023.07.002}, \href
  {https://doi.org/10.1016/J.DAM.2023.07.002}
  {\path{doi:10.1016/J.DAM.2023.07.002}}.

\bibitem{DBLP:journals/tcad/HagenK92}
Lars~W. Hagen and Andrew~B. Kahng.
\newblock New spectral methods for ratio cut partitioning and clustering.
\newblock {\em {IEEE} Trans. Comput. Aided Des. Integr. Circuits Syst.},
  11(9):1074--1085, 1992.
\newblock \href {https://doi.org/10.1109/43.159993}
  {\path{doi:10.1109/43.159993}}.

\bibitem{hardy1918asymptotic}
Godfrey~H Hardy and Srinivasa Ramanujan.
\newblock Asymptotic formula{\ae} in combinatory analysis.
\newblock {\em Proceedings of the London Mathematical Society}, 2(1):75--115,
  1918.

\bibitem{DBLP:conf/iwpec/KellerhalsKNZ21}
Leon Kellerhals, Tomohiro Koana, Andr{\'{e}} Nichterlein, and Philipp Zschoche.
\newblock The {PACE} 2021 parameterized algorithms and computational
  experiments challenge: Cluster editing.
\newblock In Petr~A. Golovach and Meirav Zehavi, editors, {\em 16th
  International Symposium on Parameterized and Exact Computation, {IPEC} 2021,
  September 8-10, 2021, Lisbon, Portugal}, volume 214 of {\em LIPIcs}, pages
  26:1--26:18. Schloss Dagstuhl - Leibniz-Zentrum f{\"{u}}r Informatik, 2021.
\newblock \href {https://doi.org/10.4230/LIPIcs.IPEC.2021.26}
  {\path{doi:10.4230/LIPIcs.IPEC.2021.26}}.

\bibitem{knuth2014art}
Donald~E. Knuth.
\newblock {\em The Art of Computer Programming, Volume 4A: Combinatorial
  Algorithms, Part 1}.
\newblock Pearson Education, 2014.
\newblock URL: \url{https://books.google.co.uk/books?id=IkuEBAAAQBAJ}.

\bibitem{DBLP:journals/siamdm/KoanaKS22}
Tomohiro Koana, Christian Komusiewicz, and Frank Sommer.
\newblock Exploiting {\textdollar}c{\textdollar}-closure in kernelization
  algorithms for graph problems.
\newblock {\em {SIAM} J. Discret. Math.}, 36(4):2798--2821, 2022.
\newblock \href {https://doi.org/10.1137/21m1449476}
  {\path{doi:10.1137/21m1449476}}.

\bibitem{doi:10.1080/03610926.2014.894070}
Jukka Kohonen and Jukka Corander.
\newblock Computing exact clustering posteriors with subset convolution.
\newblock {\em Communications in Statistics - Theory and Methods},
  45(10):3048--3058, 2016.
\newblock \href
  {http://arxiv.org/abs/https://doi.org/10.1080/03610926.2014.894070}
  {\path{arXiv:https://doi.org/10.1080/03610926.2014.894070}}, \href
  {https://doi.org/10.1080/03610926.2014.894070}
  {\path{doi:10.1080/03610926.2014.894070}}.

\bibitem{DBLP:journals/dam/KomusiewiczU12}
Christian Komusiewicz and Johannes Uhlmann.
\newblock Cluster editing with locally bounded modifications.
\newblock {\em Discret. Appl. Math.}, 160(15):2259--2270, 2012.
\newblock \href {https://doi.org/10.1016/j.dam.2012.05.019}
  {\path{doi:10.1016/j.dam.2012.05.019}}.

\bibitem{DBLP:reference/opt/KundakciogluA09}
O.~Erhun Kundakcioglu and Saed Alizamir.
\newblock Generalized assignment problem.
\newblock In Christodoulos~A. Floudas and Panos~M. Pardalos, editors, {\em
  Encyclopedia of Optimization, Second Edition}, pages 1153--1162. Springer,
  2009.
\newblock \href {https://doi.org/10.1007/978-0-387-74759-0\_200}
  {\path{doi:10.1007/978-0-387-74759-0\_200}}.

\bibitem{liao2012load}
Ying Liao, Huan Qi, and Weiqun Li.
\newblock Load-balanced clustering algorithm with distributed self-organization
  for wireless sensor networks.
\newblock {\em IEEE sensors journal}, 13(5):1498--1506, 2012.

\bibitem{DBLP:conf/ijcai/LinHX19}
Weibo Lin, Zhu He, and Mingyu Xiao.
\newblock Balanced clustering: {A} uniform model and fast algorithm.
\newblock In Sarit Kraus, editor, {\em Proceedings of the Twenty-Eighth
  International Joint Conference on Artificial Intelligence, {IJCAI} 2019,
  Macao, China, August 10-16, 2019}, pages 2987--2993. ijcai.org, 2019.
\newblock \href {https://doi.org/10.24963/ijcai.2019/414}
  {\path{doi:10.24963/ijcai.2019/414}}.

\bibitem{DBLP:journals/iandc/LokshtanovM13}
Daniel Lokshtanov and D{\'{a}}niel Marx.
\newblock Clustering with local restrictions.
\newblock {\em Inf. Comput.}, 222:278--292, 2013.
\newblock \href {https://doi.org/10.1016/j.ic.2012.10.016}
  {\path{doi:10.1016/j.ic.2012.10.016}}.

\bibitem{DBLP:conf/sspr/MalinenF14}
Mikko~I. Malinen and Pasi Fr{\"{a}}nti.
\newblock Balanced k-means for clustering.
\newblock In Pasi Fr{\"{a}}nti, Gavin Brown, Marco Loog, Francisco Escolano,
  and Marcello Pelillo, editors, {\em Structural, Syntactic, and Statistical
  Pattern Recognition - Joint {IAPR} International Workshop, {S+SSPR} 2014,
  Joensuu, Finland, August 20-22, 2014. Proceedings}, volume 8621 of {\em
  Lecture Notes in Computer Science}, pages 32--41. Springer, 2014.
\newblock \href {https://doi.org/10.1007/978-3-662-44415-3\_4}
  {\path{doi:10.1007/978-3-662-44415-3\_4}}.

\bibitem{DBLP:conf/issac/Moenck76}
Robert~T. Moenck.
\newblock Practical fast polynomial multiplication.
\newblock In Richard~D. Jenks, editor, {\em Proceedings of the third {ACM}
  Symposium on Symbolic and Algebraic Manipulation, {SYMSAC} 1976, Yorktown
  Heights, New York, USA, August 10-12, 1976}, pages 136--148. {ACM}, 1976.
\newblock \href {https://doi.org/10.1145/800205.806332}
  {\path{doi:10.1145/800205.806332}}.

\bibitem{morris2011clustermaker}
John~H Morris, Leonard Apeltsin, Aaron~M Newman, Jan Baumbach, Tobias Wittkop,
  Gang Su, Gary~D Bader, and Thomas~E Ferrin.
\newblock clustermaker: a multi-algorithm clustering plugin for cytoscape.
\newblock {\em BMC bioinformatics}, 12(1):1--14, 2011.

\bibitem{DBLP:journals/mst/ProttiSS09}
F{\'{a}}bio Protti, Maise~Dantas da~Silva, and Jayme~Luiz Szwarcfiter.
\newblock Applying modular decomposition to parameterized cluster editing
  problems.
\newblock {\em Theory Comput. Syst.}, 44(1):91--104, 2009.
\newblock \href {https://doi.org/10.1007/s00224-007-9032-7}
  {\path{doi:10.1007/s00224-007-9032-7}}.

\bibitem{rahmann2007exact}
Sven Rahmann, Tobias Wittkop, Jan Baumbach, Marcel Martin, Anke Truss, and
  Sebastian B{\"o}cker.
\newblock Exact and heuristic algorithms for weighted cluster editing.
\newblock In {\em Computational Systems Bioinformatics: (Volume 6)}, pages
  391--401. World Scientific, 2007.

\bibitem{DBLP:journals/dam/ShamirST04}
Ron Shamir, Roded Sharan, and Dekel Tsur.
\newblock Cluster graph modification problems.
\newblock {\em Discret. Appl. Math.}, 144(1-2):173--182, 2004.
\newblock \href {https://doi.org/10.1016/j.dam.2004.01.007}
  {\path{doi:10.1016/j.dam.2004.01.007}}.

\bibitem{shang2010energy}
Fengjun Shang and Yang Lei.
\newblock An energy-balanced clustering routing algorithm for wireless sensor
  network.
\newblock {\em Wireless Sensor Network}, 2(10):777, 2010.

\bibitem{steinvik2020kernelization}
Andreas Steinvik.
\newblock Kernelization for balanced graph clustering.
\newblock Master's thesis, The University of Bergen, 2020.
\newblock URL: \url{https://hdl.handle.net/1956/24115}.

\bibitem{DBLP:journals/mst/Tsur21}
Dekel Tsur.
\newblock Faster parameterized algorithm for cluster vertex deletion.
\newblock {\em Theory Comput. Syst.}, 65(2):323--343, 2021.
\newblock URL: \url{https://doi.org/10.1007/s00224-020-10005-w}, \href
  {https://doi.org/10.1007/S00224-020-10005-W}
  {\path{doi:10.1007/S00224-020-10005-W}}.

\bibitem{DBLP:books/daglib/0030297}
David~P. Williamson and David~B. Shmoys.
\newblock {\em The Design of Approximation Algorithms}.
\newblock Cambridge University Press, 2011.
\newblock URL:
  \url{http://www.cambridge.org/de/knowledge/isbn/item5759340/?site\_locale=de\_DE}.

\bibitem{wittkop2007large}
Tobias Wittkop, Jan Baumbach, Francisco~P Lobo, and Sven Rahmann.
\newblock Large scale clustering of protein sequences with force-a layout based
  heuristic for weighted cluster editing.
\newblock {\em BMC bioinformatics}, 8:1--12, 2007.
\newblock \href {https://doi.org/10.1186/1471-2105-8-396}
  {\path{doi:10.1186/1471-2105-8-396}}.

\bibitem{wittkop2010partitioning}
Tobias Wittkop, Dorothea Emig, Sita Lange, Sven Rahmann, Mario Albrecht, John~H
  Morris, Sebastian B{\"o}cker, Jens Stoye, and Jan Baumbach.
\newblock Partitioning biological data with transitivity clustering.
\newblock {\em Nature methods}, 7(6):419--420, 2010.

\bibitem{DBLP:journals/ipl/XiaoK22}
Mingyu Xiao and Shaowei Kou.
\newblock A simple and improved parameterized algorithm for bicluster editing.
\newblock {\em Inf. Process. Lett.}, 174:106193, 2022.
\newblock \href {https://doi.org/10.1016/j.ipl.2021.106193}
  {\path{doi:10.1016/j.ipl.2021.106193}}.

\bibitem{xu2015comprehensive}
Dongkuan Xu and Yingjie Tian.
\newblock A comprehensive survey of clustering algorithms.
\newblock {\em Annals of Data Science}, 2:165--193, 2015.

\end{thebibliography}
